\documentclass[
  prb,
  reprint,
]{revtex4-2}

\usepackage{amssymb}
\usepackage{amsmath}
\usepackage{amsthm}
\usepackage{graphicx}
\usepackage{dcolumn}
\usepackage{hyperref}
\usepackage{tikz}
\usepackage{dsfont}
\usetikzlibrary{calc}
\usepackage{physics}
\usepackage[caption=false]{subfig}
\usepackage[dvipsnames]{xcolor}
\usepackage[normalem]{ulem}

\begin{document}

\title{Quantum contextuality with mixed states of 1D symmetry-protected topological order}
\author{Leroy Fagan}
\email{lfagan23@unm.edu}

\author{Akimasa Miyake}
\email{amiyake@unm.edu}

\affiliation{
Department of Physics and Astronomy,
Center for Quantum Information and Control, 
Quantum New Mexico Institute,
University of New Mexico, Albuquerque, NM 87106, USA
}

\date{\today}
\begin{abstract}

Bell theorems of many-body nonlocality and contextuality serve as a benchmark for proving quantum advantage in that a quantum computer outperforms a classical computer for a certain problem. In practice, however, near-term quantum devices do not prepare perfectly pure states but rather mixed states produced from noisy channels. We investigate noisy quantum advantage by considering thermal mixed states of one-dimensional many-body systems with a symmetry-protected topological (SPT) order. In the pure-state (or zero-temperature) case, these states are known to be useful for measurement-based quantum computation, and to outperform classical computers in a many-body contextuality game, provided string order parameters (SOPs) of SPT are sufficiently large. Here, we show that quantum advantage in mixed states is measured by a combination of twisted SOP and symmetry representation expectation values. Using the minimally entangled typical thermal states algorithm, it is demonstrated that quantum advantage persists to a nonzero critical temperature for finite-sized instances of the many-body contextuality game. While this critical temperature goes to zero in the thermodynamic limit, it is relatively robust to system size, suggesting that these states remain useful for demonstrating genuine ``quantumness'' of noisy hardware in a scalable fashion. Finally, we show that the quantum winning probability is lower bounded by the global fidelity with the 1D cluster state, so that our contextuality game can provide an operational meaning to benchmark the capacity to create long-range order like SPT states in near-term experimental devices.

\end{abstract}

\maketitle
\section{Introduction} \label{sec:introduction}

The current age of quantum computing is marked by pre-fault-tolerant devices which are not able to perfectly perform computational tasks. While in most situations the goal of an ideal quantum computer would be to prepare and manipulate a pure quantum state $\ket{\psi}$, these noisy intermediate-scale quantum (NISQ) devices \cite{Preskill2018quantumcomputingin,RevModPhys.94.015004} are influenced by interactions with the environment, creating instead a classical probabilistic mixture of pure states (or a \emph{mixed} state $\rho$). The power of NISQ devices is thus closely linked to the properties of mixed quantum states resulting from physically realistic noise models on useful pure states. One of the most important questions in quantum computation concerns quantum advantage, which asks if a quantum computer can solve a problem more efficiently than a comparable classical computer and how this behavior scales with the problem size. While there are a number of results establishing that fault-tolerant quantum algorithms outperform the best known classical algorithms (e.g. \cite{Grover,Shor} are some of the most famous), less is known about whether near-term devices could demonstrate such a separation. One might suspect that the lack of complete quantum coherence would limit noisy devices' ability to exhibit quantum advantage. Despite this general view, however, in this work we demonstrate that computationally useful properties of a class of many-body entangled states would persist at the level of size and noise strength of modern NISQ hardware. In particular, finite-sized mixed states resulting from statistical thermal fluctuations are able to surpass classical bounds in Bell-type nonlocal games, up to a certain critical temperature.

 We adopt the framework of measurement-based quantum computing (MBQC) \cite{RaussendorfBriegel,Raussendorf2003}, because it matches the canonical setting of Bell nonlocal games described below. In particular, MBQC simulates quantum computation by performing adaptive single-qubit measurements on a quantum resource state, consuming its entanglement in the process. A particular class of short-ranged entangled pure states known as \emph{symmetry-protected topological} (SPT) states have been shown to be resources for MBQC \cite{Miyake,QCinHaldanePhase,SPTforMBQC,SPTinMBQCGS,Miller,Miller2016,1DSPTforComputation,SPTComputationalPower,FractalMBQC,ComputationalPhaseofMatter,Stephen2019subsystem,Daniel2020computational}. These states are said to exhibit SPT order (SPTO) and belong to respective SPT phases \cite{PollmannEntanglement,PollmannSPTO,Homology,QImeetsQM}. The ability to perform MBQC has been shown to be a property of the entire SPT phase rather than just individual states. MBQC procedures therefore have a certain robustness, as the resource state does not have to be known exactly as long as it resides in the right phase.

MBQC resource quality may be measured through the degree to which it is able to exhibit quantum advantage in \emph{nonlocal games} \cite{NonlocalgamesBook,NonlocalgamesPaper}, distributed computational tasks subject to global linear constraints designed to thwart explanation by local classical circuits. This is a modern extension of the famous Bell inequalities which demonstrate that quantum correlations cannot be reproduced by any local hidden variable theory \cite{Bell}. Translated to nonlocal games, this means that by sharing certain entangled quantum resource states and performing local measurements, players have the potential to outperform any local classical strategy. States which outperform classical strategies in these games have been shown to be useful MBQC resources \cite{nonlocalgamesMBQC}, due to a quantum-mechanical property known as \emph{contextuality} \cite{KS1967,Spekkens2005} which arises from its linear algebraic structure. The contextuality of many-body states is studied in Refs. \cite{Bulchandani3,Bulchandani2,Hart,Bulchandani,ToricCodeNonlocalGame,Braidingforthewin}.

We study a previously formulated nonlocal game known as the ``contextual triangle game'' \cite{Austin} which is based on previous work \cite{modelingPaulimeas,ShallowCircuits,Bravyi2020} in which pure states with SPTO enable quantum advantage. The necessary condition for symmetric states to exhibit quantum advantage in the game is a sufficiently high expectation value of a \emph{string order parameter} (SOP) which acts as a detector of the SPT phase \cite{denNijs,SOPinSpinLattice}.

Here we consider a simple extension to the contextual triangle game formulation which allows one to consider nonsymmetric states and more generally, {\em mixed} states. In this case, the quantum advantage cannot be directly traced back to the string order parameter but instead relies on a certain combination of symmetry operators and a ``twisted string order parameter'' which probes a topological invariant of the SPT phase. Using this framework, we study the performance of mixed states resulting from noisy preparation of SPT-ordered pure states in the triangle game, giving a notion of how quantum advantage is affected by noisy hardware. Our work is thus closely related to the emerging study of mixed state phases of matter studied in Refs. \cite{Masui2025,MixedSPT,Lessa2025,Ma2025}.

We prove that the global fidelity of the resource state with the cluster state \cite{Cluster}, the canonical example of SPTO in our SPT phase, provides a lower bound for the quantum winning probability in our nonlocal game, connecting our work more closely to experiment. In this picture, once the fidelity is sufficiently high, the state is verifiably quantum as it can be used in our nonlocal game to beat classical strategies. Framed in a different light, our contextuality game can be used to benchmark whether true long-range SPTO has been created on NISQ devices.

While decoherence in an open system depends on the physical details of the specific NISQ device, we highlight a noise model given by the Gibbs thermal mixed state of a system Hamiltonian with SPTO. This is because when the system is in contact with a thermal bath at a fixed temperature, the Gibbs state appears ubiquitously by minimizing the Gibbs free energy under standard assumptions of statistical mechanics, giving an ``average'' noise model to illustrate explicitly our test of (state-dependent) contextuality.
Using an efficient Monte-Carlo sampling technique known as the minimally entangled typical thermal states (METTS) algorithm \cite{White, Stoudenmire_2010} applied to matrix product states (MPS) \cite{White1992, Ostlund1995, Perez-Garcia, Verstraete2008,Schollwock2011} on ITensor \cite{itensor}, we obtain accurate expectation values of symmetry operators and twisted SOPs in the Gibbs state. Although the SPT phase we consider does not exist above zero temperature {\em in the thermodynamic limit} \cite{Roberts}, we find that finite-sized thermal states still contain entanglement properties necessary to exhibit quantum advantage in the triangle game up to some size-dependent finite temperature. Our result implies that demonstration of unconditional quantum advantage may be possible on NISQ devices.

The rest of the paper will be arranged as follows. In Section \ref{sec:SPTO} we describe symmetry-protected topological order. In particular, we focus on one-dimensional (1D) SPTO protected by a $\mathds{Z}_2\times\mathds{Z}_2$ symmetry in spin-$1/2$ systems. In Section \ref{sec:trianglegame} we review the contextual triangle game and subsequent quantum advantage using SPT resource states. We then extend the previous framework to arbitrary mixed states. Section \ref{sec:results} contains our main numerical results which show that finite-sized SPT states can still show quantum advantage under thermal noise. We explicitly analyze the thermal cluster state, giving an analytic solution of the temperature at which quantum advantage is lost. In Section \ref{sec:experiment} we prove that the winning probability is lower bounded by the global fidelity with the cluster state and comment on current experimental progress.
Finally, in Section \ref{sec:conclusion} we conclude and give outlook for remaining work and future directions.

\section{1D $\mathds{Z}_2\times\mathds{Z}_2$ symmetry-protected topological order} \label{sec:SPTO}

Phases of matter describe distinct emergent collective behavior in the study of many-body systems and appear ubiquitously throughout many fields of science. In classical systems, phases are often characterized by global symmetries and their spontaneous breaking; In a solid-liquid phase transition the discrete lattice symmetry of the solid breaks the continuous translational and rotational symmetries of the liquid, and a local observable that detects the number of bonds emerging from an atom is thus able to distinguish between these two phases. However, quantum many-body systems admit exotic phases of matter which do not arise from any broken symmetry, and such ``topological'' phases are characterized through global invariants rather than local observables. From a quantum information viewpoint, topological phases cannot be connected to a product state by any constant-depth quantum circuit with local interactions \cite{Chen}. It follows that local perturbations do not destroy the entanglement structure of topologically ordered states, making them natural candidates for quantum information storage and quantum error correcting codes \cite{toric,Kitaev}.

Reintroducing symmetry further enriches the phase diagram.  Some phases that are topologically trivial in the absence of symmetry become nontrivial once the symmetry constraint is imposed, in that states in that phase cannot be connected to a product state by a constant-depth local quantum circuit \emph{which respects the symmetry}. These \emph{symmetry-protected topological} (SPT) phases still remain robust against symmetry-preserving local perturbations and are characterized by nonlocal order parameters. States in an SPT phase are said to exhibit SPT order (SPTO) \cite{PollmannEntanglement,PollmannSPTO,Homology,QImeetsQM}.

SPTO is interesting in its own right as it has many uniquely quantum-mechanical properties. For one, it has been shown that certain two-dimensional (2D) SPT phases provide universal resource states for measurement-based quantum computing (MBQC) \cite{Miyake,QCinHaldanePhase,SPTforMBQC,SPTinMBQCGS,Miller,Miller2016,1DSPTforComputation,SPTComputationalPower,FractalMBQC,ComputationalPhaseofMatter,Stephen2019subsystem,Daniel2020computational}, meaning quantum computation can be simulated using only single-qubit adaptive measurements on the state. States that enable universal MBQC are also suspected to exhibit \emph{contextuality}, \cite{Raussendorf}, the intrinsically quantum mechanical property that measurement outcomes are dependent upon the set of measurements performed alongside it.

SPT states' properties make them ideal candidates for studying quantum advantage both through simulation and actual implementation on NISQ hardware. As mentioned,
SPT ground states are efficiently preparable for arbitrary system sizes. This follows from the fact that SPT states lie in the trivial phase once the protecting symmetry is removed, meaning that one can reach arbitrarily close to the states from some fiducial product state by a (non-symmetry-respecting) constant depth local quantum circuit. Additionally, while 2D SPT phases are of the primary interest, one can gain insights by studying their 1D building blocks. Gapped 1D SPT ground states are area-law entangled \cite{Hastings_2007}, allowing for efficient classical simulation using matrix product states, a fact we will take advantage of in our numerical analysis.

We first establish our notation. We work in a many-body system composed of $n$ qubits. The identity matrix is denoted as $\mathds{1}$ and the Pauli matrices as $X,Y,Z$, with subscripts denoting the site index. We denote the $Z$ ($X$) $+1$ eigenstate as $\ket{0}$ ($\ket{+}$) and the $-1$ eigenstate as  $\ket{1}$ ($\ket{-}$). Additionally the controlled phase gate between sites $j,k$ is denoted 
$CZ_{j,k}=\ket{0}\bra{0}_{j}\otimes\mathds{1}_k+\ket{1}\bra{1}_j\otimes Z_k$. The cyclic group of order 2, with elements $\{0,1\}$ under addition modulo 2, is denoted $\mathds{Z}_2$. 

Our work focuses specifically on the symmetry group $\mathds{Z}_2\times\mathds{Z}_2$ with elements $(0,0)$, $(0,1)$, $(1,0)$, and $(1,1)$ under element-wise addition modulo 2. We will sometimes refer to these group elements as $e,x,y,$ and $z$, respectively. We will express an arbitrary group element as $g\in\mathds{Z}_2\times\mathds{Z}_2$ or $(a,b)\in\mathds{Z}_2\times\mathds{Z}_2$ where $a,b\in\{0,1\}$.

\subsection{Review of general 1D SPTOs}
\label{sec:generalSPTO}

We briefly review a general construction of 
 a 1D many-body quantum system under a Hamiltonian $H$ which obeys a global symmetry $G$. On the Hilbert space, the symmetry takes the unitary representation given by $U(g)$ for $g\in G$, which by definition obeys $\comm{U(g)}{H}=0$ $\forall g\in G$. In this paper we consider only finite abelian groups with ``onsite'' representations, where the global symmetries may be decomposed into tensor products of $m$-local blocks which we label by $p$:
\begin{align}
    U(g)=\bigotimes_{p=1}^{n/m} u(g). \label{eq:onsite}
\end{align}

The possible 1D SPT phases for chains with onsite symmetries are classified by elements of the second cohomology group $H^2(G,U(1))$ which correspond to equivalence classes of projective representations of $G$ \cite{Homology}. A projective representation is a symmetry representation $V(g)$ which satisfies group multiplication up to a phase $w(g,h)\in U(1)$: for $g,h\in G$,
\begin{align}
V(g)V(h)=w(g,h)V(gh).
\label{eq:projectiverep}
\end{align}
Physically, this characterization describes how the chain's edge degrees of freedom transform under a bulk symmetry action. This can be seen by acting on the state with a truncated symmetry operator
\begin{align}
    U_{[p,q]}(g)=\bigotimes_{p+1}^{q-1}u(g).
\end{align}
 The operator acts trivially in the bulk but in general does not necessarily commute with $H$ near the boundaries $p+1$ and $q-1$. However, one can turn this operator into a Hamiltonian symmetry by attaching boundary operators $V^R(g), V^L(g)$ onto each edge which remove the corresponding excitations. These are individually exactly the projective representations of the symmetry described earlier. Note, however, that these boundary operators do not necessarily have the same support size $m$ as $u(g)$.
 It follows that
\begin{align}
S_{[p,q]}(g)=V_p^L(g)\otimes U_{[p,q]}(g) \otimes V_q^R(g)\label{eq:SOP}
\end{align}
is a symmetry of the Hamiltonian and acts as $S_{[p,q]}(g)\ket{\phi}=\ket{\phi}$ on a symmetric ground state $\ket{\phi}$ of $H$. This operator is known as a \emph{string order parameter} (SOP) \cite{denNijs,SOPinSpinLattice}.

The projective representations $V^L,V^R$ differ between states, however, their commutation properties are common throughout an SPT phase. Particularly, the \emph{twist phase} $\Omega(g,h)$ \cite{twistphase} is an SPT invariant defined by
\begin{align}
V(g)V(h)=\Omega(g,h)V(h)V(g).
\label{eq:twistphase}
\end{align}
Operationally, all of the $V$'s in Eq. \eqref{eq:twistphase} are chosen to be either $V^L$ or $V^R$.

In this work we consider only boundary operators constructed at the renormalization group (RG) fixed-point state of the SPT phase \cite{QuantumRG,QImeetsQM}, which is the ``paradigmatic'' state of the SPT phase, in that it displays most prominently the behavior of the phase. In this case the boundary operators have the same support size $m$ as $u(g)$ and are related by
 \begin{align}
     V^R(g)V^L(g)=u(g).
     \label{eq:V^RV^L}
 \end{align}
The twist phase in the RG fixed-point state can be probed by the \emph{twisted string order parameter}, created by twisting a string order parameter over a global symmetry operator:
\begin{align}
T_{[p,q]}^{(g,h)}=V_q^R(g)U(h)V_p^L(g)U_{[p,q]}(g).
\label{eq:twistedSOP}
\end{align}
Using Eqs.  \eqref{eq:twistphase} and \eqref{eq:V^RV^L} this is equivalent to $\Omega(g,h)U(h)S_{[p,q]}(g)$, which makes it clear that  $T_{[p,q]}^{(g,h)}\ket{\phi}=\Omega(g,h)\ket{\phi}$ on the fixed-point state \cite{Austin}.  Whenever we mention the SOP and twisted SOP in the rest of the text, we are referring respectively to Eq. \eqref{eq:SOP} and Eq.  \eqref{eq:twistedSOP} constructed with the fixed-point boundary operators.

The importance of the twisted SOP in our work is that we find it may be used to detect SPT properties in \emph{non-symmetric} states. When considered on symmetric pure states, such as SPTO ground states of a Hamiltonian, this metric reduces to a string order parameter (up to the twist phase). However, we are interested in SPT properties that exist in more general, potentially mixed states. We will show that these properties do indeed exist for some mixed states as measured (in part) by expectation values of the twisted SOP, despite the states not belonging strictly to the SPT phase. We elaborate on this connection in Section \ref{sec:trianglegame}.

\subsection{1D cluster SPT phase}
\label{sec:clusterSPTO}

\begin{figure*}[tb]
\includegraphics[width=\textwidth]{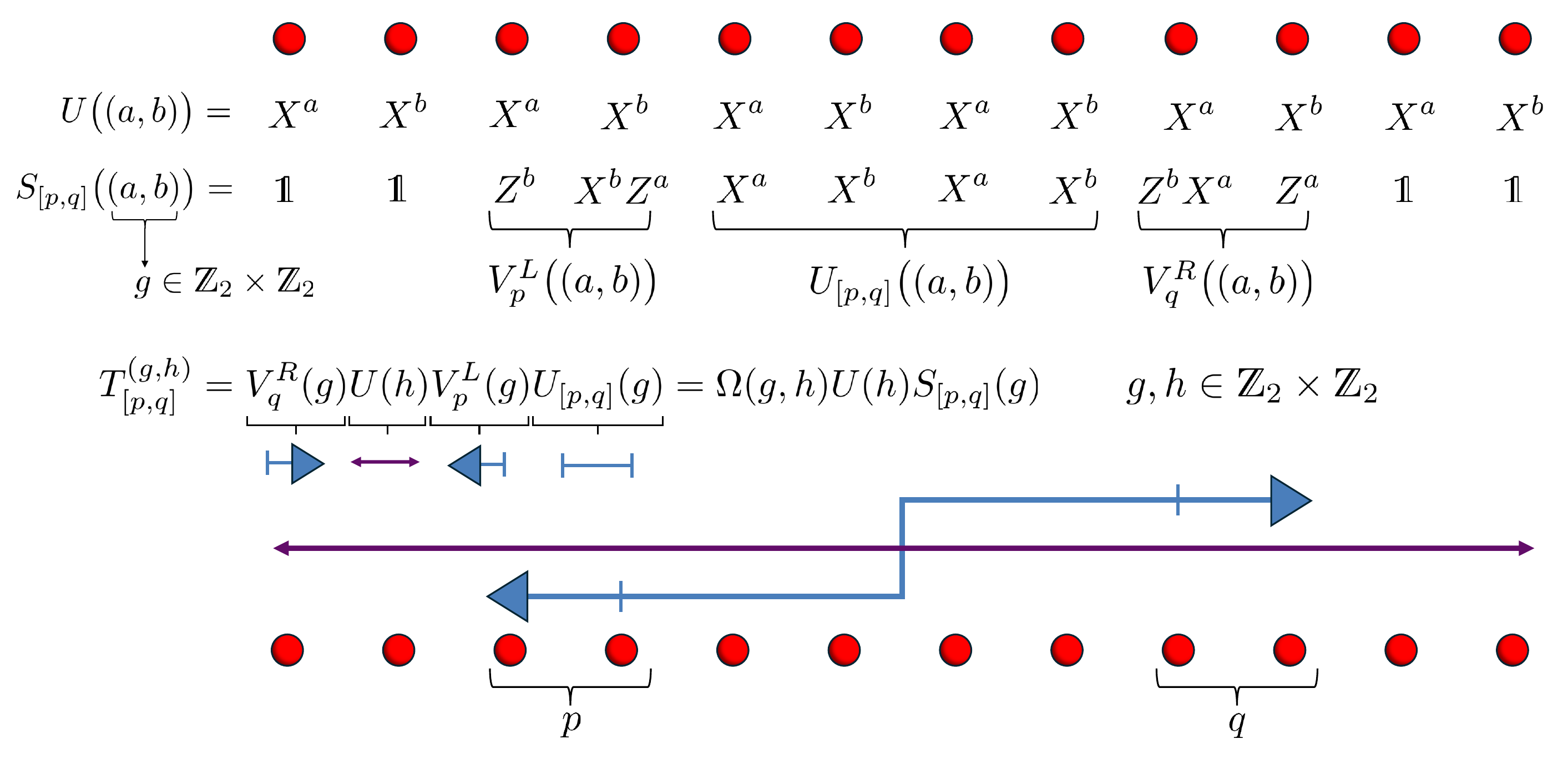}
\caption{Symmetry operators associated with $\mathds{Z}_2\times\mathds{Z}_2$ SPTO on a spin-$1/2$ chain with $n=12$ qubits.  The four $\mathds{Z}_2\times\mathds{Z}_2$ group elements are denoted compactly as $(a,b)$, where $a,b\in\{0,1\}$. The top qubit chain shows explicitly the $\mathds{Z}_2\times\mathds{Z}_2$ symmetry representations and cluster state string order parameters (SOPs). The left and right SOP boundary blocks of SOPs are labeled by $p$ and $q$; the figure illustrates the case of $p=2,q=5$.
The bottom qubit chain shows that the twisted string order parameter is constructed by ``twisting'' a symmetry operator and an SOP, and is equivalent to the product of the two up to the twist phase $\Omega(g,h)$ which measures the commutation properties of projective symmetry representations $V^{L}(g)$ or $V^{R}(g)$.}

\label{fig:SOP}
\end{figure*}

Now we focus on the specific case of a $\mathds{Z}_2\cross\mathds{Z}_2$ global symmetry in a 1D spin-1/2 system with an even number of sites $n$. In this case the representation $U(g)$ of the four $\mathds{Z}_2\times\mathds{Z}_2$ elements $(0,0),(0,1),(1,0),$ and $(1,1)$ can manifest as spin flip symmetry on no sites, all even sites, all odd sites, or all sites, respectively. For group elements $(a,b)\in\mathds{Z}_2\times\mathds{Z}_2$ the symmetry representation takes the form
\begin{align}
    U\big((a,b)\big)=\prod_{j\text{ odd}}X_j^a\prod_{k\text{ even}}X_k^b.
    \label{eq:fullsymmetry}
\end{align}
Following Eq. \eqref{eq:onsite}, the representation decomposes into an onsite operation on two adjacent sites (i.e. $m=2$):
\begin{align}
u\big((a,b)\big)=X^a\otimes X^b.
\label{eq:Z2xZ2symmetry}
\end{align}
From here on $j,k\in\{1,2,...,n\}$ will label individual qubits, whereas $p,q\in\{1,2,...,n/2\}$ will label blocks of two, for reasons that will become clear in Sec. \ref{sec:trianglegame}. The block $p$ then includes qubits $j=2p-1,2p$.

It can be shown that $H^2(\mathds{Z}_2\times\mathds{Z}_2,U(1))\cong\mathds{Z}_2$, meaning that this symmetry admits one trivial and one nontrivial SPT phase. To explore these phases we consider a $\mathds{Z}_2\times\mathds{Z}_2$ invariant Hamiltonian family on periodic boundary conditions described by the real parameters $(J_X,J_{ZZ})$:
\begin{align} 
H=-\frac{\Delta}{2}\sum_{j=1}^n\bigg(Z_{j-1}X_jZ_{j+1}+J_XX_j+J_{ZZ}Z_{j-1}Z_{j+1}\bigg),\label{eq:Hamiltonian}
\end{align}
where $\Delta$ is the energy gap at $(J_X,J_{ZZ})=(0,0)$.
The trivial SPT phase is realized in the large $J_X$ limit, where the model is a trivial paramagnet. In the large $J_{ZZ}$ limit, the model is a decoupled Ising ferromagnet on even and odd sublattices, spontaneously breaking the $\mathds{Z}_2\times\mathds{Z}_2$ symmetry.

Importantly, the nontrivial SPT phase is also witnessed: at $(J_X,J_{ZZ})=(0,0)$ the ground state of the corresponding Hamiltonian is the so-called 1D ``cluster state" \cite{Cluster}, which we denote as $\ket{C_n}$. This is a stabilizer state defined as the joint +1 eigenstate of the pairwise commuting set of Paulis $\{K_j=Z_{j-1}X_jZ_{j+1}|j\in1,2,\hdots,n\}$. The cluster state is a paradigmatic example of 1D $\mathds{Z}_2\times\mathds{Z}_2$ SPTO; it is the canonical resource for 1D MBQC (any single-qubit gate can be simulated by encoding a quantum state on boundaries of $\ket{C_n}$ carved out by $Z$ measurements and measuring adjacent qubits appropriately in the $X$-$Y$ plane) and is created by the constant depth quantum circuit

\begin{align}
\ket{C_n}=\prod_{j=1}^nCZ_{j,j+1}\ket{+}^{\otimes n}. \label{eq:cluster}
\end{align}
Since $\ket{C_n}$ is a stabilizer state it is straightforward to create its string order parameters. The boundary operators take the form
\begin{align} 
V^L\big((a,b)\big)=Z^b\otimes X^bZ^a,\label{eq:vL}\\
V^R\big((a,b)\big)=Z^bX^a\otimes Z^a.\label{eq:vR}
\end{align}
which is easily verified by constructing $S_{[p,q]}(g)$ from Eq. \eqref{eq:SOP} and noticing that it can be written as a product of stabilizers $K_{2p}$ through $K_{2q-1}$, where stabilizers on odd (even) sites are raised to the power of $a$ ($b$) according to the $\mathds{Z}_2\times\mathds{Z}_2$ group element $g=(a,b)$. The SOP thus satisfies the desired trivial action on the cluster state.

The boundary operators form a nontrivial projective representation of the symmetry group as they satisfy group multiplication up to a phase: for $(a,b),(c,d)\in\mathds{Z}_2\times\mathds{Z}_2$ and $E\in\{L,R\}$
\begin{align}
V^E\big((a,b)\big)V^E\big((c,d)\big)=(-1)^{ad}V^E\big((a+c,b+d)\big),
\end{align}
with $w\big((a,b),(c,d)\big)=(-1)^{ad}$ as in Eq. \eqref{eq:projectiverep}.
Furthermore, the twist phase associated with the SPT phase of the cluster state is
\begin{align}
    \Omega\big((a,b),(c,d)\big)=(-1)^{ad-bc}.
    \label{eq:clustertwist}
\end{align}
This is equivalently realized by $V\big((a,b)\big)=X^aZ^b$; the edge degrees of freedom in the cluster state thus transform under an effective Pauli algebra, realizing a symmetry-protected qubit.

Furthermore, the cluster state can be shown to be the RG fixed-point of the nontrivial SPT phase. One can then use Eq. \eqref{eq:twistedSOP} to create a twisted SOP $T_{[p,q]}^{(g,h)}$ whose expectation value in the cluster state is the nontrivial twist phase $\Omega(g,h)$ given by Eq. \eqref{eq:clustertwist}.
Expectation values of the twisted SOPs give a diagnostic of whether a state belongs to the nontrivial SPT phase. As mentioned previously, for pure symmetric states the action of the twisted SOP on the state is equivalent to the action of the SOP up to the twist phase. Much of the content of this paper will focus on extending previous results for these states to more general states by generalizing from SOPs to twisted SOPs. Figure \ref{fig:SOP} shows a visualization of the global symmetry operators, cluster state SOPs, and corresponding twisted SOPs which have been explained in this section.

Recall that the cluster state is the ground state of the Hamiltonian in Eq. \eqref{eq:Hamiltonian} with $(J_X,J_{ZZ})=(0,0)$. Any Hamiltonian that can be adiabatically connected to the cluster state Hamiltonian without breaking the symmetry has a ground state that exists in the same SPT phase as the cluster state. Indeed, Fig. \ref{fig:GSadvantage} shows that expectation values of cluster state string order parameters remain nonzero for many ground states of the Hamiltonian which are close to the cluster state in parameter space, implying that these states live in the nontrivial SPT phase.

\begin{figure}[t]
    \includegraphics[width=0.48\textwidth]{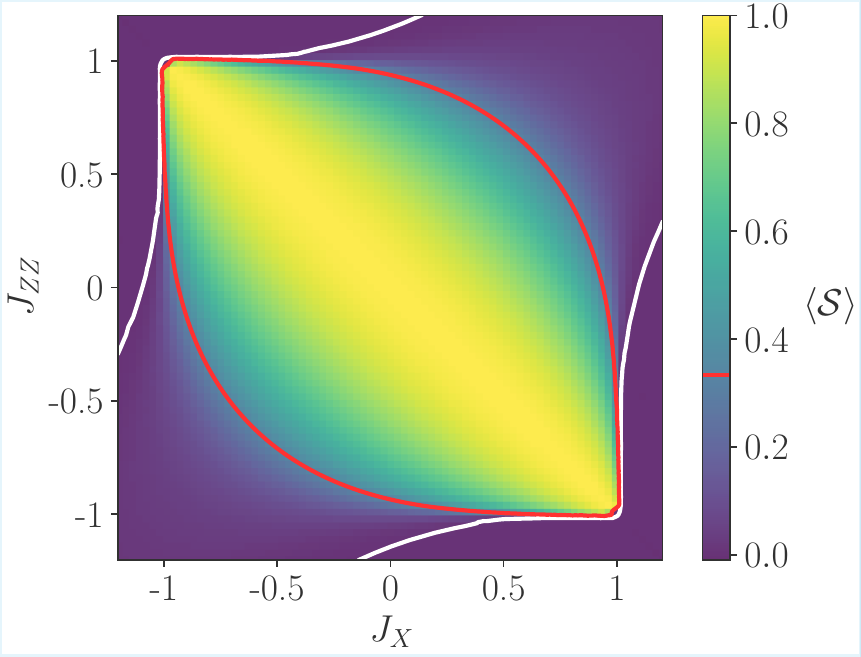}
    \caption{ Nontrivial SPT phase and quantum advantage of the ground states of our Hamiltonian Eq. \eqref{eq:Hamiltonian}. The boundary of the SPT phase is shown in white, witnessed when the expectation value of any one of the three ($g=x,y,$ or $z$) string order parameters, shown in Fig. \ref{fig:SOP}, is nonzero. The heatmap shows the minimum string order parameter expectation value $\ev{\mathcal{S}}$ given by Eq. \eqref{eq:classicalbound}. If this quantity is larger than $1/3$, shown by the red contour, the state can be used to show quantum advantage in the multiplayer contextual triangle game.}
    \label{fig:GSadvantage}
\end{figure}

\section{Contextual triangle game} \label{sec:trianglegame}

\begin{figure*}[t]
\includegraphics[width=\textwidth]{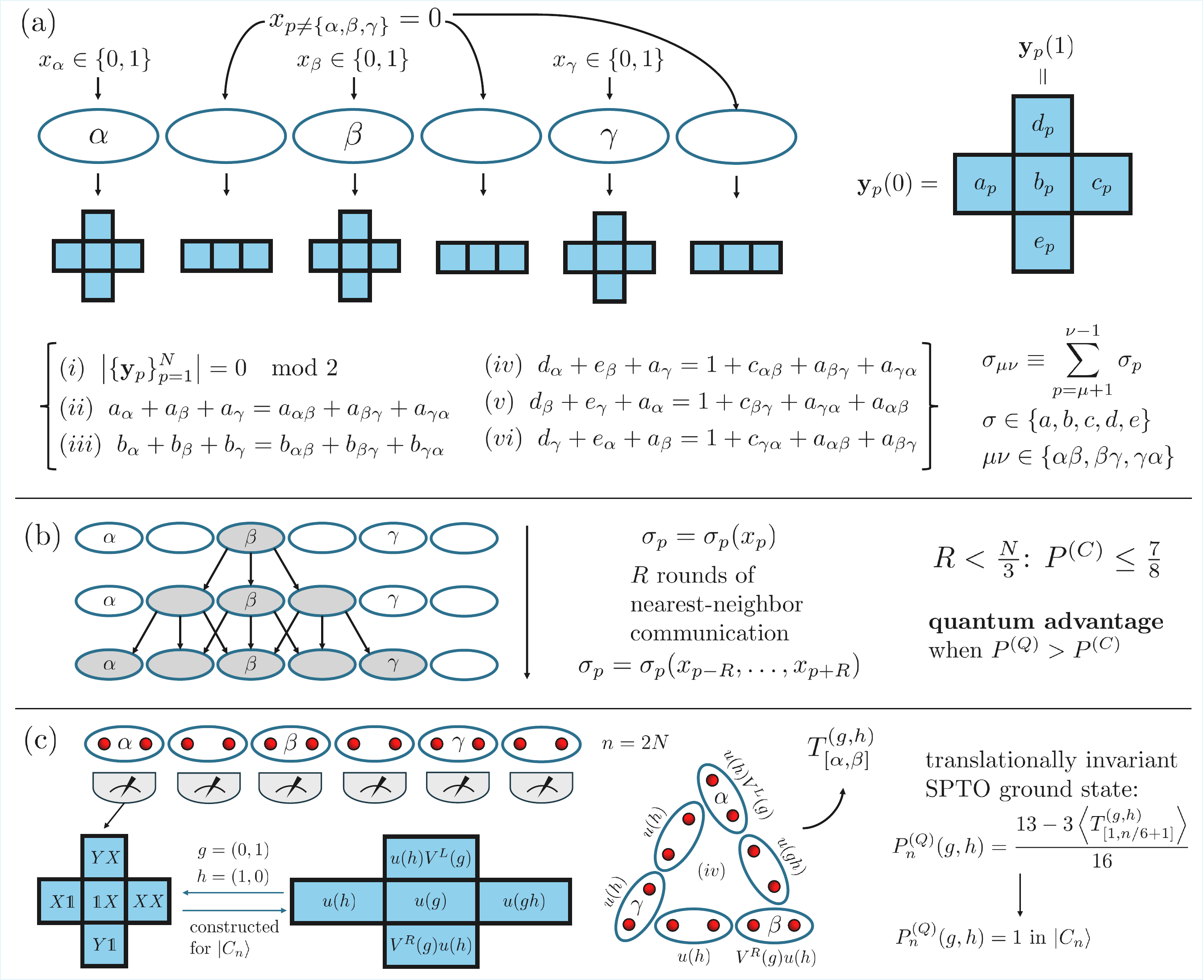}
\caption{
Multiplayer contextual triangle game and definition of quantum advantage. (a) Setup of the game.  Three players $\alpha,\beta,\gamma$, equally spaced within $N$ total players, are each given inputs $\{0,1\}$ uniformly at random, while all other players are given input $0$ and thus only ever output $(a_p,b_p,c_p)$. We can imagine players $\alpha,\beta,\gamma$ at the corners of a triangle. For each triangle edge, denoted $\mu\nu=\{\mu+1,\hdots,\nu-1\}$ for $\mu\nu\in\{\alpha\beta,\beta\gamma,\gamma\alpha\}$ and each variable $\sigma\in\{a,b,c\}$, an ``edge string'' $\sigma_{\mu\nu}$ is defined as the sum of all variables of type $\sigma$ on the edge $\mu\nu$. Winning conditions of the game are shown in the brackets and are analogous to the three-player game conditions Eqs. \eqref{eq:win1}-\eqref{eq:win4} up to addition of edge strings from the remaining $N-3$ players. (b) Classical strategies. We consider strategies that consist of $R$ rounds of nearest-neighbor communication where players can share their input information. When $R$ reaches $N/3$ the players $\alpha,\beta,\gamma$ are able to share input information and win with unit probability. However, when $R<N/3$ all classical strategies are bounded by $P^{(C)}\le7/8$. Our notion of quantum advantage is defined by a multiplayer triangle game strategy which uses quantum resources to achieve $P^{(Q)}>7/8$. (c) Quantum strategy. Each player $p$ holds 2 adjacent qubits $2p-1$, $2p$ of an $n=2N$-qubit state $\ket{\phi}$ (which we take to be translationally invariant) and outputs measurement outcomes of the cluster state symmetry operators shown in the table. We refer to this measurement strategy as the \emph{cluster state measurement strategy}. The global measurements corresponding to win conditions $(ii)$ and $(iii)$ form global symmetries and those corresponding to $(iv)$-$(vi)$ form twisted string order parameters. An SPTO ground state is symmetric so that the winning probability only depends on its expectation value of the cluster state twisted SOP. Correspondingly, the cluster state $\ket{C_n}$ wins with unit probability.}
\label{fig:multiplayertrianglegame}
\end{figure*}

Entanglement can be harnessed in a task called a nonlocal game in which players can gain an advantage by using quantum resources. In a nonlocal game, players $p\in\{1,2,...,N\}$ are given input bits $\vb{x}_p$ and asked to produce output bits $\vb{y}_p$ that satisfy a set of global linear constraints. While classical communication is prohibited between spatially separated players in the conventional setting of nonlocal games, we will generalize quantum advantage to a {\em computational} setting, by allowing limited classical communication later.

A classical strategy consists of each player assigning an output that is a deterministic or probabilistic function of inputs accessible within their lightcone. However, the richer algebraic structure of quantum mechanics allows for more powerful strategies. In particular, certain quantum states exhibit state-dependent contextuality, meaning that the outcome of measuring a given observable on that state can depend on which other jointly measurable (mutually commuting) observables are measured alongside it. Note that this arises only for specific states and observables rather than being a universal feature of all quantum states (another phenomenon known as state-\emph{independent} contextuality). A quantum strategy allows players to use contextual measurement outcomes of such a quantum state. When a quantum strategy for a nonlocal game wins with probability greater than any allowed classical strategy, we say the state gives quantum advantage in the game.

We consider a contextual game that is inherently connected to SPTO. Contextuality witnesses connected to other phases of matter have also been proposed \cite{Greenberger1990,Bulchandani2,Bulchandani3,ToricCodeNonlocalGame,Braidingforthewin}, however, SPT phases are particularly interesting in the near-term as their states can be prepared by constant-depth local quantum circuits without nonlocal classical communication of intermediate measurement outcomes. Furthermore, there is a known comparison between our nonlocal game and the so-called ``hidden linear function'' problem of Ref. \cite{ShallowCircuits}, which was originally used to show an unconditional separation between constant-depth local quantum circuits and constant-depth classical circuits of bounded fan-in gates. Our game thus has an inherent connection to computational quantum advantage. See Ref. \cite{Austin} and its supplemental material for more detail.

\subsection{Definition and quantum advantage}
We review the ``contextual triangle game'' construction from Ref. \cite{Austin} that was based on prior works \cite{modelingPaulimeas,ShallowCircuits,Bravyi2020}. In the simplest version of the game, three spatially separated players $p\in\{1,2,3\}$ are each given an input bit $x_p\in\{0,1\}$ drawn uniformly at random. Each player is asked to generate a three-bit output $\vb{y}_p\in\{0,1\}^3$ without classically communicating. Players' output depends on their input in the sense that $\vb{y}_p=\vb{y}_p(x_p)$: if player $p$ receives input $x_p=0$, they assign their three-bit output to the variables $\vb{y}_p(0)=(a_p,b_p,c_p)$, whereas if they receive input $x_p=1$ they assign it to variables $\vb{y}_p(1)=(d_p,b_p,e_p)$. Note that the variable $b_p$ is included in all outputs but other variables are input-dependent.
The players win the game if their joint output $\vb{y}=(\vb{y}_1,\vb{y}_2,\vb{y}_3)$ satisfies the following global linear constraints:
\begin{subequations} \label{eq:main}
\begin{align}
    &\abs{\vb{y}} = 0, \label{eq:win1} \\
    &a_1 + a_2 + a_3 = 0, \label{eq:win2} \\
    &b_1 + b_2 + b_3 = 0, \label{eq:win3} \\
    &d_p + e_{p+1} + a_{p+2} = 1 \quad \forall p, \label{eq:win4}
\end{align}
\end{subequations}
where $\abs{\cdot}$ denotes the Hamming weight and addition of the variables is done modulo 2. Additionally, the notation of players $p$ is cyclic. The winning conditions Eqs. \eqref{eq:win2} and \eqref{eq:win4} are only applied when certain inputs are given, whereas Eqs. \eqref{eq:win1} and \eqref{eq:win3} are applied to the output of all eight possible three-bit input strings. We will denote the average winning probability over all eight inputs as $P$, using superscripts $(C)$ or $(Q)$ to denote whether a classical or quantum resource is used.

A classical strategy assigns outputs as a linear boolean function of inputs. For example, a non-communicating classical strategy assigns $\sigma_p=\sigma_p(x_p)$  for $\sigma\in\{a,b,c,d,e\}$. Adding the winning conditions Eqs. \eqref{eq:win2} and \eqref{eq:win4} $\forall p$ with Eq. \eqref{eq:win3} for input $(1,1,1)$ gives $\sum_{p=1}^3(d_p+b_p(1)+e_p)=1$, implying the global output $\vb{y}(1,1,1)$ cannot satisfy Eq. \eqref{eq:win1}  (note we've written only $b_p$ as a function of the input because $d_p$ and $e_p$ only occur with input $x_p=1$). Thus any non-communicating classical strategy is constrained by the average winning probability $P^{(C)}\le7/8$, assuming all eight possible inputs are chosen uniformly at random.
While we focus on the $7/8$ bound in this paper, we note that the bound can be lowered to $4/5$ if only a certain subset of eight inputs (namely certain 5 inputs) are drawn randomly (see Ref.~\cite{AustinExperiment}).

On the other hand, if players are given the six-qubit cyclic cluster state $\ket{C_6}$, they can win the game with unit probability. Each player $p$ holds the two qubits $2p-1, 2p$ and measures two-qubit Pauli observables according to their input: for input $x_p=0$, the player outputs $ (a_p,b_p,c_p)$ according to measurement outcomes of $(X\mathds{1},\mathds{1}X,XX)$, whereas for $x_p=1$ the player outputs $(d_p,b_p,e_p)$ according to the outcomes of $(YX,\mathds{1}X,Y\mathds{1})$. Operators within each measurement context pairwise commute and thus allow simultaneous measurement outcomes, and additionally multiply to the identity, constraining the output string to have even parity and satisfying Eq. \eqref{eq:win1}. The other win conditions are satisfied as the global measurements corresponding to Eqs. \eqref{eq:win2} and \eqref{eq:win3} form cluster state stabilizers and the global measurement corresponding to Eq. \eqref{eq:win4} forms the negative of a stabilizer.

This game can be extended to a multipartite setting as explained in Fig. \ref{fig:multiplayertrianglegame} and Ref.~\cite{Austin}. Three players $\alpha,\beta,\gamma$ of $N=n/2$ total players are arbitrarily chosen to play the three-player triangle game and all other players are given input $0$. By sharing the $n$-qubit cyclic cluster state $\ket{C_n}$, players still maintain quantum advantage in the sense that even any classical strategy with geometrically restricted communication cannot win with probability greater than $7/8$ until the rounds of communication become of the order of the side length of the triangle $N/3=n/6$.

The multiplayer triangle game may be reframed as a computational task where a device is asked to solve a relation problem specified by the winning conditions of the game (see supplementary material of Ref. \cite{Austin}). In this sense the classical strategies we consider, which allow $R$ rounds of nearest-neighbor communication, correspond to depth-$R$ classical circuits consisting only of gates with nearest-neighbor (three-input) fan-in (Fig. \ref{fig:multiplayertrianglegame}(b)). As mentioned, the classical circuit is only able to solve the problem with unit probability when $R=\mathcal{O}(n/6)$. In contrast, the perfect quantum strategy requires only a constant-depth nearest-neighbor quantum circuit which prepares and measures $\ket{C_n}$, revealing an unconditional separation between the power of constant-depth quantum circuits and sublinear-depth classical circuits under the same gate connectivity in 1D. An embedding to a 2D square grid allows one to remove geometric constraints in the classical strategy and to grant more power to classical circuits, as originally done in Ref.~\cite{ShallowCircuits}.

For our purposes we restrict to the situation where players $\alpha,\beta,$ and $\gamma$ are equally spaced as this is the most difficult for classical strategies. Additionally, we consider classical strategies that only allow for $R<N/3$ rounds of nearest-neighbor classical communication. Players then obtain quantum advantage if they use a quantum strategy to obtain $P^{(Q)}>7/8$. When comparing against the computational power of classical circuits later, we will restrict this so that the quantum resources must be preparable in constant depth. To emphasize that this probability in the many-body game depends on the size $n$ of the resource state, we will write the probability as $P_n^{(Q)}$.

The quantum strategy of the multiplayer triangle game relies on players locally performing two-qubit measurements which globally form the SPTO operators of the cluster state shown in Fig. \ref{fig:SOP}. In particular, for input $x_p=0$, the player outputs $ (a_p,b_p,c_p)$ according to measurement outcomes of the two-qubit operators $(u(h),u(g),u(gh))$, whereas for $x_p=1$ the player outputs $(d_p,b_p,e_p)$ according to the outcomes of $(u(h)V^L(g),u(g),V^R(g)u(h))$. We will call this the ``cluster state measurement strategy''.  Note that there are several choices of measurement settings corresponding to different choices of group elements $g,h\in\mathds{Z}_2\times\mathds{Z}_2$. This is a natural generalization of the three-player case: taking $n=6$ and $g=(0,1)=x$, $h=(1,0)=y$ reconstructs the three-player strategy. 

This notation makes it clear that the measurement outcomes of SPTO operators are actually what control the winning probabilities. Specifically, winning conditions which require even parity ($(ii)$ and $(iii)$ in Fig. \ref{fig:multiplayertrianglegame}) correspond to measurements of global symmetries of the cluster state, whereas winning conditions which require odd parity ($(iv)-(vi)$) correspond to measurements of its' twisted string order parameters.
Recall that expectation values of cluster state twisted SOPs give a diagnostic on whether a state is in the nontrivial $\mathds{Z}_2\times\mathds{Z}_2$ SPT phase, implying that states in the same phase as $\ket{C_n}$ have the potential to allow for quantum advantage in the triangle game using the cluster state measurement strategy. We now elaborate on this fact.

\subsection{SPTO pure-state resources}

It was shown in Ref.~\cite{Austin} that many $n$-qubit pure states in the nontrivial (cluster) $\mathds{Z}_2\times\mathds{Z}_2$ SPT phase also admit quantum advantage in the $n/2$-player triangle game, in the sense that $P_n^{(Q)}>7/8$ using only a constant-depth quantum circuit and the cluster state measurement strategy. Particularly, if the expectation value of all three nontrivial cluster state SOPs (Eq. \eqref{eq:SOP} and Fig. \ref{fig:SOP}) are greater than 1/3 in the state $\ket{\phi}$, then it has a winning probability $P^{(Q)}_n>7/8$:
\begin{align}
\ev{\mathcal{S}}\equiv\min_{g\in\mathds{Z}_2\cross\mathds{Z}_2}\bra{\phi}S_{[p,q]}(g)\ket{\phi}>\frac{1}{3}\Rightarrow{}P_n^{(Q)}>\frac{7}{8}.
\label{eq:classicalbound}
\end{align}
The two-site blocks $p,q$ must be chosen such that $S_{[p,q]}(g)$ spans the distance between the qubits held by adjacent players $\alpha,\beta,\gamma$, which in our setting is $N/3=n/6$. However, string order parameters remain nonzero in the thermodynamic limit for SPT ground states, so that a sufficiently high SOP implies quantum advantage for all multiplayer triangle game instances. In Ref.~\cite{Austin} a full proof of Eq. \eqref{eq:classicalbound} is given.

We apply this theorem to the SPT-ordered ground states of the Hamiltonian Eq. \eqref{eq:Hamiltonian}: the phase diagram in Fig. \ref{fig:GSadvantage} shows that a large portion of ground states in the 1D cluster SPT phase provide quantum advantage in the contextual triangle game. In the next section we prove a general bound regarding mixed states in which this equation emerges as the pure-state limit when the state is symmetric.

Note that only a part of the SPT phase may be used as a contextuality witness, as here we utilize the cluster-state measurement strategy independent of the SPT states. It remains an open question whether nonzero SOP expectation values imply the presence of contextuality in a similar fashion with Ref.~\cite{Raussendorf2023measurementbased}; 
The setup of the latter allows for communication between different sites of the chain and is thus not identical to our nonlocal game setting.

\subsection{SPTO mixed-state resources}
\label{sec:mixedresources}

A mixed state rather than a pure state may be used in the contextual triangle game as a quantum resource. However, unlike the SPT-ordered ground states of Hamiltonian Eq. \eqref{eq:Hamiltonian}, arbitrary mixed states are generally not symmetric, i.e. $\ev{U(g)}\ne 1$. The quantum advantage is extended to mixed states as follows.

\newtheorem{thm}{Theorem}
\begin{thm}\label{thm:1}
Let $\rho$ denote a translationally invariant mixed state on a spin-1/2 chain of even length $n$.
Using the cluster state measurement strategy on $\rho$, the winning probability of the $n/2$-player contextual triangle game for all possible random inputs and equidistant players $\alpha,\beta,\gamma$ is denoted as $P_n^{(Q)}(g,h)$
with nontrivial group elements $g,h\in \mathds{Z}_2\times\mathds{Z}_2$. Then
\begin{equation}
\begin{aligned}
P_n^{(Q)}(g,h)=&\frac{1}{32}\bigg[12\Big(1+\ev{U(g)}\Big)+\ev{U(h)}+\ev{U(gh)}\\
&-3\Big(\ev{T_{[1,n/6+1]}^{(g,h)}}+\ev{U(g)T_{[1,n/6+1]}^{(g,h)}}\Big)\bigg] ,\label{eq:pwin}
\end{aligned}
\end{equation}
using the symmetry operators and twisted string order parameters defined in Section \ref{sec:SPTO}, and the overall performance of quantum players is determined by 
\begin{equation}
P_n^{(Q,min)}\equiv\min_{g\ne h\in\mathds{Z}_2\times\mathds{Z}_2}P_n^{(Q)}(g,h).
\end{equation}
In particular, a nontrivial SPTO state is necessary to pass the value $P_n^{(Q)}(g,h)>13/16$.
\end{thm}

\begin{proof} We use the same proof as Theorem 2 of Ref. \cite{Austin} but do not assume the state is pure or symmetric. By definition, $P_n^{(Q)}(g,h)$ is the average winning probability over all triangle game inputs using the cluster state measurement strategy. For the four inputs where $x_1 +x_2 +x_3=1\mod2$, only winning condition $(iii)$ in Fig. \ref{fig:multiplayertrianglegame} is applied (besides $(i)$ which is automatically satisfied using the cluster state measurement strategy), so that globally $U(g)$ is measured. The winning condition $(iii)$ is satisfied only when the measurement outcome of $U(g)$ is $+1$.

For the input $(0,0,0)$, both $(ii)$ and $(iii)$ are applied, which are satisfied only when both $U(h)$ and $U(g)$ have measurement outcome $+1$. Finally for the remaining inputs $(0,1,1),(1,0,1), (1,1,0)$,  winning conditions $(iii)$ and one of $(iv)-(vi)$ are applied, so that $U(g)$ and $T_{[p,q]}^{(g,h)}$ are measured for some $p,q\in\{\alpha,\beta,\gamma\}$ depending on the input. The winning conditions are satisfied when they have measurement outcomes of $+1$ and $-1$, respectively. 
Since $\rho$ is translationally invariant only the value $q-p$ affects the expectation value of $T_{[p,q]}^{(g,h)}$ for fixed $(g,h)$. In our setting of equidistant players, $q-p=N/3=n/6$ so that all twisted SOP expectation values are equivalent to $\ev{T_{[1,n/6+1]}^{(g,h)}}$.

Let $P(\mathcal{O}=a)$ denote the probability that measurement of the operator $\mathcal{O}$ in the state $\rho$ results in eigenvalue $a$. Averaging over all inputs, the probability of winning using this strategy is then
\begin{widetext}
\begin{align}
    P_n^{(Q)}(g,h)&=\frac{1}{8}\bigg[4P\big(U(g)=+1\big)+P\big(U(g)=+1,U(h)=+1\big)+3P\big(U(g)=+1,T_{[1,n/6+1]}^{(g,h)}=-1\big)\bigg]\\
    &=\frac{1}{32}\bigg[12\Big(1+\ev{U(g)}\Big)+\ev{U(h)}+\ev{U(gh)}-3\Big(\ev{T_{[1,n/6+1]}^{(g,h)}}+\ev{U(g)T_{[1,n/6+1]}^{(g,h)}}\Big)\bigg], \label{eq:pwinproof}
\end{align}
\end{widetext}
where $\ev{\mathcal{O}}=\tr(\mathcal{O}\rho)$. In the second line we used the identities $P(\mathcal{O}=a)=(1+a\ev{\mathcal{O}})/2$ and $P(\mathcal{O}_1=a,\mathcal{O}_2=b)=(1+a\ev{\mathcal{O}_1}+b\ev{\mathcal{O}_2}+ab\ev{\mathcal{O}_1\mathcal{O}_2})/4$ since all symmetry and boundary operators of the cluster state are jointly-measurable Paulis with eigenvalues $a,b\in\{\pm1\}$.
\end{proof}

In the case that $\rho$ is pure and symmetric, $\ev{U(g)}=1$ and $\ev{T_{[p,q]}^{(g,h)}}=\ev{\Omega(g,h)U(h)S_{[p,q]}(g)}=-\ev{S_{[p,q]}(g)}$ for distinct nontrivial choices of $g$ and $h$, reducing Eq. \eqref{eq:pwin} to the SOP condition given by Eq. \eqref{eq:classicalbound}. Note that the proof makes no explicit mention of mixed states but only takes into account the potential non-symmetric nature of the state. 
This theorem leads to the following corollary which connects the performance of the mixed state $\rho$ in the multiplayer triangle game to quantum \emph{computational} advantage.

\newtheorem{cor}{Corollary}
\begin{cor}\label{cor:1}
If the mixed state $\rho$ of Theorem~\ref{thm:1} is preparable by a nearest-neighbor quantum circuit in constant depth, a quantum computation that solves the $n/2$-player contextual triangle game with $P_n^{(Q,min)}$ greater than $7/8$ proves computational advantage, in that any classical computation using geometrically-restricted fan-in gates with bounded inputs from the nearest neighboring players needs a depth growing linearly as $n/6$ to match the winning probability.
The computational advantage is scalable asymptotically for SPTO pure states but may not be scalable asymptotically for mixed states.

\end{cor}

\section{Demonstrations of noisy quantum advantage}\label{sec:results}

We are interested in how the minimum winning probability $P_n^{(Q,min)}$ changes for mixed states generated by noisy preparation of the SPT ground states of Eq. \eqref{eq:Hamiltonian}. Gaining insight into this question is difficult as (i) NISQ devices are exposed to complicated noise that can vary between devices, and (ii) even if the exact noise model was known, an exact classical simulation of the $2^{n} \times 2^{n}$-dimensional mixed quantum state would be computationally unattainable for reasonable $n$. To counter the first obstacle we consider the Gibbs state. This state has a statistical mechanical meaning as a Hamiltonian system which is in thermal equilibrium with an external environment, therefore washing out details of the specific noise model. By considering this state, the second obstacle is also overcome: an efficient algorithm known as the minimally entangled typical thermal states (METTS) algorithm exists which approximates expectation values in the Gibbs state by sampling over classically simulable Matrix Product States. In this section we use the METTS algorithm to show that quantum advantage in the triangle game persists to nonzero temperatures. Details of the METTS algorithm and its application to our situation are explained in Appendix \ref{sec:METTS}.

The Gibbs state for a system under the Hamiltonian $H$ is given by
\begin{align}
    \rho=\frac{1}{\mathcal{Z}}e^{-\beta H},\hspace{5mm}\mathcal{Z}=\tr(e^{-\beta H}),\label{eq:Gibbs}
\end{align}
where $\beta=(k_BT)^{-1}$ and we take $k_B=1$. 
Using Theorem~\ref{thm:1}, we will show that thermal states of the nontrivial SPTO regime of Hamiltonian Eq. \eqref{eq:Hamiltonian} can still exhibit quantum advantage in the multiplayer triangle game using the cluster state measurement strategy, up to some critical temperature which is only dependent on the system size. In the thermodynamic limit this temperature drops to zero, but for finite-sized systems it remains nonzero and is surprisingly robust to system size. We study a system with $n=64$ total qubits, chosen to reflect the capability of modern quantum devices. In our numerics we set $\Delta=2$ in the Hamiltonian.

To connect the thermal state $\rho$ to quantum computational advantage by applying our Corollary \ref{cor:1}, it must be preparable by a constant-depth nearest-neighbor quantum circuit. Imagine a situation where one attempts to prepare a pure SPT state but instead prepares a mixed state due to an imperfection of a quantum device. In this situation, the constant-depth preparation of SPT states translates to constant-depth preparation of the corresponding mixed states. However, here we will not consider a noise model specific to a particular device, but rather a ubiquitous noise model given by the Gibbs thermal state which often has nontrivial complexity to prepare. It is noteworthy that no nontrivial SPT phase protected by onsite symmetry is robust (in the thermodynamic limit) above zero temperature \cite{Roberts}.
On the other hand, it is also known that a thermal cluster state is equivalent to the mixed state resulting from dephasing noise applied independently to each qubit of the pure cluster state, so that the thermal cluster state can indeed be constructed by the constant-depth channels consisting of preparation of the cluster state followed by single-qubit dephasing noise. However, whether Gibbs thermal states associated with an arbitrary $\mathds{Z}_2\times\mathds{Z}_2$ 1D SPT Hamiltonian can be constructed by a constant-depth quantum channel seems to be an open question \cite{Brandao2019}. This may be provable through techniques similar to those in Refs. \cite{Sang2024,Lake2025}.

\subsection{Cluster state analysis}\label{sec:clusteranalysis}

Before investigating the behavior of arbitrary thermal mixed states of the SPTO Hamiltonian in Eq. \eqref{eq:Hamiltonian}, it is helpful to investigate the behavior in the exactly solvable thermal cluster state. 
As mentioned, the thermal cluster state can be seen as a pure cluster state with single-qubit dephasing applied to each qubit with probability $\epsilon=(1+e^{2\beta})^{-1}$. By studying the Gibbs state we thus also get a notion of dephasing noise for free.

Let $K_j=Z_{j-1}X_jZ_{j+1}$ denote the cluster state stabilizer. We consider periodic boundary conditions so that site labeling is cyclic. The thermal state $\rho=e^{\beta\sum_jK_j}/\mathcal{Z}$ can be written as
\begin{align}
    \rho=\frac{1}{2^n}\sum_{\vb{r}\in\mathds{Z}_2^n}\tanh^{\abs{\vb{r}}}\bigg(\frac{\beta\Delta}{2}\bigg)\kappa(\vb{r}),
\end{align}
where $\abs{\vb{r}}$ denotes the Hamming weight of the bit string $\vb{r}$ and
\begin{align}
\kappa(\vb{r})=\prod_{j=1}^nK_j^{r_j}.
\label{eq:kappa}
\end{align}

All global symmetry representations $U(g)$ and twisted SOPs $T_{[p,q]}^{(g,h)}$ can be written in the form of Eq. \eqref{eq:kappa}, allowing an analytic expression for their expectation values. For example, $U(g)=\kappa(\vb{r}_g)$ where $\vb{r}_g$ is the bit string with ones in every even, every odd, and every entry (with zeros everywhere else) for $g=(a,b)=(0,1),(1,0),(1,1)\in\mathds{Z}_2\times\mathds{Z}_2$, respectively. Thus
\begin{align}
\ev{U\big((a,b)\big)}=\tanh^{(a+b)n/2}\bigg(\frac{\beta\Delta}{2}\bigg).
\end{align}
With the same method we find
\begin{align}
    \ev{T_{[p,q]}^{(g,h)}}=-
    \begin{cases}
        \tanh^{(n+d)/2}(\frac{\beta\Delta}{2}) & g,h = x,y \text{ or } y,x\\
        \tanh^{n-d/2}(\frac{\beta\Delta}{2}) & g = x \text{ or } y, h=z\\
        \tanh^{n/2}(\frac{\beta\Delta}{2}) & g = z, h=x \text{ or } y
    \end{cases},
\end{align}
where $d=2(q-p)$.
Using Eq. \eqref{eq:pwin} and analyzing over all possible $g,h\in\mathds{Z}_2\times\mathds{Z}_2$ we then find that the minimum winning probability belongs to the strategy $(g,h)=(z,x)$ or $(z,y)$ with
\begin{align}
    P_n^{(Q,min)}=\frac{1}{8}\bigg[3+2\tanh^{n/2}\Big(\frac{\beta\Delta}{2}\Big)+3\tanh^n\Big(\frac{\beta\Delta}{2}\Big)\bigg].
    \label{eq:Pcluster}
\end{align}
This is greater than $7/8$ whenever $T=\beta^{-1}<T_c$, where
\begin{align}
    T_c=\frac{\Delta}{2\text{ arctanh}\bigg[\Big(\frac{\sqrt{13}-1}{3}\Big)^{2/n}\bigg]}.
    \label{eq:betac}
\end{align}
For our case of $n=64,\Delta=2$, this gives $T_c=0.327$, corresponding to a single-qubit dephasing noise probability of $\epsilon\approx 0.2\%$. A $64$-qubit thermal cluster state with $T<T_c$ thus admits quantum advantage in a $32$-player multiplayer triangle game instance.
Conversely, rearranging Eq. \eqref{eq:betac} shows that the thermal cluster state at fixed temperature $T$ supports quantum advantage in the multiplayer triangle game when the system size $n$ is less than $n_c$, where
\begin{align}
    n_c=\frac{2\ln((\sqrt{13}-1)/3)}{\ln(\tanh(\Delta/2T))}.
    \label{eq:ncrit}
\end{align}

This behavior is shown in Fig. \ref{fig:thermalcluster}. Surprisingly, although $T_c$ goes to zero in the thermodynamic limit, for reasonably large system sizes it remains relatively constant, only decreasing slightly as the system size is increased. Indeed, Eq. \eqref{eq:ncrit} implies that the maximum system size $n_c$ which supports quantum advantage scales exponentially as $e^{\Delta/T}$ in the small-$T$ limit. In this regime, small decreases in temperature (correspondingly decreases in dephasing probability) result in large increases in quantum advantageous system sizes.

\begin{figure}[t]
    \includegraphics[width=0.48\textwidth]{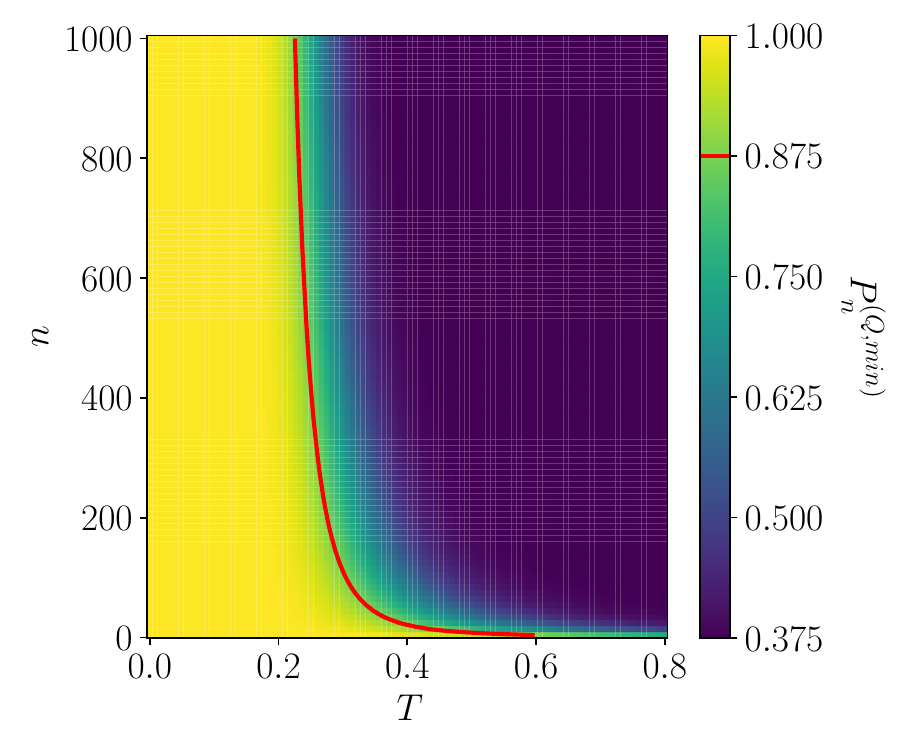}
    \caption{Behavior of the minimum winning probability $P_{n}^{(Q,min)}$ in a thermal cluster state as a function of temperature and system size. The quantum advantageous contour $P_{n}^{(Q,min)}=7/8$ is shown in red. Although the critical temperature $T_c$ (where $P_{n}^{(Q,min)}=7/8$ for fixed $n$) decays to zero in the thermodynamic limit, it is robust for finite-sized systems.}
    \label{fig:thermalcluster}
\end{figure}

\begin{figure}[t]
\centering

\begin{tikzpicture}[font=\large]

% -----------------------
% Layout parameters
% -----------------------
\def\W{8cm}          % width of each panel (a) and (b)
\def\Ysep{0.8cm}     % horizontal separation between panels

% ======================
% Left panel (a)
% ======================
\node[inner sep=0pt, anchor=north west] (A) at (0,0)
{
    \includegraphics[width=\W, keepaspectratio]{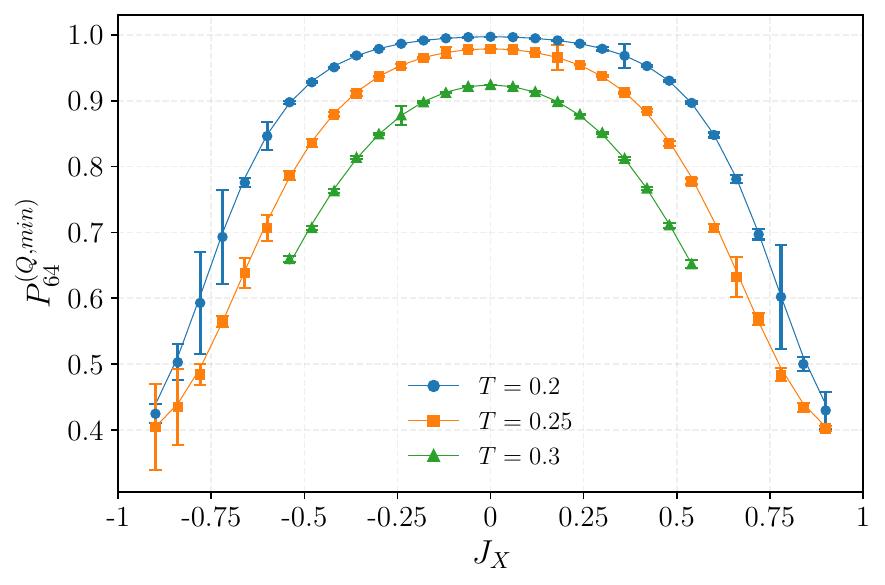}
};

\node[anchor=north] at (A.north west) {(a)};
% ======================
% Right panel (b)
% ======================
\node[inner sep=0pt, anchor=north west] (B)
    at ($(A.south west) - (0, \Ysep)$)
{
    \includegraphics[width=\W, keepaspectratio]{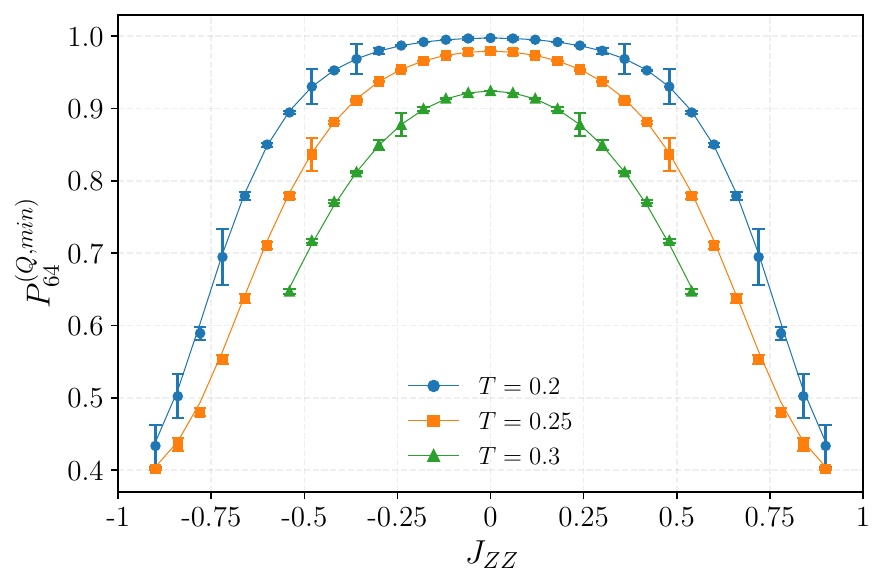}
};
\node[anchor=north] at (B.north west) {(b)};
\end{tikzpicture}

\caption{ Analysis of accuracy and precision of the METTS algorithm by plotting METTS estimates of triangle game winning probability against exact solutions along (a) the $J_X$ axis and (b) the $J_{ZZ}$ axis. Exact solutions are shown as solid lines and METTS estimates are plotted as points with error bars corresponding to relative error. The algorithm performs well across these axes.}
\label{fig:METTSvsExact}
\end{figure}

\subsection{Thermal advantage in the SPT phase}

\begin{figure*}[t]
\centering

\begin{tikzpicture}[font=\large]

% -----------------------
% Layout parameters
% -----------------------
\def\LeftW{10.5cm}      % width of left figure (a)
\def\RightW{7cm}     % width of right figures (b) and (c)
\def\Xsep{0.25cm}       % horizontal separation
\def\Ysep{0.6cm}     % vertical separation between (b) and (c)

% ======================
% Right top panel (b)
% ======================
\node[inner sep=0pt, anchor=north west] (B)
    at (0,0)
{
    \includegraphics[width=\RightW, keepaspectratio]{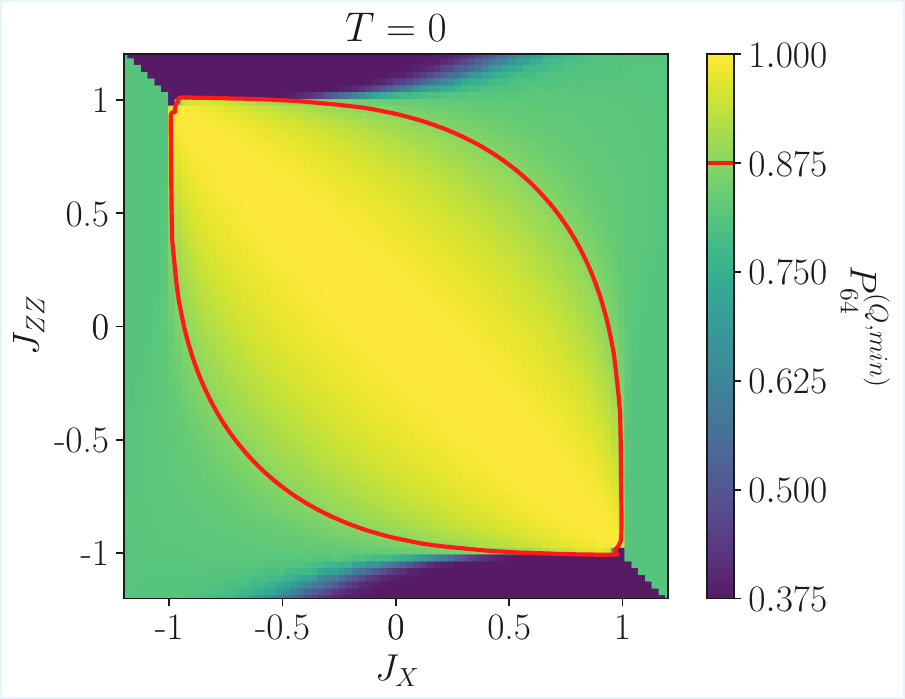}
};
\node[anchor=north west] at (B.north west) {(b)};

% ======================
% Right bottom panel (c)
% ======================
\node[inner sep=0pt, anchor=north west] (C)
    at ($(B.south west) + (0, -\Ysep)$)
{
    \includegraphics[width=\RightW, keepaspectratio]{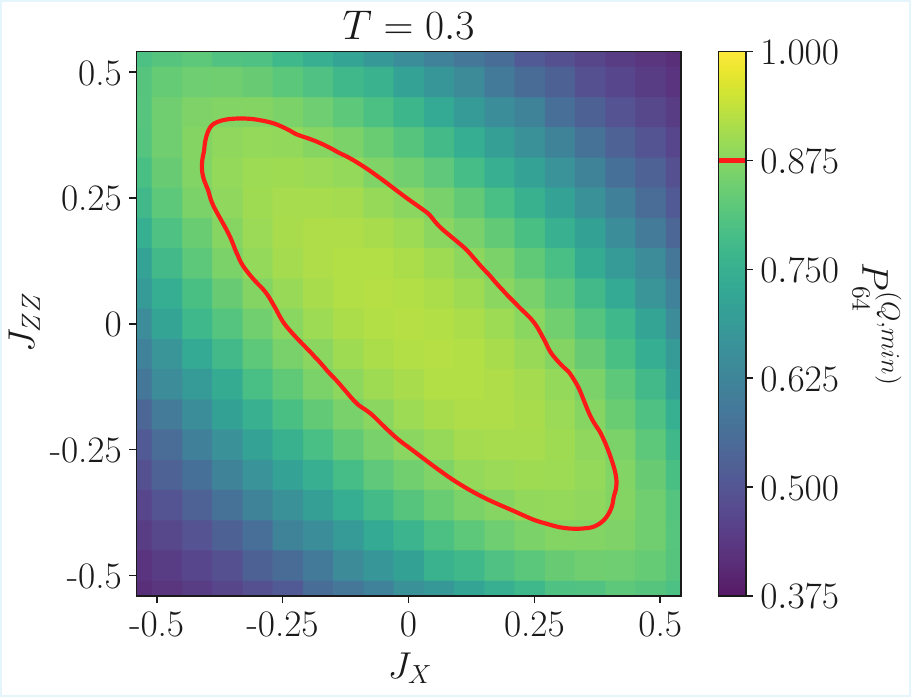}
};
\node[anchor=north west] at (C.north west) {(c)};

% ======================
% Left panel (a), vertically centered
% ======================
\node[inner sep=0pt, anchor=center] (A)
    at ($(B.north west)!0.5!(C.south west) + (-\Xsep - 0.5*\LeftW, 0)$)
{
    \includegraphics[width=\LeftW, keepaspectratio]{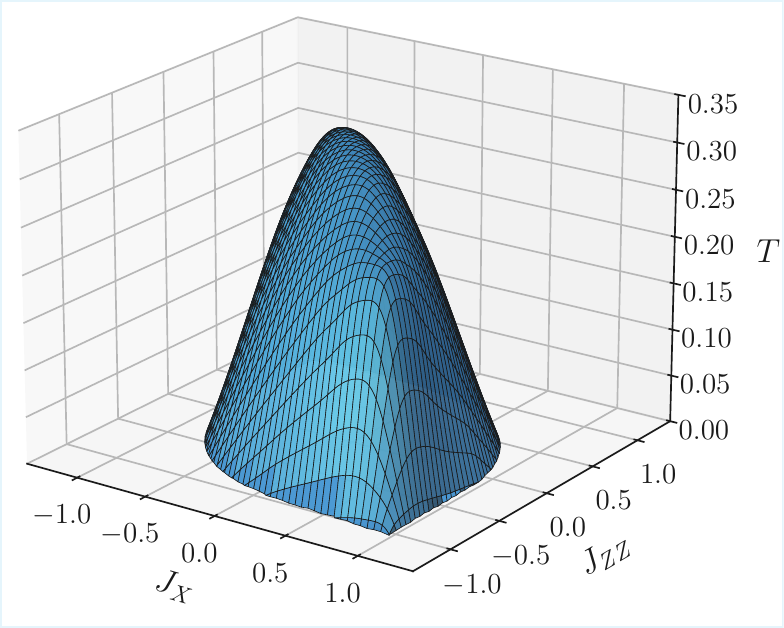}
};
\node[anchor=north west] at (A.north west) {(a)};

\end{tikzpicture}

\caption{
(a) Persistence of quantum advantage to finite temperature in a $N=32$-player contextual triangle game with $n=64$ qubits. The surface shows the boundary of the quantum advantageous region in thermal phase space. Gibbs states, characterized by Hamiltonian parameters $J_X,J_{ZZ}$ and temperature $T$, which lie inside the surface are able to beat classical strategies in the multiplayer triangle game (i.e. $P_{64}^{(Q,min)}>7/8$) by using the cluster state measurement strategy. As the temperature increases, the quantum advantageous region shrinks until completely disappearing at a finite temperature determined by Eq. \eqref{eq:betac}. (b) A cross section in the $T=0$ plane. The outline of the quantum advantageous region, shown in red, is identical to that shown in Fig. \ref{fig:GSadvantage}. This plot also shows the symmetry-broken phase of the Hamiltonian (purple), where the winning probability is 3/8 according to Eq. \eqref{eq:pwin}. (c) A cross section in the $T=0.3$ plane.}
\label{fig:thermaladv}
\end{figure*}
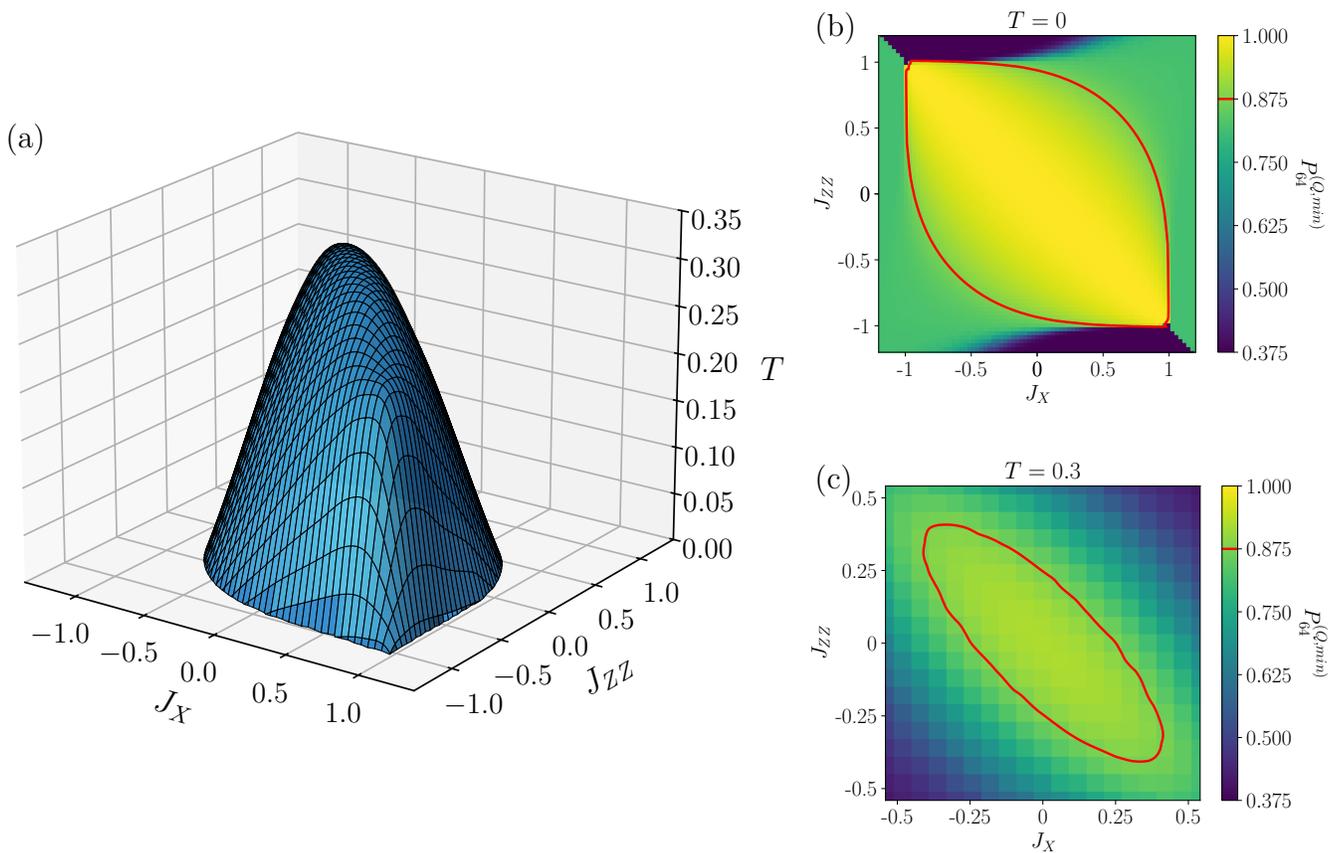

The exactly solvable region of our Hamiltonian further extends to the $J_X$ and $J_{ZZ}$ axes of the Hamiltonian by mapping the model to a quadratic fermionic Hamiltonian via a Jordan-Wigner transformation.
Thermal expectation values of symmetry operators and twisted string order parameters can then be exactly calculated \cite{LIEB1961407}. We outline this approach in Appendix \ref{app:exactsolns}. At $J_X=J_{ZZ}=0$ these results match the cluster state analysis. 
Along these axes, we find similar scaling results to the cluster state, namely, the critical temperature is robust to system size below the thermodynamic limit.

Using the METTS algorithm described in Appendix ~\ref{sec:METTS}, we extend the results of quantum advantage to arbitrary thermal states parameterized by $(J_X,J_{ZZ})$ in Hamiltonian Eq. \eqref{eq:Hamiltonian}. 
The exact solutions previously discussed allow us to fine-tune the parameters of the METTS algorithm to most accurately model the system; Figure \ref{fig:METTSvsExact} shows strong agreement between these exact solutions and the METTS algorithm. The exact solutions also help with the choice of the twisted SOP endpoint blocks $p,q$. These are selected to balance two constraints in the ground state: they must be far enough apart to capture bulk behavior, but not long enough that they are close due to the periodicity. We find that $p=7, q=26$ satisfy these properties.

Figure \ref{fig:thermaladv} shows our main numerical results. As temperature increases, the quantum advantageous region of Fig. \ref{fig:GSadvantage}, which is plotted again in Fig. \ref{fig:thermaladv}(b) using the success probability, shrinks until the temperature given by Eq. \eqref{eq:betac}, at which point all quantum advantage within the SPT phase disappears. The SPT phase (on finite systems) thus retains its inherently quantum mechanical properties which could allow its states to outperform classical strategies in the contextual triangle game when exposed to thermal noise.

The surface in Fig. \ref{fig:thermaladv}(a) represents the critical temperature $T_c(J_X,J_{ZZ})$, above which a 64-qubit thermal state loses quantum advantage. It is always less than or equal to the cluster state critical temperature $T_c(0,0)$ in Eq. \eqref{eq:betac}. In the thermodynamic limit, $T_c(J_X,J_{ZZ})\xrightarrow{}0$, but quantum advantage still exists as long as the system size $n$ remains finite. Conversely, for a fixed temperature $T$, there exists a critical system size $n_c(J_X,J_{ZZ})$ such that a thermal state can support quantum advantage in the $n/2$ player multiplayer triangle game when $n<n_c(J_X,J_{ZZ})$. Throughout the phase this is less than or equal to the cluster state result $n_c(0,0)$ in Eq. \eqref{eq:ncrit}. Using our Corollary \ref{cor:1} we can relate this to quantum \emph{computational} advantage. Since $n_c$ is finite for a fixed nonzero temperature we clearly do not have asymptotic computational advantage, as the quantum advantage is not scalable past $n\ge n_c$. However, it would be possible to witness the first glimpse of computational advantage of shallow quantum circuits when $n<n_c$ in that a classical circuit with nearest neighbor fan-in gates would require depths of $n/6$ to win the multiplayer triangle game with $P^{(C)}>7/8$.

\subsection{Player geometry does not affect minimum winning probability}
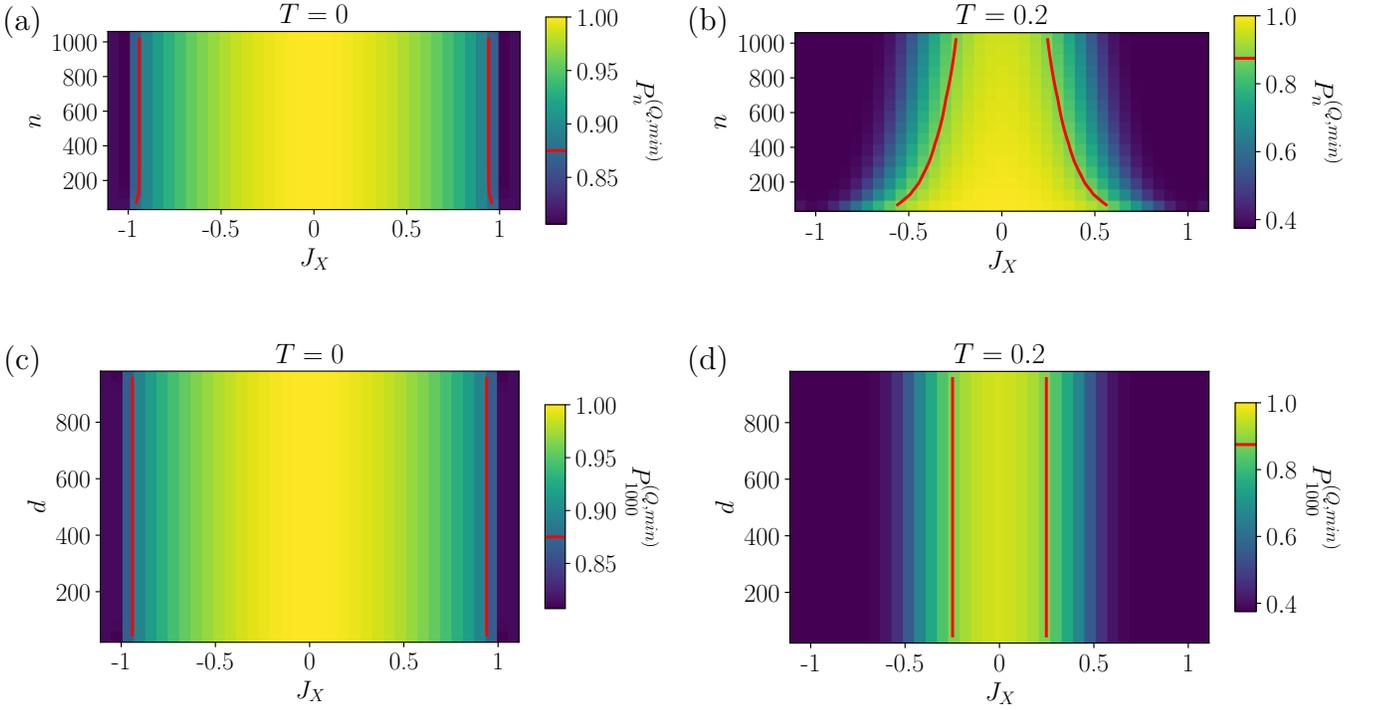
\begin{figure*}[t]
\centering

\begin{tikzpicture}[font=\large]

% -----------------------
% Layout parameters
% -----------------------
\def\W{8.5cm}          % width of each panel
\def\Xsep{0.6cm}     % horizontal separation
\def\Ysep{0.8cm}     % vertical separation

% ======================
% Top row: (a) (b)
% ======================
\node[inner sep=0pt, anchor=north west] (A) at (0,0)
{
    \includegraphics[width=\W]{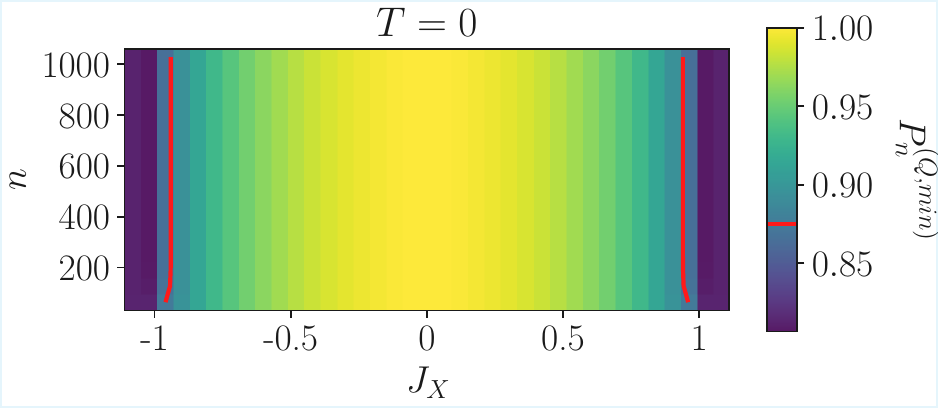}\label{fig:7a}
};
\node[anchor=north] at (A.north west) {(a)};

\node[inner sep=0pt, anchor=north west] (B)
    at ($(A.north east) + (\Xsep,0)$)
{
    \includegraphics[width=\W]{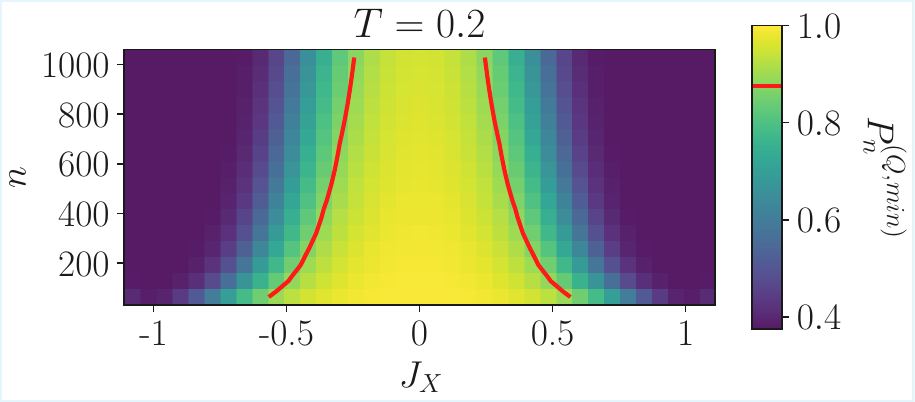}\label{fig:7b}
};
\node[anchor=north] at (B.north west) {(b)};

% ======================
% Bottom row: (c) (d)
% ======================
\node[inner sep=0pt, anchor=north west] (C)
    at ($(A.south west) - (0,\Ysep)$)
{
    \includegraphics[width=\W]{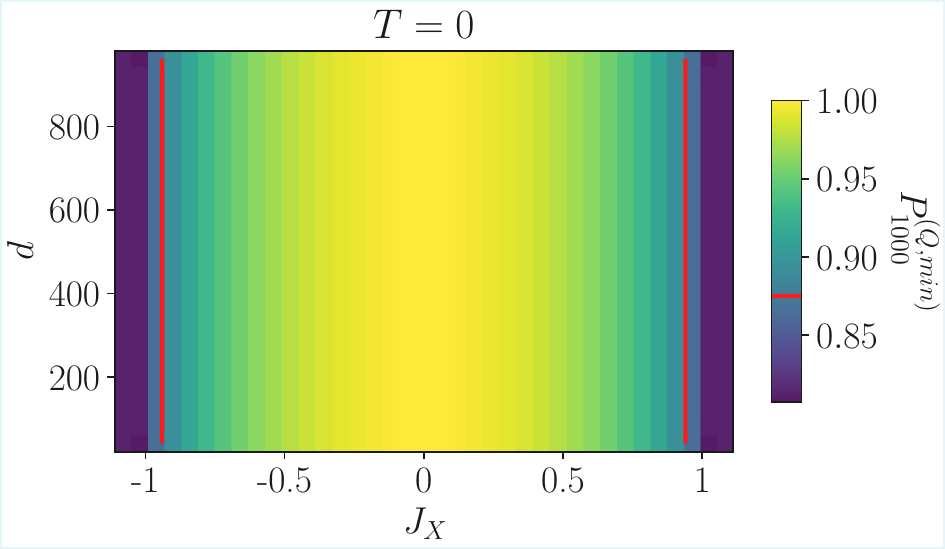}
};
\node[anchor=north] at (C.north west) {(c)};

\node[inner sep=0pt, anchor=north west] (D)
    at ($(C.north east) + (\Xsep,0)$)
{
    \includegraphics[width=\W]{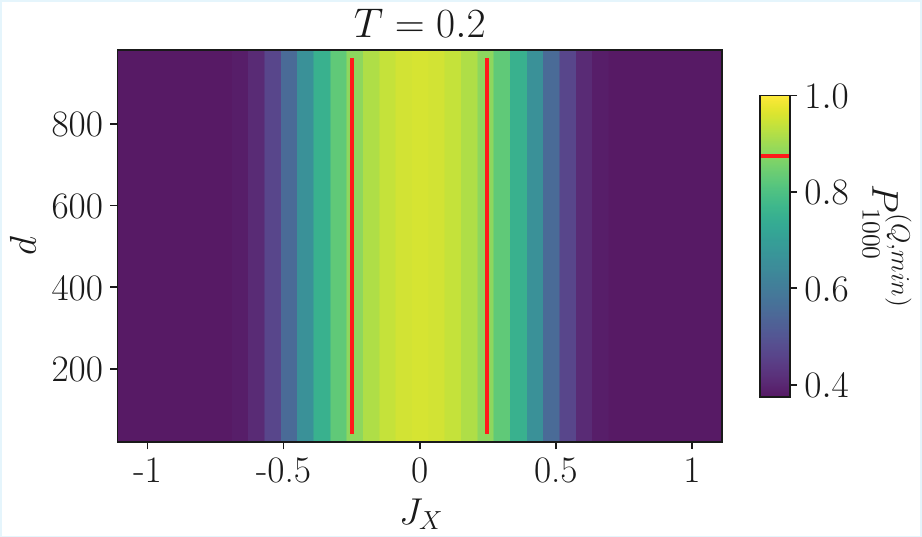}
};
\node[anchor=north] at (D.north west) {(d)};

\end{tikzpicture}

\caption{ Triangle game minimum winning probability $P_n^{(Q,min)}$ (using the cluster state measurement strategy) and quantum advantageous contour (red) of thermal states along the $J_X$ axis of Eq.~\eqref{eq:Hamiltonian} while changing length parameters.
(a) At zero temperature, $P_n^{(Q,min)}$ is constant with system size $n$.
(b) At finite temperature, $P_n^{(Q,min)}$ is nonzero for finite systems but vanishes in the thermodynamic limit.
Both (a) and (b) have constant $d=2(p-q)$ for twisted SOPs $T_{[p,q]}^{(g,h)}$ while varying $n$.
(c) and (d) fix $n=1000$ while varying $d$.
Varying SOP length at fixed system size has no effect on $P_{1000}^{(Q,min)}$ for all temperatures, indicating that the quantum advantageous region depends only on $n$ and $T$.}
\label{fig:varyingnjk}

\end{figure*}

It is notable that the cluster state minimum winning probability, given by Eq.~\eqref{eq:Pcluster}, is completely independent of the positions of players $\alpha,\beta,\gamma$; the only reference to the physical layout of the triangle game problem is in the system size $n$. Interestingly, we find numerical evidence that this is true throughout the quantum advantageous region. 
This behavior is shown in Fig. \ref{fig:varyingnjk} for states along the $J_X$ axis: provided the SOP length $d=2(q-p)$ is greater than the correlation length, then varying $d$ with a fixed system size $n$ does not affect $P_n^{(Q,min)}$, however, increasing $n$ with a fixed $d$ does decrease $P_n^{(Q,min)}$ when $T$ is nonzero.

Our numerics suggest that this behavior holds for all states in the SPT phase.
This implies that the quantum minimum winning probability in the multiplayer triangle game is not affected by the layout of the three players, and only by the size of the triangle game instance. This should be compared with the classical hardness, where $P^{(C)}$ increases from $7/8$ only when the depth of the classical circuit reaches the minimum separation between the three players, and thus depends on the specific layout of the players. Indeed, we assumed the three players were equidistant as this situation is the hardest for classical strategies, but these findings imply that this is actually irrelevant for $P_n^{(Q,min)}$.

\section{Connecting to experiment}\label{sec:experiment}
Although the thermal state analysis provides a concrete example of how quantum advantage could theoretically persist under noisy conditions, it is hard to compare with current experimental progress. The following theorem and corollary connect our work more closely to the quantum state fidelity prepared in experiments, hypothetically assuming that noiseless contextual measurements can be made.

\begin{thm}\label{thm:2}
Let $\rho$ be a translationally invariant density matrix on $n$ qubits, so that its use in the multiplayer triangle game is described by Theorem \ref{thm:1}. Furthermore, denote the winning probabilities of Eq. \eqref{eq:pwin} as $P_n^{(Q)}(g,h|\rho)$ to highlight their state-dependence. Define the graph state basis $\{\ket{G_{\vb r}}\}_{\vb{r}\in\mathds{Z}_2^n}$ as the simultaneous eigenbasis of all cluster state stabilizers $\{K_j=Z_{j-1}X_jZ_{j+1}|j\in 1,2,\hdots, n\}$ so that $K_j\ket{G_{\vb r}}=(-1)^{r_j}\ket{G_{\vb r}}$ and define $\Pi_{\vb{r}}=\ket{G_{\vb{r}}}\bra{G_{\vb{r}}}$. Then the winning probabilities $P_n^{(Q)}(g,h|\rho)$ only depend on $\rho$'s graph diagonal components $c_{\vb{rr}}=\bra{G_{\vb r}}\rho\ket{G_{\vb r}}$:
\begin{align}
    P_n^{(Q)}(g,h|\rho)=\sum_{\vb{r}\in\mathds{Z}_2^n}c_{\vb{rr}}P_n^{(Q)}(g,h|\Pi_{\vb r}).
    \label{eq:gdprobs}
\end{align}

\end{thm}

\begin{proof}
The winning probability $P_n^{(Q)}(g,h|\rho)$ can be written as the expectation value of the operator $\mathcal{O}=(12\mathds{1}+12U(g)+U(h)+U(gh)-3T_{[1,n/6+1]}^{(g,h)}-3U(g)T_{[1,n/6+1]}^{(g,h)})/32$ as in Eq. \eqref{eq:pwin}. Expanding out $\rho$ in the graph state basis
\begin{equation}
\begin{aligned}
    P_n^{(Q)}(g,h|\rho)&=\tr(\rho\mathcal{O})\\
    &=\tr\bigg(\sum_{\vb r,\vb s\in\mathds{Z}_2^n}c_{\vb{rs}}\ket{G_{\vb{r}}}\bra{G_{\vb{s}}}\mathcal{O}\bigg)\\
    &=\sum_{\vb r,\vb s\in\mathds{Z}_2^n}c_{\vb{rs}}\bra{G_{\vb{s}}}\mathcal{O}\ket{G_{\vb r}}.
\end{aligned}
\label{eq:trrhoo}
\end{equation}
The graph state basis is an eigenbasis for $\mathcal{O}$ as
each of the operators contained in it, e.g. $\mathds{1},U(g), T_{[p,q]}^{(g,h)}$, can be written as products of cluster state stabilizers, i.e. the form of Eq. \eqref{eq:kappa}. Thus the off-diagonal terms in Eq. \eqref{eq:trrhoo} vanish. Equation \eqref{eq:gdprobs} is obtained by noticing $P_n^{(Q)}(g,h|\Pi_{\vb r}) =\bra{G_{\vb r}}\mathcal{O}\ket{G_{\vb r}}$.
\end{proof}

We note that this theorem connects to the study of strong and weak symmetries in mixed state phases. A weak symmetry is described by $U\rho U^\dag=\rho$, whereas the strong symmetry is a more constraining requirement that the symmetry must be respected for each element of the ensemble $U\rho=e^{i\theta}\rho$ \cite{Lieu2020,Lessa2025}. Equation \eqref{eq:gdprobs} shows that the winning probability $P_n^{(Q)}(g,h|\rho)$ is completely determined by the weakly symmetric state that is the diagonal elements of $\rho$ in the graph state basis. However, the winning probability can only reach one when the state is strongly symmetric, as all symmetry representations must have a unit expectation value. For the winning probability to fully reach one, there is also the additional requirement that all twisted SOP expectation values must be one. This is interesting to compare to Ref. \cite{deGroot2022}, which found that SOP expectation values are preserved if and only if the noise respects a strong symmetry.

In the graph state basis notation, the cluster state is denoted $\ket{G_{\vb 0}}$, where $\vb 0$ is the all-zero bit string. Since the cluster state performs perfectly under the cluster state measurement strategy, $P_n^{(Q)}(g,h|\Pi_{\vb 0})=1$ and thus
\begin{align}
    P_n^{(Q)}(g,h|\rho)=c_{\vb{00}}+\sum_{\vb{r}\ne\vb{0}\in\mathds{Z}_2^n}c_{\vb{rr}}P_n^{(Q)}(g,h|\Pi_{\vb r}).
\end{align}
Since $c_{\vb{00}}$ is the fidelity of $\rho$ with the cluster state and $c_{\vb{rr}}$ and $P_n^{(Q)}$ are both nonnegative, this leads to the following useful corollary.
\begin{cor}
Let $f_n$ denote the global fidelity between the density matrix $\rho$ and the $n$-qubit cyclic cluster state $\ket{C_n}$. Then $f_n$ lower bounds the quantum winning probability obtained by using the cluster state measurement strategy on $\rho$ in the $n/2$-player contextual triangle game, namely
\begin{align}
    f_n\le P_n^{(Q,min)}
    \label{eq:fidelity} .
\end{align}
\label{cor:2}
\end{cor}

Equality in Eq. \eqref{eq:fidelity} is reached only when $\rho$ is the pure cluster state.

Corollary \ref{thm:2} allows us to link this work to experiment.
The reference \cite{AustinExperiment} reports an early experiment of the triangle game, by preparing a six-qubit cluster state on a trapped-ion quantum computer with a fidelity of $0.606$ before correcting measurement readout errors. It is clearly insufficient to leverage our Corollary \ref{cor:2}, but the work also evaluated directly the winning probability $P_6^{(Q,min)}$ to be 0.794, which is close to the classical bound $4/5$ of a variant of the triangle game mentioned briefly in Section~\ref{sec:trianglegame}.

A more recent work has shown that 2D cluster states can be built from microwave photons on superconducting devices of up to twenty photons with higher fidelities \cite{OSullivan2025}, and the six photon case has a fidelity of $0.77$. The reference \cite{Harvard} prepared a 12-qubit cluster state on a neutral atom quantum computer with an average  expectation value $\ev{K_j}=0.87$ of the stabilizer generators only (i.e., the operators with $\abs{\vb{r}}=1$ in Eq. \eqref{eq:kappa}) Similarly, Ref. \cite{VTech} 
prepared a 36-qubit (two-dimensional) cluster state on a neutral-atom analog quantum simulator, Aquila of QuEra, with average $\ev{K_j}=0.985$ for the stabilizer generators only. Although this is still not enough to verifiably apply Theorem 2, it shows that NISQ devices are very close to reaching the classical limit of computation.

Other works indicate that they actually have high enough quality many-body states to pass the classical limit of the contextual triangle game, hypothetically assuming that ideal contextual measurements could be applied. 
References \cite{Cao2023, jiang2025generation95qubitgenuineentanglement} created 1D cluster states on superconducting hardware with estimated fidelity greater than $7/8$ for up to $\approx 15$ qubits, using the randomized fidelity estimation of Ref.~\cite{FlammiaLiu}.
Recent preparations of two-dimensional graph states on Google's superconducting hardware have solved the original 2D hidden linear function problem of Ref.~\cite{ShallowCircuits} with probability greater than $7/8$ for systems of $\approx 30$ qubits \cite{Google}. Choosing a cycle graph out of the 2D grid, this problem may be translated to our multiplayer triangle game construction.

\section{Conclusion}\label{sec:conclusion}
Our results show that states with symmetry-protected topological (SPT) order, known to be useful resources for nonlocal games and MBQC, are also practical for near-term demonstration of uniquely quantum mechanical behavior. We've demonstrated that these efficiently preparable states have reasonable robustness to thermal noise. Indeed, the ``most quantum'' aspects of the states survive to nonzero temperatures on finite-sized systems. This feature makes them good candidates for experimentally realizing unconditional separations between quantum and classical computation.

Also, we've shown that noisy states' performance in the triangle game can be lower bounded by their global fidelity with the cluster state, assuming the contextual measurements of the cluster state measurement strategy could be applied noiselessly. If the fidelity is greater than $7/8$, this guarantees that the state could be used to beat all classical strategies in the triangle game. This gives an experimentally-friendly measure of whether near-term states are ``quantum advantageous''. At the time of this paper, NISQ devices are close to or potentially beyond this threshold. 

There are several open questions that are beyond the scope of the current paper.
The triangle game can be played with the $\mathds{Z}_2\times\mathds{Z}_2$ SPTO of the spin-1 Haldane phase. This phase is generated by a Hamiltonian with two-body rather than the three-body interactions of the cluster phase and includes the antiferromagnetic Heisenberg model, which is of particular interest due to its physically simplistic Hamiltonian. Our Theorem \ref{thm:1} relies on dichotomic operators (i.e. the operators have eigenvalues $\pm1$), which is not the case for spin-1 operators. However, in Ref. \cite{Austin}, such operators are constructed directly through an isometry that maps the cluster state MPS tensors to those of the RG fixed point state called the AKLT state \cite{AKLT}. It would be interesting to investigate the interplay between noise and quantum advantage in this phase in a similar manner.

Finally, it would be interesting to extend our nonlocal game to models with more complicated SPTO. For example, in Ref. \cite{Roberts} it was shown that the 3D cluster state model has a nontrivial SPT phase which persists to finite temperature in the thermodynamic limit when the model is protected by a 1-form symmetry. Constructing a nonlocal game for such a model would have an inherent robustness to noise, and that may be useful to benchmark eventually fault-tolerant quantum devices.

\section*{Acknowledgments}
We thank Austin Daniel, Ben Corbett, and Manuel Mu\~{n}oz-Arias for useful discussions and help with numerics. 

This work was supported by the National Science Foundation Grant No. PHY-2310567 and STAQ Project (Grants No.PHY-2325080). Support is acknowledged from the U.S. Department of Energy, Office of Science National Quantum Information Science Research Center, Quantum Systems Accelerator (Award No. DE-SCL0000121). We thank the UNM Center for Advanced Research Computing, supported in part by the National Science Foundation, for providing the research computing resources used in this work.

\bibliography{bibliography}

%apsrev4-2.bst 2019-01-14 (MD) hand-edited version of apsrev4-1.bst
%Control: key (0)
%Control: author (8) initials jnrlst
%Control: editor formatted (1) identically to author
%Control: production of article title (0) allowed
%Control: page (0) single
%Control: year (1) truncated
%Control: production of eprint (0) enabled
\begin{thebibliography}{81}%
\makeatletter
\providecommand \@ifxundefined [1]{%
 \@ifx{#1\undefined}
}%
\providecommand \@ifnum [1]{%
 \ifnum #1\expandafter \@firstoftwo
 \else \expandafter \@secondoftwo
 \fi
}%
\providecommand \@ifx [1]{%
 \ifx #1\expandafter \@firstoftwo
 \else \expandafter \@secondoftwo
 \fi
}%
\providecommand \natexlab [1]{#1}%
\providecommand \enquote  [1]{``#1''}%
\providecommand \bibnamefont  [1]{#1}%
\providecommand \bibfnamefont [1]{#1}%
\providecommand \citenamefont [1]{#1}%
\providecommand \href@noop [0]{\@secondoftwo}%
\providecommand \href [0]{\begingroup \@sanitize@url \@href}%
\providecommand \@href[1]{\@@startlink{#1}\@@href}%
\providecommand \@@href[1]{\endgroup#1\@@endlink}%
\providecommand \@sanitize@url [0]{\catcode `\\12\catcode `\$12\catcode `\&12\catcode `\#12\catcode `\^12\catcode `\_12\catcode `\%12\relax}%
\providecommand \@@startlink[1]{}%
\providecommand \@@endlink[0]{}%
\providecommand \url  [0]{\begingroup\@sanitize@url \@url }%
\providecommand \@url [1]{\endgroup\@href {#1}{\urlprefix }}%
\providecommand \urlprefix  [0]{URL }%
\providecommand \Eprint [0]{\href }%
\providecommand \doibase [0]{https://doi.org/}%
\providecommand \selectlanguage [0]{\@gobble}%
\providecommand \bibinfo  [0]{\@secondoftwo}%
\providecommand \bibfield  [0]{\@secondoftwo}%
\providecommand \translation [1]{[#1]}%
\providecommand \BibitemOpen [0]{}%
\providecommand \bibitemStop [0]{}%
\providecommand \bibitemNoStop [0]{.\EOS\space}%
\providecommand \EOS [0]{\spacefactor3000\relax}%
\providecommand \BibitemShut  [1]{\csname bibitem#1\endcsname}%
\let\auto@bib@innerbib\@empty
%</preamble>
\bibitem [{\citenamefont {Preskill}(2018)}]{Preskill2018quantumcomputingin}%
  \BibitemOpen
  \bibfield  {author} {\bibinfo {author} {\bibfnamefont {J.}~\bibnamefont {Preskill}},\ }\bibfield  {title} {\bibinfo {title} {Quantum {C}omputing in the {NISQ} era and beyond},\ }\href {https://doi.org/10.22331/q-2018-08-06-79} {\bibfield  {journal} {\bibinfo  {journal} {{Quantum}}\ }\textbf {\bibinfo {volume} {2}},\ \bibinfo {pages} {79} (\bibinfo {year} {2018})}\BibitemShut {NoStop}%
\bibitem [{\citenamefont {Bharti}\ \emph {et~al.}(2022)\citenamefont {Bharti}, \citenamefont {Cervera-Lierta}, \citenamefont {Kyaw}, \citenamefont {Haug}, \citenamefont {Alperin-Lea}, \citenamefont {Anand}, \citenamefont {Degroote}, \citenamefont {Heimonen}, \citenamefont {Kottmann}, \citenamefont {Menke}, \citenamefont {Mok}, \citenamefont {Sim}, \citenamefont {Kwek},\ and\ \citenamefont {Aspuru-Guzik}}]{RevModPhys.94.015004}%
  \BibitemOpen
  \bibfield  {author} {\bibinfo {author} {\bibfnamefont {K.}~\bibnamefont {Bharti}}, \bibinfo {author} {\bibfnamefont {A.}~\bibnamefont {Cervera-Lierta}}, \bibinfo {author} {\bibfnamefont {T.~H.}\ \bibnamefont {Kyaw}}, \bibinfo {author} {\bibfnamefont {T.}~\bibnamefont {Haug}}, \bibinfo {author} {\bibfnamefont {S.}~\bibnamefont {Alperin-Lea}}, \bibinfo {author} {\bibfnamefont {A.}~\bibnamefont {Anand}}, \bibinfo {author} {\bibfnamefont {M.}~\bibnamefont {Degroote}}, \bibinfo {author} {\bibfnamefont {H.}~\bibnamefont {Heimonen}}, \bibinfo {author} {\bibfnamefont {J.~S.}\ \bibnamefont {Kottmann}}, \bibinfo {author} {\bibfnamefont {T.}~\bibnamefont {Menke}}, \bibinfo {author} {\bibfnamefont {W.-K.}\ \bibnamefont {Mok}}, \bibinfo {author} {\bibfnamefont {S.}~\bibnamefont {Sim}}, \bibinfo {author} {\bibfnamefont {L.-C.}\ \bibnamefont {Kwek}},\ and\ \bibinfo {author} {\bibfnamefont {A.}~\bibnamefont {Aspuru-Guzik}},\ }\bibfield  {title} {\bibinfo {title} {Noisy intermediate-scale quantum algorithms},\ }\href
  {https://doi.org/10.1103/RevModPhys.94.015004} {\bibfield  {journal} {\bibinfo  {journal} {Rev. Mod. Phys.}\ }\textbf {\bibinfo {volume} {94}},\ \bibinfo {pages} {015004} (\bibinfo {year} {2022})}\BibitemShut {NoStop}%
\bibitem [{\citenamefont {Grover}(1996)}]{Grover}%
  \BibitemOpen
  \bibfield  {author} {\bibinfo {author} {\bibfnamefont {L.~K.}\ \bibnamefont {Grover}},\ }\bibfield  {title} {\bibinfo {title} {A fast quantum mechanical algorithm for database search},\ }in\ \href {https://doi.org/10.1145/237814.237866} {\emph {\bibinfo {booktitle} {Proceedings of the Twenty-Eighth Annual ACM Symposium on Theory of Computing}}},\ \bibinfo {series and number} {STOC '96}\ (\bibinfo  {publisher} {Association for Computing Machinery},\ \bibinfo {address} {New York, NY, USA},\ \bibinfo {year} {1996})\ p.\ \bibinfo {pages} {212–219}\BibitemShut {NoStop}%
\bibitem [{\citenamefont {Shor}(1997)}]{Shor}%
  \BibitemOpen
  \bibfield  {author} {\bibinfo {author} {\bibfnamefont {P.~W.}\ \bibnamefont {Shor}},\ }\bibfield  {title} {\bibinfo {title} {Polynomial time algorithms for prime factorization and discrete logarithms on a quantum computer},\ }\href {https://doi.org/10.1137/S0097539795293172} {\bibfield  {journal} {\bibinfo  {journal} {SIAM J. Sci. Statist. Comput.}\ }\textbf {\bibinfo {volume} {26}},\ \bibinfo {pages} {1484} (\bibinfo {year} {1997})}\BibitemShut {NoStop}%
\bibitem [{\citenamefont {Raussendorf}\ and\ \citenamefont {Briegel}(2001)}]{RaussendorfBriegel}%
  \BibitemOpen
  \bibfield  {author} {\bibinfo {author} {\bibfnamefont {R.}~\bibnamefont {Raussendorf}}\ and\ \bibinfo {author} {\bibfnamefont {H.~J.}\ \bibnamefont {Briegel}},\ }\bibfield  {title} {\bibinfo {title} {A one-way quantum computer},\ }\href {https://doi.org/10.1103/PhysRevLett.86.5188} {\bibfield  {journal} {\bibinfo  {journal} {Phys. Rev. Lett.}\ }\textbf {\bibinfo {volume} {86}},\ \bibinfo {pages} {5188} (\bibinfo {year} {2001})}\BibitemShut {NoStop}%
\bibitem [{\citenamefont {Raussendorf}\ \emph {et~al.}(2003)\citenamefont {Raussendorf}, \citenamefont {Browne},\ and\ \citenamefont {Briegel}}]{Raussendorf2003}%
  \BibitemOpen
  \bibfield  {author} {\bibinfo {author} {\bibfnamefont {R.}~\bibnamefont {Raussendorf}}, \bibinfo {author} {\bibfnamefont {D.~E.}\ \bibnamefont {Browne}},\ and\ \bibinfo {author} {\bibfnamefont {H.~J.}\ \bibnamefont {Briegel}},\ }\bibfield  {title} {\bibinfo {title} {Measurement-based quantum computation on cluster states},\ }\href {https://doi.org/10.1103/PhysRevA.68.022312} {\bibfield  {journal} {\bibinfo  {journal} {Phys. Rev. A}\ }\textbf {\bibinfo {volume} {68}},\ \bibinfo {pages} {022312} (\bibinfo {year} {2003})}\BibitemShut {NoStop}%
\bibitem [{\citenamefont {Miyake}(2010)}]{Miyake}%
  \BibitemOpen
  \bibfield  {author} {\bibinfo {author} {\bibfnamefont {A.}~\bibnamefont {Miyake}},\ }\bibfield  {title} {\bibinfo {title} {Quantum computation on the edge of a symmetry-protected topological order},\ }\href {https://doi.org/10.1103/PhysRevLett.105.040501} {\bibfield  {journal} {\bibinfo  {journal} {Phys. Rev. Lett.}\ }\textbf {\bibinfo {volume} {105}},\ \bibinfo {pages} {040501} (\bibinfo {year} {2010})}\BibitemShut {NoStop}%
\bibitem [{\citenamefont {Bartlett}\ \emph {et~al.}(2010)\citenamefont {Bartlett}, \citenamefont {Brennen}, \citenamefont {Miyake},\ and\ \citenamefont {Renes}}]{QCinHaldanePhase}%
  \BibitemOpen
  \bibfield  {author} {\bibinfo {author} {\bibfnamefont {S.~D.}\ \bibnamefont {Bartlett}}, \bibinfo {author} {\bibfnamefont {G.~K.}\ \bibnamefont {Brennen}}, \bibinfo {author} {\bibfnamefont {A.}~\bibnamefont {Miyake}},\ and\ \bibinfo {author} {\bibfnamefont {J.~M.}\ \bibnamefont {Renes}},\ }\bibfield  {title} {\bibinfo {title} {Quantum computational renormalization in the haldane phase},\ }\href {https://doi.org/10.1103/PhysRevLett.105.110502} {\bibfield  {journal} {\bibinfo  {journal} {Phys. Rev. Lett.}\ }\textbf {\bibinfo {volume} {105}},\ \bibinfo {pages} {110502} (\bibinfo {year} {2010})}\BibitemShut {NoStop}%
\bibitem [{\citenamefont {Else}\ \emph {et~al.}(2012{\natexlab{a}})\citenamefont {Else}, \citenamefont {Schwarz}, \citenamefont {Bartlett},\ and\ \citenamefont {Doherty}}]{SPTforMBQC}%
  \BibitemOpen
  \bibfield  {author} {\bibinfo {author} {\bibfnamefont {D.~V.}\ \bibnamefont {Else}}, \bibinfo {author} {\bibfnamefont {I.}~\bibnamefont {Schwarz}}, \bibinfo {author} {\bibfnamefont {S.~D.}\ \bibnamefont {Bartlett}},\ and\ \bibinfo {author} {\bibfnamefont {A.~C.}\ \bibnamefont {Doherty}},\ }\bibfield  {title} {\bibinfo {title} {Symmetry-protected phases for measurement-based quantum computation},\ }\href {https://doi.org/10.1103/PhysRevLett.108.240505} {\bibfield  {journal} {\bibinfo  {journal} {Phys. Rev. Lett.}\ }\textbf {\bibinfo {volume} {108}},\ \bibinfo {pages} {240505} (\bibinfo {year} {2012}{\natexlab{a}})}\BibitemShut {NoStop}%
\bibitem [{\citenamefont {Else}\ \emph {et~al.}(2012{\natexlab{b}})\citenamefont {Else}, \citenamefont {Bartlett},\ and\ \citenamefont {Doherty}}]{SPTinMBQCGS}%
  \BibitemOpen
  \bibfield  {author} {\bibinfo {author} {\bibfnamefont {D.~V.}\ \bibnamefont {Else}}, \bibinfo {author} {\bibfnamefont {S.~D.}\ \bibnamefont {Bartlett}},\ and\ \bibinfo {author} {\bibfnamefont {A.~C.}\ \bibnamefont {Doherty}},\ }\bibfield  {title} {\bibinfo {title} {Symmetry protection of measurement-based quantum computation in ground states},\ }\href {https://doi.org/10.1088/1367-2630/14/11/113016} {\bibfield  {journal} {\bibinfo  {journal} {New Journal of Physics}\ }\textbf {\bibinfo {volume} {14}},\ \bibinfo {pages} {113016} (\bibinfo {year} {2012}{\natexlab{b}})}\BibitemShut {NoStop}%
\bibitem [{\citenamefont {Miller}\ and\ \citenamefont {Miyake}(2015)}]{Miller}%
  \BibitemOpen
  \bibfield  {author} {\bibinfo {author} {\bibfnamefont {J.}~\bibnamefont {Miller}}\ and\ \bibinfo {author} {\bibfnamefont {A.}~\bibnamefont {Miyake}},\ }\bibfield  {title} {\bibinfo {title} {Resource quality of a symmetry-protected topologically ordered phase for quantum computation},\ }\href {https://doi.org/10.1103/PhysRevLett.114.120506} {\bibfield  {journal} {\bibinfo  {journal} {Phys. Rev. Lett.}\ }\textbf {\bibinfo {volume} {114}},\ \bibinfo {pages} {120506} (\bibinfo {year} {2015})}\BibitemShut {NoStop}%
\bibitem [{\citenamefont {Miller}\ and\ \citenamefont {Miyake}(2016)}]{Miller2016}%
  \BibitemOpen
  \bibfield  {author} {\bibinfo {author} {\bibfnamefont {J.}~\bibnamefont {Miller}}\ and\ \bibinfo {author} {\bibfnamefont {A.}~\bibnamefont {Miyake}},\ }\bibfield  {title} {\bibinfo {title} {Hierarchy of universal entanglement in 2d measurement-based quantum computation},\ }\href {https://doi.org/10.1038/npjqi.2016.36} {\bibfield  {journal} {\bibinfo  {journal} {npj Quantum Information}\ }\textbf {\bibinfo {volume} {2}},\ \bibinfo {pages} {16036} (\bibinfo {year} {2016})}\BibitemShut {NoStop}%
\bibitem [{\citenamefont {Raussendorf}\ \emph {et~al.}(2017)\citenamefont {Raussendorf}, \citenamefont {Wang}, \citenamefont {Prakash}, \citenamefont {Wei},\ and\ \citenamefont {Stephen}}]{1DSPTforComputation}%
  \BibitemOpen
  \bibfield  {author} {\bibinfo {author} {\bibfnamefont {R.}~\bibnamefont {Raussendorf}}, \bibinfo {author} {\bibfnamefont {D.-S.}\ \bibnamefont {Wang}}, \bibinfo {author} {\bibfnamefont {A.}~\bibnamefont {Prakash}}, \bibinfo {author} {\bibfnamefont {T.-C.}\ \bibnamefont {Wei}},\ and\ \bibinfo {author} {\bibfnamefont {D.~T.}\ \bibnamefont {Stephen}},\ }\bibfield  {title} {\bibinfo {title} {Symmetry-protected topological phases with uniform computational power in one dimension},\ }\href {https://doi.org/10.1103/PhysRevA.96.012302} {\bibfield  {journal} {\bibinfo  {journal} {Phys. Rev. A}\ }\textbf {\bibinfo {volume} {96}},\ \bibinfo {pages} {012302} (\bibinfo {year} {2017})}\BibitemShut {NoStop}%
\bibitem [{\citenamefont {Stephen}\ \emph {et~al.}(2017)\citenamefont {Stephen}, \citenamefont {Wang}, \citenamefont {Prakash}, \citenamefont {Wei},\ and\ \citenamefont {Raussendorf}}]{SPTComputationalPower}%
  \BibitemOpen
  \bibfield  {author} {\bibinfo {author} {\bibfnamefont {D.~T.}\ \bibnamefont {Stephen}}, \bibinfo {author} {\bibfnamefont {D.-S.}\ \bibnamefont {Wang}}, \bibinfo {author} {\bibfnamefont {A.}~\bibnamefont {Prakash}}, \bibinfo {author} {\bibfnamefont {T.-C.}\ \bibnamefont {Wei}},\ and\ \bibinfo {author} {\bibfnamefont {R.}~\bibnamefont {Raussendorf}},\ }\bibfield  {title} {\bibinfo {title} {Computational power of symmetry-protected topological phases},\ }\href {https://doi.org/10.1103/PhysRevLett.119.010504} {\bibfield  {journal} {\bibinfo  {journal} {Phys. Rev. Lett.}\ }\textbf {\bibinfo {volume} {119}},\ \bibinfo {pages} {010504} (\bibinfo {year} {2017})}\BibitemShut {NoStop}%
\bibitem [{\citenamefont {Devakul}\ and\ \citenamefont {Williamson}(2018)}]{FractalMBQC}%
  \BibitemOpen
  \bibfield  {author} {\bibinfo {author} {\bibfnamefont {T.}~\bibnamefont {Devakul}}\ and\ \bibinfo {author} {\bibfnamefont {D.~J.}\ \bibnamefont {Williamson}},\ }\bibfield  {title} {\bibinfo {title} {Universal quantum computation using fractal symmetry-protected cluster phases},\ }\href {https://doi.org/10.1103/PhysRevA.98.022332} {\bibfield  {journal} {\bibinfo  {journal} {Phys. Rev. A}\ }\textbf {\bibinfo {volume} {98}},\ \bibinfo {pages} {022332} (\bibinfo {year} {2018})}\BibitemShut {NoStop}%
\bibitem [{\citenamefont {Raussendorf}\ \emph {et~al.}(2019)\citenamefont {Raussendorf}, \citenamefont {Okay}, \citenamefont {Wang}, \citenamefont {Stephen},\ and\ \citenamefont {Nautrup}}]{ComputationalPhaseofMatter}%
  \BibitemOpen
  \bibfield  {author} {\bibinfo {author} {\bibfnamefont {R.}~\bibnamefont {Raussendorf}}, \bibinfo {author} {\bibfnamefont {C.}~\bibnamefont {Okay}}, \bibinfo {author} {\bibfnamefont {D.-S.}\ \bibnamefont {Wang}}, \bibinfo {author} {\bibfnamefont {D.~T.}\ \bibnamefont {Stephen}},\ and\ \bibinfo {author} {\bibfnamefont {H.~P.}\ \bibnamefont {Nautrup}},\ }\bibfield  {title} {\bibinfo {title} {Computationally universal phase of quantum matter},\ }\href {https://doi.org/10.1103/PhysRevLett.122.090501} {\bibfield  {journal} {\bibinfo  {journal} {Phys. Rev. Lett.}\ }\textbf {\bibinfo {volume} {122}},\ \bibinfo {pages} {090501} (\bibinfo {year} {2019})}\BibitemShut {NoStop}%
\bibitem [{\citenamefont {Stephen}\ \emph {et~al.}(2019)\citenamefont {Stephen}, \citenamefont {Nautrup}, \citenamefont {Bermejo-Vega}, \citenamefont {Eisert},\ and\ \citenamefont {Raussendorf}}]{Stephen2019subsystem}%
  \BibitemOpen
  \bibfield  {author} {\bibinfo {author} {\bibfnamefont {D.~T.}\ \bibnamefont {Stephen}}, \bibinfo {author} {\bibfnamefont {H.~P.}\ \bibnamefont {Nautrup}}, \bibinfo {author} {\bibfnamefont {J.}~\bibnamefont {Bermejo-Vega}}, \bibinfo {author} {\bibfnamefont {J.}~\bibnamefont {Eisert}},\ and\ \bibinfo {author} {\bibfnamefont {R.}~\bibnamefont {Raussendorf}},\ }\bibfield  {title} {\bibinfo {title} {Subsystem symmetries, quantum cellular automata, and computational phases of quantum matter},\ }\href {https://doi.org/10.22331/q-2019-05-20-142} {\bibfield  {journal} {\bibinfo  {journal} {{Quantum}}\ }\textbf {\bibinfo {volume} {3}},\ \bibinfo {pages} {142} (\bibinfo {year} {2019})}\BibitemShut {NoStop}%
\bibitem [{\citenamefont {Daniel}\ \emph {et~al.}(2020)\citenamefont {Daniel}, \citenamefont {Alexander},\ and\ \citenamefont {Miyake}}]{Daniel2020computational}%
  \BibitemOpen
  \bibfield  {author} {\bibinfo {author} {\bibfnamefont {A.~K.}\ \bibnamefont {Daniel}}, \bibinfo {author} {\bibfnamefont {R.~N.}\ \bibnamefont {Alexander}},\ and\ \bibinfo {author} {\bibfnamefont {A.}~\bibnamefont {Miyake}},\ }\bibfield  {title} {\bibinfo {title} {Computational universality of symmetry-protected topologically ordered cluster phases on 2{D} {A}rchimedean lattices},\ }\href {https://doi.org/10.22331/q-2020-02-10-228} {\bibfield  {journal} {\bibinfo  {journal} {{Quantum}}\ }\textbf {\bibinfo {volume} {4}},\ \bibinfo {pages} {228} (\bibinfo {year} {2020})}\BibitemShut {NoStop}%
\bibitem [{\citenamefont {Pollmann}\ \emph {et~al.}(2010)\citenamefont {Pollmann}, \citenamefont {Turner}, \citenamefont {Berg},\ and\ \citenamefont {Oshikawa}}]{PollmannEntanglement}%
  \BibitemOpen
  \bibfield  {author} {\bibinfo {author} {\bibfnamefont {F.}~\bibnamefont {Pollmann}}, \bibinfo {author} {\bibfnamefont {A.~M.}\ \bibnamefont {Turner}}, \bibinfo {author} {\bibfnamefont {E.}~\bibnamefont {Berg}},\ and\ \bibinfo {author} {\bibfnamefont {M.}~\bibnamefont {Oshikawa}},\ }\bibfield  {title} {\bibinfo {title} {Entanglement spectrum of a topological phase in one dimension},\ }\href {https://doi.org/10.1103/PhysRevB.81.064439} {\bibfield  {journal} {\bibinfo  {journal} {Phys. Rev. B}\ }\textbf {\bibinfo {volume} {81}},\ \bibinfo {pages} {064439} (\bibinfo {year} {2010})}\BibitemShut {NoStop}%
\bibitem [{\citenamefont {Pollmann}\ \emph {et~al.}(2012)\citenamefont {Pollmann}, \citenamefont {Berg}, \citenamefont {Turner},\ and\ \citenamefont {Oshikawa}}]{PollmannSPTO}%
  \BibitemOpen
  \bibfield  {author} {\bibinfo {author} {\bibfnamefont {F.}~\bibnamefont {Pollmann}}, \bibinfo {author} {\bibfnamefont {E.}~\bibnamefont {Berg}}, \bibinfo {author} {\bibfnamefont {A.~M.}\ \bibnamefont {Turner}},\ and\ \bibinfo {author} {\bibfnamefont {M.}~\bibnamefont {Oshikawa}},\ }\bibfield  {title} {\bibinfo {title} {Symmetry protection of topological phases in one-dimensional quantum spin systems},\ }\href {https://doi.org/10.1103/PhysRevB.85.075125} {\bibfield  {journal} {\bibinfo  {journal} {Phys. Rev. B}\ }\textbf {\bibinfo {volume} {85}},\ \bibinfo {pages} {075125} (\bibinfo {year} {2012})}\BibitemShut {NoStop}%
\bibitem [{\citenamefont {Chen}\ \emph {et~al.}(2013)\citenamefont {Chen}, \citenamefont {Gu}, \citenamefont {Liu},\ and\ \citenamefont {Wen}}]{Homology}%
  \BibitemOpen
  \bibfield  {author} {\bibinfo {author} {\bibfnamefont {X.}~\bibnamefont {Chen}}, \bibinfo {author} {\bibfnamefont {Z.-C.}\ \bibnamefont {Gu}}, \bibinfo {author} {\bibfnamefont {Z.-X.}\ \bibnamefont {Liu}},\ and\ \bibinfo {author} {\bibfnamefont {X.-G.}\ \bibnamefont {Wen}},\ }\bibfield  {title} {\bibinfo {title} {Symmetry protected topological orders and the group cohomology of their symmetry group},\ }\href {https://doi.org/10.1103/PhysRevB.87.155114} {\bibfield  {journal} {\bibinfo  {journal} {Phys. Rev. B}\ }\textbf {\bibinfo {volume} {87}},\ \bibinfo {pages} {155114} (\bibinfo {year} {2013})}\BibitemShut {NoStop}%
\bibitem [{\citenamefont {Zeng}\ \emph {et~al.}(2019)\citenamefont {Zeng}, \citenamefont {Chen}, \citenamefont {Zhou},\ and\ \citenamefont {Wen}}]{QImeetsQM}%
  \BibitemOpen
  \bibfield  {author} {\bibinfo {author} {\bibfnamefont {B.}~\bibnamefont {Zeng}}, \bibinfo {author} {\bibfnamefont {X.}~\bibnamefont {Chen}}, \bibinfo {author} {\bibfnamefont {D.}~\bibnamefont {Zhou}},\ and\ \bibinfo {author} {\bibfnamefont {X.}~\bibnamefont {Wen}},\ }\href {https://books.google.com/books?id=BPnPxgEACAAJ} {\emph {\bibinfo {title} {Quantum Information Meets Quantum Matter: From Quantum Entanglement to Topological Phases of Many-Body Systems}}},\ Quantum science and technology\ (\bibinfo  {publisher} {Springer New York},\ \bibinfo {year} {2019})\BibitemShut {NoStop}%
\bibitem [{\citenamefont {Cleve}\ and\ \citenamefont {Mittal}(2014)}]{NonlocalgamesBook}%
  \BibitemOpen
  \bibfield  {author} {\bibinfo {author} {\bibfnamefont {R.}~\bibnamefont {Cleve}}\ and\ \bibinfo {author} {\bibfnamefont {R.}~\bibnamefont {Mittal}},\ }\bibfield  {title} {\bibinfo {title} {Characterization of binary constraint system games},\ }in\ \href {https://doi.org/10.1007/978-3-662-43948-7_27} {\emph {\bibinfo {booktitle} {Automata, Languages, and Programming}}},\ \bibinfo {editor} {edited by\ \bibinfo {editor} {\bibfnamefont {J.}~\bibnamefont {Esparza}}, \bibinfo {editor} {\bibfnamefont {P.}~\bibnamefont {Fraigniaud}}, \bibinfo {editor} {\bibfnamefont {T.}~\bibnamefont {Husfeldt}},\ and\ \bibinfo {editor} {\bibfnamefont {E.}~\bibnamefont {Koutsoupias}}}\ (\bibinfo  {publisher} {Springer Berlin Heidelberg},\ \bibinfo {address} {Berlin, Heidelberg},\ \bibinfo {year} {2014})\ pp.\ \bibinfo {pages} {320--331}\BibitemShut {NoStop}%
\bibitem [{\citenamefont {Brunner}\ \emph {et~al.}(2014)\citenamefont {Brunner}, \citenamefont {Cavalcanti}, \citenamefont {Pironio}, \citenamefont {Scarani},\ and\ \citenamefont {Wehner}}]{NonlocalgamesPaper}%
  \BibitemOpen
  \bibfield  {author} {\bibinfo {author} {\bibfnamefont {N.}~\bibnamefont {Brunner}}, \bibinfo {author} {\bibfnamefont {D.}~\bibnamefont {Cavalcanti}}, \bibinfo {author} {\bibfnamefont {S.}~\bibnamefont {Pironio}}, \bibinfo {author} {\bibfnamefont {V.}~\bibnamefont {Scarani}},\ and\ \bibinfo {author} {\bibfnamefont {S.}~\bibnamefont {Wehner}},\ }\bibfield  {title} {\bibinfo {title} {Bell nonlocality},\ }\href {https://doi.org/10.1103/RevModPhys.86.419} {\bibfield  {journal} {\bibinfo  {journal} {Rev. Mod. Phys.}\ }\textbf {\bibinfo {volume} {86}},\ \bibinfo {pages} {419} (\bibinfo {year} {2014})}\BibitemShut {NoStop}%
\bibitem [{\citenamefont {Bell}(1964)}]{Bell}%
  \BibitemOpen
  \bibfield  {author} {\bibinfo {author} {\bibfnamefont {J.~S.}\ \bibnamefont {Bell}},\ }\bibfield  {title} {\bibinfo {title} {On the {E}instein {P}odolsky {R}osen paradox},\ }\href {https://doi.org/10.1103/PhysicsPhysiqueFizika.1.195} {\bibfield  {journal} {\bibinfo  {journal} {Physics Physique Fizika}\ }\textbf {\bibinfo {volume} {1}},\ \bibinfo {pages} {195} (\bibinfo {year} {1964})}\BibitemShut {NoStop}%
\bibitem [{\citenamefont {Anders}\ and\ \citenamefont {Browne}(2009)}]{nonlocalgamesMBQC}%
  \BibitemOpen
  \bibfield  {author} {\bibinfo {author} {\bibfnamefont {J.}~\bibnamefont {Anders}}\ and\ \bibinfo {author} {\bibfnamefont {D.~E.}\ \bibnamefont {Browne}},\ }\bibfield  {title} {\bibinfo {title} {Computational power of correlations},\ }\href {https://doi.org/10.1103/PhysRevLett.102.050502} {\bibfield  {journal} {\bibinfo  {journal} {Phys. Rev. Lett.}\ }\textbf {\bibinfo {volume} {102}},\ \bibinfo {pages} {050502} (\bibinfo {year} {2009})}\BibitemShut {NoStop}%
\bibitem [{\citenamefont {KOCHEN}\ and\ \citenamefont {SPECKER}(1967)}]{KS1967}%
  \BibitemOpen
  \bibfield  {author} {\bibinfo {author} {\bibfnamefont {S.}~\bibnamefont {KOCHEN}}\ and\ \bibinfo {author} {\bibfnamefont {E.~P.}\ \bibnamefont {SPECKER}},\ }\bibfield  {title} {\bibinfo {title} {The problem of hidden variables in quantum mechanics},\ }\href {http://www.jstor.org/stable/24902153} {\bibfield  {journal} {\bibinfo  {journal} {Journal of Mathematics and Mechanics}\ }\textbf {\bibinfo {volume} {17}},\ \bibinfo {pages} {59} (\bibinfo {year} {1967})}\BibitemShut {NoStop}%
\bibitem [{\citenamefont {Spekkens}(2005)}]{Spekkens2005}%
  \BibitemOpen
  \bibfield  {author} {\bibinfo {author} {\bibfnamefont {R.~W.}\ \bibnamefont {Spekkens}},\ }\bibfield  {title} {\bibinfo {title} {Contextuality for preparations, transformations, and unsharp measurements},\ }\href {https://doi.org/10.1103/PhysRevA.71.052108} {\bibfield  {journal} {\bibinfo  {journal} {Phys. Rev. A}\ }\textbf {\bibinfo {volume} {71}},\ \bibinfo {pages} {052108} (\bibinfo {year} {2005})}\BibitemShut {NoStop}%
\bibitem [{\citenamefont {Bulchandani}\ \emph {et~al.}(2023{\natexlab{a}})\citenamefont {Bulchandani}, \citenamefont {Burnell},\ and\ \citenamefont {Sondhi}}]{Bulchandani3}%
  \BibitemOpen
  \bibfield  {author} {\bibinfo {author} {\bibfnamefont {V.~B.}\ \bibnamefont {Bulchandani}}, \bibinfo {author} {\bibfnamefont {F.~J.}\ \bibnamefont {Burnell}},\ and\ \bibinfo {author} {\bibfnamefont {S.~L.}\ \bibnamefont {Sondhi}},\ }\bibfield  {title} {\bibinfo {title} {A multiplayer multiteam nonlocal game for the toric code},\ }\href {https://doi.org/10.1103/PhysRevB.107.035409} {\bibfield  {journal} {\bibinfo  {journal} {Phys. Rev. B}\ }\textbf {\bibinfo {volume} {107}},\ \bibinfo {pages} {035409} (\bibinfo {year} {2023}{\natexlab{a}})}\BibitemShut {NoStop}%
\bibitem [{\citenamefont {Bulchandani}\ \emph {et~al.}(2023{\natexlab{b}})\citenamefont {Bulchandani}, \citenamefont {Burnell},\ and\ \citenamefont {Sondhi}}]{Bulchandani2}%
  \BibitemOpen
  \bibfield  {author} {\bibinfo {author} {\bibfnamefont {V.~B.}\ \bibnamefont {Bulchandani}}, \bibinfo {author} {\bibfnamefont {F.~J.}\ \bibnamefont {Burnell}},\ and\ \bibinfo {author} {\bibfnamefont {S.~L.}\ \bibnamefont {Sondhi}},\ }\bibfield  {title} {\bibinfo {title} {Playing nonlocal games with phases of quantum matter},\ }\href {https://doi.org/10.1103/PhysRevB.107.045412} {\bibfield  {journal} {\bibinfo  {journal} {Phys. Rev. B}\ }\textbf {\bibinfo {volume} {107}},\ \bibinfo {pages} {045412} (\bibinfo {year} {2023}{\natexlab{b}})}\BibitemShut {NoStop}%
\bibitem [{\citenamefont {Hart}\ \emph {et~al.}(2025{\natexlab{a}})\citenamefont {Hart}, \citenamefont {Stephen}, \citenamefont {Wickenden},\ and\ \citenamefont {Nandkishore}}]{Hart}%
  \BibitemOpen
  \bibfield  {author} {\bibinfo {author} {\bibfnamefont {O.}~\bibnamefont {Hart}}, \bibinfo {author} {\bibfnamefont {D.~T.}\ \bibnamefont {Stephen}}, \bibinfo {author} {\bibfnamefont {E.}~\bibnamefont {Wickenden}},\ and\ \bibinfo {author} {\bibfnamefont {R.}~\bibnamefont {Nandkishore}},\ }\href@noop {} {\bibinfo {title} {Many-body contextuality and self-testing quantum matter via nonlocal games}} (\bibinfo {year} {2025}{\natexlab{a}}),\ \Eprint {https://arxiv.org/abs/2512.16886} {arXiv:2512.16886 [quant-ph]} \BibitemShut {NoStop}%
\bibitem [{\citenamefont {Zhao}\ \emph {et~al.}(2025)\citenamefont {Zhao}, \citenamefont {Liew}, \citenamefont {Ho}, \citenamefont {Liu},\ and\ \citenamefont {Bulchandani}}]{Bulchandani}%
  \BibitemOpen
  \bibfield  {author} {\bibinfo {author} {\bibfnamefont {W.}~\bibnamefont {Zhao}}, \bibinfo {author} {\bibfnamefont {H.~W.~S.}\ \bibnamefont {Liew}}, \bibinfo {author} {\bibfnamefont {W.~W.}\ \bibnamefont {Ho}}, \bibinfo {author} {\bibfnamefont {C.}~\bibnamefont {Liu}},\ and\ \bibinfo {author} {\bibfnamefont {V.~B.}\ \bibnamefont {Bulchandani}},\ }\href@noop {} {\bibinfo {title} {Scalable tests of quantum contextuality from stabilizer-testing nonlocal games}} (\bibinfo {year} {2025}),\ \Eprint {https://arxiv.org/abs/2512.16654} {arXiv:2512.16654 [quant-ph]} \BibitemShut {NoStop}%
\bibitem [{\citenamefont {Hart}\ \emph {et~al.}(2025{\natexlab{b}})\citenamefont {Hart}, \citenamefont {Stephen}, \citenamefont {Williamson}, \citenamefont {Foss-Feig},\ and\ \citenamefont {Nandkishore}}]{ToricCodeNonlocalGame}%
  \BibitemOpen
  \bibfield  {author} {\bibinfo {author} {\bibfnamefont {O.}~\bibnamefont {Hart}}, \bibinfo {author} {\bibfnamefont {D.~T.}\ \bibnamefont {Stephen}}, \bibinfo {author} {\bibfnamefont {D.~J.}\ \bibnamefont {Williamson}}, \bibinfo {author} {\bibfnamefont {M.}~\bibnamefont {Foss-Feig}},\ and\ \bibinfo {author} {\bibfnamefont {R.}~\bibnamefont {Nandkishore}},\ }\bibfield  {title} {\bibinfo {title} {Playing nonlocal games across a topological phase transition on a quantum computer},\ }\href {https://doi.org/10.1103/PhysRevLett.134.130602} {\bibfield  {journal} {\bibinfo  {journal} {Phys. Rev. Lett.}\ }\textbf {\bibinfo {volume} {134}},\ \bibinfo {pages} {130602} (\bibinfo {year} {2025}{\natexlab{b}})}\BibitemShut {NoStop}%
\bibitem [{\citenamefont {Hart}\ \emph {et~al.}(2025{\natexlab{c}})\citenamefont {Hart}, \citenamefont {Stephen}, \citenamefont {Williamson},\ and\ \citenamefont {Nandkishore}}]{Braidingforthewin}%
  \BibitemOpen
  \bibfield  {author} {\bibinfo {author} {\bibfnamefont {O.}~\bibnamefont {Hart}}, \bibinfo {author} {\bibfnamefont {D.~T.}\ \bibnamefont {Stephen}}, \bibinfo {author} {\bibfnamefont {D.~J.}\ \bibnamefont {Williamson}},\ and\ \bibinfo {author} {\bibfnamefont {R.}~\bibnamefont {Nandkishore}},\ }\bibfield  {title} {\bibinfo {title} {Braiding for the win: Harnessing braiding statistics in topological states to play quantum games},\ }\href {https://doi.org/10.1103/832w-618z} {\bibfield  {journal} {\bibinfo  {journal} {Phys. Rev. B}\ }\textbf {\bibinfo {volume} {112}},\ \bibinfo {pages} {165132} (\bibinfo {year} {2025}{\natexlab{c}})}\BibitemShut {NoStop}%
\bibitem [{\citenamefont {Daniel}\ and\ \citenamefont {Miyake}(2021)}]{Austin}%
  \BibitemOpen
  \bibfield  {author} {\bibinfo {author} {\bibfnamefont {A.}~\bibnamefont {Daniel}}\ and\ \bibinfo {author} {\bibfnamefont {A.}~\bibnamefont {Miyake}},\ }\bibfield  {title} {\bibinfo {title} {Quantum computational advantage with string order parameters of one-dimensional symmetry-protected topological order},\ }\href {https://doi.org/10.1103/PhysRevLett.126.090505} {\bibfield  {journal} {\bibinfo  {journal} {Phys. Rev. Lett.}\ }\textbf {\bibinfo {volume} {126}},\ \bibinfo {pages} {090505} (\bibinfo {year} {(2021)})}\BibitemShut {NoStop}%
\bibitem [{\citenamefont {Barrett}\ \emph {et~al.}(2007)\citenamefont {Barrett}, \citenamefont {Caves}, \citenamefont {Eastin}, \citenamefont {Elliott},\ and\ \citenamefont {Pironio}}]{modelingPaulimeas}%
  \BibitemOpen
  \bibfield  {author} {\bibinfo {author} {\bibfnamefont {J.}~\bibnamefont {Barrett}}, \bibinfo {author} {\bibfnamefont {C.~M.}\ \bibnamefont {Caves}}, \bibinfo {author} {\bibfnamefont {B.}~\bibnamefont {Eastin}}, \bibinfo {author} {\bibfnamefont {M.~B.}\ \bibnamefont {Elliott}},\ and\ \bibinfo {author} {\bibfnamefont {S.}~\bibnamefont {Pironio}},\ }\bibfield  {title} {\bibinfo {title} {Modeling pauli measurements on graph states with nearest-neighbor classical communication},\ }\href {https://doi.org/10.1103/PhysRevA.75.012103} {\bibfield  {journal} {\bibinfo  {journal} {Phys. Rev. A}\ }\textbf {\bibinfo {volume} {75}},\ \bibinfo {pages} {012103} (\bibinfo {year} {2007})}\BibitemShut {NoStop}%
\bibitem [{\citenamefont {Bravyi}\ \emph {et~al.}(2018)\citenamefont {Bravyi}, \citenamefont {Gosset},\ and\ \citenamefont {König}}]{ShallowCircuits}%
  \BibitemOpen
  \bibfield  {author} {\bibinfo {author} {\bibfnamefont {S.}~\bibnamefont {Bravyi}}, \bibinfo {author} {\bibfnamefont {D.}~\bibnamefont {Gosset}},\ and\ \bibinfo {author} {\bibfnamefont {R.}~\bibnamefont {König}},\ }\bibfield  {title} {\bibinfo {title} {Quantum advantage with shallow circuits},\ }\href {https://doi.org/10.1126/science.aar3106} {\bibfield  {journal} {\bibinfo  {journal} {Science}\ }\textbf {\bibinfo {volume} {362}},\ \bibinfo {pages} {308} (\bibinfo {year} {2018})}\BibitemShut {NoStop}%
\bibitem [{\citenamefont {Bravyi}\ \emph {et~al.}(2020)\citenamefont {Bravyi}, \citenamefont {Gosset}, \citenamefont {K{\"o}nig},\ and\ \citenamefont {Tomamichel}}]{Bravyi2020}%
  \BibitemOpen
  \bibfield  {author} {\bibinfo {author} {\bibfnamefont {S.}~\bibnamefont {Bravyi}}, \bibinfo {author} {\bibfnamefont {D.}~\bibnamefont {Gosset}}, \bibinfo {author} {\bibfnamefont {R.}~\bibnamefont {K{\"o}nig}},\ and\ \bibinfo {author} {\bibfnamefont {M.}~\bibnamefont {Tomamichel}},\ }\bibfield  {title} {\bibinfo {title} {Quantum advantage with noisy shallow circuits},\ }\href {https://doi.org/10.1038/s41567-020-0948-z} {\bibfield  {journal} {\bibinfo  {journal} {Nature Physics}\ }\textbf {\bibinfo {volume} {16}},\ \bibinfo {pages} {1040} (\bibinfo {year} {2020})}\BibitemShut {NoStop}%
\bibitem [{\citenamefont {den Nijs}\ and\ \citenamefont {Rommelse}(1989)}]{denNijs}%
  \BibitemOpen
  \bibfield  {author} {\bibinfo {author} {\bibfnamefont {M.}~\bibnamefont {den Nijs}}\ and\ \bibinfo {author} {\bibfnamefont {K.}~\bibnamefont {Rommelse}},\ }\bibfield  {title} {\bibinfo {title} {Preroughening transitions in crystal surfaces and valence-bond phases in quantum spin chains},\ }\href {https://doi.org/10.1103/PhysRevB.40.4709} {\bibfield  {journal} {\bibinfo  {journal} {Phys. Rev. B}\ }\textbf {\bibinfo {volume} {40}},\ \bibinfo {pages} {4709} (\bibinfo {year} {1989})}\BibitemShut {NoStop}%
\bibitem [{\citenamefont {P\'erez-Garc\'{\i}a}\ \emph {et~al.}(2008)\citenamefont {P\'erez-Garc\'{\i}a}, \citenamefont {Wolf}, \citenamefont {Sanz}, \citenamefont {Verstraete},\ and\ \citenamefont {Cirac}}]{SOPinSpinLattice}%
  \BibitemOpen
  \bibfield  {author} {\bibinfo {author} {\bibfnamefont {D.}~\bibnamefont {P\'erez-Garc\'{\i}a}}, \bibinfo {author} {\bibfnamefont {M.~M.}\ \bibnamefont {Wolf}}, \bibinfo {author} {\bibfnamefont {M.}~\bibnamefont {Sanz}}, \bibinfo {author} {\bibfnamefont {F.}~\bibnamefont {Verstraete}},\ and\ \bibinfo {author} {\bibfnamefont {J.~I.}\ \bibnamefont {Cirac}},\ }\bibfield  {title} {\bibinfo {title} {String order and symmetries in quantum spin lattices},\ }\href {https://doi.org/10.1103/PhysRevLett.100.167202} {\bibfield  {journal} {\bibinfo  {journal} {Phys. Rev. Lett.}\ }\textbf {\bibinfo {volume} {100}},\ \bibinfo {pages} {167202} (\bibinfo {year} {2008})}\BibitemShut {NoStop}%
\bibitem [{\citenamefont {Masui}\ and\ \citenamefont {Totsuka}(2025)}]{Masui2025}%
  \BibitemOpen
  \bibfield  {author} {\bibinfo {author} {\bibfnamefont {R.}~\bibnamefont {Masui}}\ and\ \bibinfo {author} {\bibfnamefont {K.}~\bibnamefont {Totsuka}},\ }\bibfield  {title} {\bibinfo {title} {Computational characterization of symmetry-protected topological phases in open quantum systems},\ }\href {https://doi.org/10.1103/9bbt-x32g} {\bibfield  {journal} {\bibinfo  {journal} {Phys. Rev. B}\ }\textbf {\bibinfo {volume} {112}},\ \bibinfo {pages} {085108} (\bibinfo {year} {2025})}\BibitemShut {NoStop}%
\bibitem [{\citenamefont {Zhang}\ \emph {et~al.}(2025)\citenamefont {Zhang}, \citenamefont {Agrawal},\ and\ \citenamefont {Vijay}}]{MixedSPT}%
  \BibitemOpen
  \bibfield  {author} {\bibinfo {author} {\bibfnamefont {Z.}~\bibnamefont {Zhang}}, \bibinfo {author} {\bibfnamefont {U.}~\bibnamefont {Agrawal}},\ and\ \bibinfo {author} {\bibfnamefont {S.}~\bibnamefont {Vijay}},\ }\bibfield  {title} {\bibinfo {title} {Quantum communication and mixed-state order in decohered symmetry-protected topological states},\ }\href {https://doi.org/10.1103/PhysRevB.111.115141} {\bibfield  {journal} {\bibinfo  {journal} {Phys. Rev. B}\ }\textbf {\bibinfo {volume} {111}},\ \bibinfo {pages} {115141} (\bibinfo {year} {2025})}\BibitemShut {NoStop}%
\bibitem [{\citenamefont {Lessa}\ \emph {et~al.}(2025)\citenamefont {Lessa}, \citenamefont {Ma}, \citenamefont {Zhang}, \citenamefont {Bi}, \citenamefont {Cheng},\ and\ \citenamefont {Wang}}]{Lessa2025}%
  \BibitemOpen
  \bibfield  {author} {\bibinfo {author} {\bibfnamefont {L.~A.}\ \bibnamefont {Lessa}}, \bibinfo {author} {\bibfnamefont {R.}~\bibnamefont {Ma}}, \bibinfo {author} {\bibfnamefont {J.-H.}\ \bibnamefont {Zhang}}, \bibinfo {author} {\bibfnamefont {Z.}~\bibnamefont {Bi}}, \bibinfo {author} {\bibfnamefont {M.}~\bibnamefont {Cheng}},\ and\ \bibinfo {author} {\bibfnamefont {C.}~\bibnamefont {Wang}},\ }\bibfield  {title} {\bibinfo {title} {Strong-to-weak spontaneous symmetry breaking in mixed quantum states},\ }\href {https://doi.org/10.1103/PRXQuantum.6.010344} {\bibfield  {journal} {\bibinfo  {journal} {PRX Quantum}\ }\textbf {\bibinfo {volume} {6}},\ \bibinfo {pages} {010344} (\bibinfo {year} {2025})}\BibitemShut {NoStop}%
\bibitem [{\citenamefont {Ma}\ and\ \citenamefont {Turzillo}(2025)}]{Ma2025}%
  \BibitemOpen
  \bibfield  {author} {\bibinfo {author} {\bibfnamefont {R.}~\bibnamefont {Ma}}\ and\ \bibinfo {author} {\bibfnamefont {A.}~\bibnamefont {Turzillo}},\ }\bibfield  {title} {\bibinfo {title} {Symmetry-protected topological phases of mixed states in the doubled space},\ }\href {https://doi.org/10.1103/PRXQuantum.6.010348} {\bibfield  {journal} {\bibinfo  {journal} {PRX Quantum}\ }\textbf {\bibinfo {volume} {6}},\ \bibinfo {pages} {010348} (\bibinfo {year} {2025})}\BibitemShut {NoStop}%
\bibitem [{\citenamefont {Briegel}\ and\ \citenamefont {Raussendorf}(2001)}]{Cluster}%
  \BibitemOpen
  \bibfield  {author} {\bibinfo {author} {\bibfnamefont {H.~J.}\ \bibnamefont {Briegel}}\ and\ \bibinfo {author} {\bibfnamefont {R.}~\bibnamefont {Raussendorf}},\ }\bibfield  {title} {\bibinfo {title} {Persistent entanglement in arrays of interacting particles},\ }\href {https://doi.org/10.1103/PhysRevLett.86.910} {\bibfield  {journal} {\bibinfo  {journal} {Phys. Rev. Lett.}\ }\textbf {\bibinfo {volume} {86}},\ \bibinfo {pages} {910} (\bibinfo {year} {2001})}\BibitemShut {NoStop}%
\bibitem [{\citenamefont {White}(2009)}]{White}%
  \BibitemOpen
  \bibfield  {author} {\bibinfo {author} {\bibfnamefont {S.~R.}\ \bibnamefont {White}},\ }\bibfield  {title} {\bibinfo {title} {Minimally entangled typical quantum states at finite temperature},\ }\href {https://doi.org/10.1103/PhysRevLett.102.190601} {\bibfield  {journal} {\bibinfo  {journal} {Phys. Rev. Lett.}\ }\textbf {\bibinfo {volume} {102}},\ \bibinfo {pages} {190601} (\bibinfo {year} {2009})}\BibitemShut {NoStop}%
\bibitem [{\citenamefont {Stoudenmire}\ and\ \citenamefont {White}(2010)}]{Stoudenmire_2010}%
  \BibitemOpen
  \bibfield  {author} {\bibinfo {author} {\bibfnamefont {E.~M.}\ \bibnamefont {Stoudenmire}}\ and\ \bibinfo {author} {\bibfnamefont {S.~R.}\ \bibnamefont {White}},\ }\bibfield  {title} {\bibinfo {title} {Minimally entangled typical thermal state algorithms},\ }\href {https://doi.org/10.1088/1367-2630/12/5/055026} {\bibfield  {journal} {\bibinfo  {journal} {New Journal of Physics}\ }\textbf {\bibinfo {volume} {12}},\ \bibinfo {pages} {055026} (\bibinfo {year} {2010})}\BibitemShut {NoStop}%
\bibitem [{\citenamefont {White}(1992)}]{White1992}%
  \BibitemOpen
  \bibfield  {author} {\bibinfo {author} {\bibfnamefont {S.~R.}\ \bibnamefont {White}},\ }\bibfield  {title} {\bibinfo {title} {Density matrix formulation for quantum renormalization groups},\ }\href {https://doi.org/10.1103/PhysRevLett.69.2863} {\bibfield  {journal} {\bibinfo  {journal} {Phys. Rev. Lett.}\ }\textbf {\bibinfo {volume} {69}},\ \bibinfo {pages} {2863} (\bibinfo {year} {1992})}\BibitemShut {NoStop}%
\bibitem [{\citenamefont {\"Ostlund}\ and\ \citenamefont {Rommer}(1995)}]{Ostlund1995}%
  \BibitemOpen
  \bibfield  {author} {\bibinfo {author} {\bibfnamefont {S.}~\bibnamefont {\"Ostlund}}\ and\ \bibinfo {author} {\bibfnamefont {S.}~\bibnamefont {Rommer}},\ }\bibfield  {title} {\bibinfo {title} {Thermodynamic limit of density matrix renormalization},\ }\href {https://doi.org/10.1103/PhysRevLett.75.3537} {\bibfield  {journal} {\bibinfo  {journal} {Phys. Rev. Lett.}\ }\textbf {\bibinfo {volume} {75}},\ \bibinfo {pages} {3537} (\bibinfo {year} {1995})}\BibitemShut {NoStop}%
\bibitem [{\citenamefont {Perez-Garcia}\ \emph {et~al.}(2007)\citenamefont {Perez-Garcia}, \citenamefont {Verstraete}, \citenamefont {Wolf},\ and\ \citenamefont {Cirac}}]{Perez-Garcia}%
  \BibitemOpen
  \bibfield  {author} {\bibinfo {author} {\bibfnamefont {D.}~\bibnamefont {Perez-Garcia}}, \bibinfo {author} {\bibfnamefont {F.}~\bibnamefont {Verstraete}}, \bibinfo {author} {\bibfnamefont {M.~M.}\ \bibnamefont {Wolf}},\ and\ \bibinfo {author} {\bibfnamefont {J.~I.}\ \bibnamefont {Cirac}},\ }\bibfield  {title} {\bibinfo {title} {Matrix product state representations},\ }\href {https://doi.org/10.26421/QIC7.5-6-1} {\bibfield  {journal} {\bibinfo  {journal} {Quant. Inf. Comput.}\ }\textbf {\bibinfo {volume} {7}},\ \bibinfo {pages} {401} (\bibinfo {year} {2007})}\BibitemShut {NoStop}%
\bibitem [{\citenamefont {{Verstraete}}\ \emph {et~al.}(2008)\citenamefont {{Verstraete}}, \citenamefont {{Murg}},\ and\ \citenamefont {{Cirac}}}]{Verstraete2008}%
  \BibitemOpen
  \bibfield  {author} {\bibinfo {author} {\bibfnamefont {F.}~\bibnamefont {{Verstraete}}}, \bibinfo {author} {\bibfnamefont {V.}~\bibnamefont {{Murg}}},\ and\ \bibinfo {author} {\bibfnamefont {J.~I.}\ \bibnamefont {{Cirac}}},\ }\bibfield  {title} {\bibinfo {title} {Matrix product states, projected entangled pair states, and variational renormalization group methods for quantum spin systems},\ }\href {https://doi.org/10.1080/14789940801912366} {\bibfield  {journal} {\bibinfo  {journal} {Advances in Physics}\ }\textbf {\bibinfo {volume} {57}},\ \bibinfo {pages} {143} (\bibinfo {year} {2008})}\BibitemShut {NoStop}%
\bibitem [{\citenamefont {Schollwöck}(2011)}]{Schollwock2011}%
  \BibitemOpen
  \bibfield  {author} {\bibinfo {author} {\bibfnamefont {U.}~\bibnamefont {Schollwöck}},\ }\bibfield  {title} {\bibinfo {title} {The density-matrix renormalization group in the age of matrix product states},\ }\href {https://doi.org/https://doi.org/10.1016/j.aop.2010.09.012} {\bibfield  {journal} {\bibinfo  {journal} {Annals of Physics}\ }\textbf {\bibinfo {volume} {326}},\ \bibinfo {pages} {96} (\bibinfo {year} {2011})},\ \bibinfo {note} {january 2011 Special Issue}\BibitemShut {NoStop}%
\bibitem [{\citenamefont {Fishman}\ \emph {et~al.}(2022)\citenamefont {Fishman}, \citenamefont {White},\ and\ \citenamefont {Stoudenmire}}]{itensor}%
  \BibitemOpen
  \bibfield  {author} {\bibinfo {author} {\bibfnamefont {M.}~\bibnamefont {Fishman}}, \bibinfo {author} {\bibfnamefont {S.~R.}\ \bibnamefont {White}},\ and\ \bibinfo {author} {\bibfnamefont {E.~M.}\ \bibnamefont {Stoudenmire}},\ }\bibfield  {title} {\bibinfo {title} {{The ITensor Software Library for Tensor Network Calculations}},\ }\href {https://doi.org/10.21468/SciPostPhysCodeb.4} {\bibfield  {journal} {\bibinfo  {journal} {SciPost Phys. Codebases}\ ,\ \bibinfo {pages} {4}} (\bibinfo {year} {2022})}\BibitemShut {NoStop}%
\bibitem [{\citenamefont {Roberts}\ \emph {et~al.}(2017)\citenamefont {Roberts}, \citenamefont {Yoshida}, \citenamefont {Kubica},\ and\ \citenamefont {Bartlett}}]{Roberts}%
  \BibitemOpen
  \bibfield  {author} {\bibinfo {author} {\bibfnamefont {S.}~\bibnamefont {Roberts}}, \bibinfo {author} {\bibfnamefont {B.}~\bibnamefont {Yoshida}}, \bibinfo {author} {\bibfnamefont {A.}~\bibnamefont {Kubica}},\ and\ \bibinfo {author} {\bibfnamefont {S.~D.}\ \bibnamefont {Bartlett}},\ }\bibfield  {title} {\bibinfo {title} {Symmetry-protected topological order at nonzero temperature},\ }\href {https://doi.org/10.1103/PhysRevA.96.022306} {\bibfield  {journal} {\bibinfo  {journal} {Phys. Rev. A}\ }\textbf {\bibinfo {volume} {96}},\ \bibinfo {pages} {022306} (\bibinfo {year} {2017})}\BibitemShut {NoStop}%
\bibitem [{\citenamefont {Chen}\ \emph {et~al.}(2010)\citenamefont {Chen}, \citenamefont {Gu},\ and\ \citenamefont {Wen}}]{Chen}%
  \BibitemOpen
  \bibfield  {author} {\bibinfo {author} {\bibfnamefont {X.}~\bibnamefont {Chen}}, \bibinfo {author} {\bibfnamefont {Z.-C.}\ \bibnamefont {Gu}},\ and\ \bibinfo {author} {\bibfnamefont {X.-G.}\ \bibnamefont {Wen}},\ }\bibfield  {title} {\bibinfo {title} {Local unitary transformation, long-range quantum entanglement, wave function renormalization, and topological order},\ }\href {https://doi.org/10.1103/PhysRevB.82.155138} {\bibfield  {journal} {\bibinfo  {journal} {Phys. Rev. B}\ }\textbf {\bibinfo {volume} {82}},\ \bibinfo {pages} {155138} (\bibinfo {year} {2010})}\BibitemShut {NoStop}%
\bibitem [{\citenamefont {Dennis}\ \emph {et~al.}(2002)\citenamefont {Dennis}, \citenamefont {Kitaev}, \citenamefont {Landahl},\ and\ \citenamefont {Preskill}}]{toric}%
  \BibitemOpen
  \bibfield  {author} {\bibinfo {author} {\bibfnamefont {E.}~\bibnamefont {Dennis}}, \bibinfo {author} {\bibfnamefont {A.}~\bibnamefont {Kitaev}}, \bibinfo {author} {\bibfnamefont {A.}~\bibnamefont {Landahl}},\ and\ \bibinfo {author} {\bibfnamefont {J.}~\bibnamefont {Preskill}},\ }\bibfield  {title} {\bibinfo {title} {Topological quantum memory},\ }\href {https://doi.org/10.1063/1.1499754} {\bibfield  {journal} {\bibinfo  {journal} {Journal of Mathematical Physics}\ }\textbf {\bibinfo {volume} {43}},\ \bibinfo {pages} {4452} (\bibinfo {year} {2002})}\BibitemShut {NoStop}%
\bibitem [{\citenamefont {Kitaev}(2003)}]{Kitaev}%
  \BibitemOpen
  \bibfield  {author} {\bibinfo {author} {\bibfnamefont {A.}~\bibnamefont {Kitaev}},\ }\bibfield  {title} {\bibinfo {title} {Fault-tolerant quantum computation by anyons},\ }\href {https://doi.org/https://doi.org/10.1016/S0003-4916(02)00018-0} {\bibfield  {journal} {\bibinfo  {journal} {Annals of Physics}\ }\textbf {\bibinfo {volume} {303}},\ \bibinfo {pages} {2} (\bibinfo {year} {2003})}\BibitemShut {NoStop}%
\bibitem [{\citenamefont {Raussendorf}(2013)}]{Raussendorf}%
  \BibitemOpen
  \bibfield  {author} {\bibinfo {author} {\bibfnamefont {R.}~\bibnamefont {Raussendorf}},\ }\bibfield  {title} {\bibinfo {title} {Contextuality in measurement-based quantum computation},\ }\href {https://doi.org/10.1103/PhysRevA.88.022322} {\bibfield  {journal} {\bibinfo  {journal} {Phys. Rev. A}\ }\textbf {\bibinfo {volume} {88}},\ \bibinfo {pages} {022322} (\bibinfo {year} {2013})}\BibitemShut {NoStop}%
\bibitem [{\citenamefont {Hastings}(2007)}]{Hastings_2007}%
  \BibitemOpen
  \bibfield  {author} {\bibinfo {author} {\bibfnamefont {M.~B.}\ \bibnamefont {Hastings}},\ }\bibfield  {title} {\bibinfo {title} {An area law for one-dimensional quantum systems},\ }\href {https://doi.org/10.1088/1742-5468/2007/08/P08024} {\bibfield  {journal} {\bibinfo  {journal} {Journal of Statistical Mechanics: Theory and Experiment}\ }\textbf {\bibinfo {volume} {2007}},\ \bibinfo {pages} {P08024} (\bibinfo {year} {2007})}\BibitemShut {NoStop}%
\bibitem [{\citenamefont {Hung}\ and\ \citenamefont {Wen}(2014)}]{twistphase}%
  \BibitemOpen
  \bibfield  {author} {\bibinfo {author} {\bibfnamefont {L.-Y.}\ \bibnamefont {Hung}}\ and\ \bibinfo {author} {\bibfnamefont {X.-G.}\ \bibnamefont {Wen}},\ }\bibfield  {title} {\bibinfo {title} {Universal symmetry-protected topological invariants for symmetry-protected topological states},\ }\href {https://doi.org/10.1103/PhysRevB.89.075121} {\bibfield  {journal} {\bibinfo  {journal} {Phys. Rev. B}\ }\textbf {\bibinfo {volume} {89}},\ \bibinfo {pages} {075121} (\bibinfo {year} {2014})}\BibitemShut {NoStop}%
\bibitem [{\citenamefont {Verstraete}\ \emph {et~al.}(2005)\citenamefont {Verstraete}, \citenamefont {Cirac}, \citenamefont {Latorre}, \citenamefont {Rico},\ and\ \citenamefont {Wolf}}]{QuantumRG}%
  \BibitemOpen
  \bibfield  {author} {\bibinfo {author} {\bibfnamefont {F.}~\bibnamefont {Verstraete}}, \bibinfo {author} {\bibfnamefont {J.~I.}\ \bibnamefont {Cirac}}, \bibinfo {author} {\bibfnamefont {J.~I.}\ \bibnamefont {Latorre}}, \bibinfo {author} {\bibfnamefont {E.}~\bibnamefont {Rico}},\ and\ \bibinfo {author} {\bibfnamefont {M.~M.}\ \bibnamefont {Wolf}},\ }\bibfield  {title} {\bibinfo {title} {Renormalization-group transformations on quantum states},\ }\href {https://doi.org/10.1103/PhysRevLett.94.140601} {\bibfield  {journal} {\bibinfo  {journal} {Phys. Rev. Lett.}\ }\textbf {\bibinfo {volume} {94}},\ \bibinfo {pages} {140601} (\bibinfo {year} {2005})}\BibitemShut {NoStop}%
\bibitem [{\citenamefont {Greenberger}\ \emph {et~al.}(1990)\citenamefont {Greenberger}, \citenamefont {Horne}, \citenamefont {Shimony},\ and\ \citenamefont {Zeilinger}}]{Greenberger1990}%
  \BibitemOpen
  \bibfield  {author} {\bibinfo {author} {\bibfnamefont {D.~M.}\ \bibnamefont {Greenberger}}, \bibinfo {author} {\bibfnamefont {M.~A.}\ \bibnamefont {Horne}}, \bibinfo {author} {\bibfnamefont {A.}~\bibnamefont {Shimony}},\ and\ \bibinfo {author} {\bibfnamefont {A.}~\bibnamefont {Zeilinger}},\ }\bibfield  {title} {\bibinfo {title} {Bell's theorem without inequalities},\ }\href {https://doi.org/10.1119/1.16243} {\bibfield  {journal} {\bibinfo  {journal} {American Journal of Physics}\ }\textbf {\bibinfo {volume} {58}},\ \bibinfo {pages} {1131} (\bibinfo {year} {1990})}\BibitemShut {NoStop}%
\bibitem [{\citenamefont {Daniel}\ \emph {et~al.}(2022)\citenamefont {Daniel}, \citenamefont {Zhu}, \citenamefont {Alderete}, \citenamefont {Buchemmavari}, \citenamefont {Green}, \citenamefont {Nguyen}, \citenamefont {Thurtell}, \citenamefont {Zhao}, \citenamefont {Linke},\ and\ \citenamefont {Miyake}}]{AustinExperiment}%
  \BibitemOpen
  \bibfield  {author} {\bibinfo {author} {\bibfnamefont {A.~K.}\ \bibnamefont {Daniel}}, \bibinfo {author} {\bibfnamefont {Y.}~\bibnamefont {Zhu}}, \bibinfo {author} {\bibfnamefont {C.~H.}\ \bibnamefont {Alderete}}, \bibinfo {author} {\bibfnamefont {V.}~\bibnamefont {Buchemmavari}}, \bibinfo {author} {\bibfnamefont {A.~M.}\ \bibnamefont {Green}}, \bibinfo {author} {\bibfnamefont {N.~H.}\ \bibnamefont {Nguyen}}, \bibinfo {author} {\bibfnamefont {T.~G.}\ \bibnamefont {Thurtell}}, \bibinfo {author} {\bibfnamefont {A.}~\bibnamefont {Zhao}}, \bibinfo {author} {\bibfnamefont {N.~M.}\ \bibnamefont {Linke}},\ and\ \bibinfo {author} {\bibfnamefont {A.}~\bibnamefont {Miyake}},\ }\bibfield  {title} {\bibinfo {title} {Quantum computational advantage attested by nonlocal games with the cyclic cluster state},\ }\href {https://doi.org/10.1103/PhysRevResearch.4.033068} {\bibfield  {journal} {\bibinfo  {journal} {Phys. Rev. Res.}\ }\textbf {\bibinfo {volume} {4}},\ \bibinfo {pages} {033068} (\bibinfo {year}
  {2022})}\BibitemShut {NoStop}%
\bibitem [{\citenamefont {Raussendorf}\ \emph {et~al.}(2023)\citenamefont {Raussendorf}, \citenamefont {Yang},\ and\ \citenamefont {Adhikary}}]{Raussendorf2023measurementbased}%
  \BibitemOpen
  \bibfield  {author} {\bibinfo {author} {\bibfnamefont {R.}~\bibnamefont {Raussendorf}}, \bibinfo {author} {\bibfnamefont {W.}~\bibnamefont {Yang}},\ and\ \bibinfo {author} {\bibfnamefont {A.}~\bibnamefont {Adhikary}},\ }\bibfield  {title} {\bibinfo {title} {Measurement-based quantum computation in finite one-dimensional systems: string order implies computational power},\ }\href {https://doi.org/10.22331/q-2023-12-28-1215} {\bibfield  {journal} {\bibinfo  {journal} {{Quantum}}\ }\textbf {\bibinfo {volume} {7}},\ \bibinfo {pages} {1215} (\bibinfo {year} {2023})}\BibitemShut {NoStop}%
\bibitem [{\citenamefont {Brand{\~a}o}\ and\ \citenamefont {Kastoryano}(2019)}]{Brandao2019}%
  \BibitemOpen
  \bibfield  {author} {\bibinfo {author} {\bibfnamefont {F.~G. S.~L.}\ \bibnamefont {Brand{\~a}o}}\ and\ \bibinfo {author} {\bibfnamefont {M.~J.}\ \bibnamefont {Kastoryano}},\ }\bibfield  {title} {\bibinfo {title} {Finite correlation length implies efficient preparation of quantum thermal states},\ }\href {https://doi.org/10.1007/s00220-018-3150-8} {\bibfield  {journal} {\bibinfo  {journal} {Communications in Mathematical Physics}\ }\textbf {\bibinfo {volume} {365}},\ \bibinfo {pages} {1} (\bibinfo {year} {2019})}\BibitemShut {NoStop}%
\bibitem [{\citenamefont {Sang}\ \emph {et~al.}(2024)\citenamefont {Sang}, \citenamefont {Zou},\ and\ \citenamefont {Hsieh}}]{Sang2024}%
  \BibitemOpen
  \bibfield  {author} {\bibinfo {author} {\bibfnamefont {S.}~\bibnamefont {Sang}}, \bibinfo {author} {\bibfnamefont {Y.}~\bibnamefont {Zou}},\ and\ \bibinfo {author} {\bibfnamefont {T.~H.}\ \bibnamefont {Hsieh}},\ }\bibfield  {title} {\bibinfo {title} {Mixed-state quantum phases: Renormalization and quantum error correction},\ }\href {https://doi.org/10.1103/PhysRevX.14.031044} {\bibfield  {journal} {\bibinfo  {journal} {Phys. Rev. X}\ }\textbf {\bibinfo {volume} {14}},\ \bibinfo {pages} {031044} (\bibinfo {year} {2024})}\BibitemShut {NoStop}%
\bibitem [{\citenamefont {Lake}\ \emph {et~al.}(2025)\citenamefont {Lake}, \citenamefont {Balasubramanian},\ and\ \citenamefont {Choi}}]{Lake2025}%
  \BibitemOpen
  \bibfield  {author} {\bibinfo {author} {\bibfnamefont {E.}~\bibnamefont {Lake}}, \bibinfo {author} {\bibfnamefont {S.}~\bibnamefont {Balasubramanian}},\ and\ \bibinfo {author} {\bibfnamefont {S.}~\bibnamefont {Choi}},\ }\bibfield  {title} {\bibinfo {title} {Exact quantum algorithms for quantum phase recognition: Renormalization group and error correction},\ }\href {https://doi.org/10.1103/PRXQuantum.6.010350} {\bibfield  {journal} {\bibinfo  {journal} {PRX Quantum}\ }\textbf {\bibinfo {volume} {6}},\ \bibinfo {pages} {010350} (\bibinfo {year} {2025})}\BibitemShut {NoStop}%
\bibitem [{\citenamefont {Lieb}\ \emph {et~al.}(1961)\citenamefont {Lieb}, \citenamefont {Schultz},\ and\ \citenamefont {Mattis}}]{LIEB1961407}%
  \BibitemOpen
  \bibfield  {author} {\bibinfo {author} {\bibfnamefont {E.}~\bibnamefont {Lieb}}, \bibinfo {author} {\bibfnamefont {T.}~\bibnamefont {Schultz}},\ and\ \bibinfo {author} {\bibfnamefont {D.}~\bibnamefont {Mattis}},\ }\bibfield  {title} {\bibinfo {title} {Two soluble models of an antiferromagnetic chain},\ }\href {https://doi.org/https://doi.org/10.1016/0003-4916(61)90115-4} {\bibfield  {journal} {\bibinfo  {journal} {Annals of Physics}\ }\textbf {\bibinfo {volume} {16}},\ \bibinfo {pages} {407} (\bibinfo {year} {1961})}\BibitemShut {NoStop}%
\bibitem [{\citenamefont {Lieu}\ \emph {et~al.}(2020)\citenamefont {Lieu}, \citenamefont {Belyansky}, \citenamefont {Young}, \citenamefont {Lundgren}, \citenamefont {Albert},\ and\ \citenamefont {Gorshkov}}]{Lieu2020}%
  \BibitemOpen
  \bibfield  {author} {\bibinfo {author} {\bibfnamefont {S.}~\bibnamefont {Lieu}}, \bibinfo {author} {\bibfnamefont {R.}~\bibnamefont {Belyansky}}, \bibinfo {author} {\bibfnamefont {J.~T.}\ \bibnamefont {Young}}, \bibinfo {author} {\bibfnamefont {R.}~\bibnamefont {Lundgren}}, \bibinfo {author} {\bibfnamefont {V.~V.}\ \bibnamefont {Albert}},\ and\ \bibinfo {author} {\bibfnamefont {A.~V.}\ \bibnamefont {Gorshkov}},\ }\bibfield  {title} {\bibinfo {title} {Symmetry breaking and error correction in open quantum systems},\ }\href {https://doi.org/10.1103/PhysRevLett.125.240405} {\bibfield  {journal} {\bibinfo  {journal} {Phys. Rev. Lett.}\ }\textbf {\bibinfo {volume} {125}},\ \bibinfo {pages} {240405} (\bibinfo {year} {2020})}\BibitemShut {NoStop}%
\bibitem [{\citenamefont {de~Groot}\ \emph {et~al.}(2022)\citenamefont {de~Groot}, \citenamefont {Turzillo},\ and\ \citenamefont {Schuch}}]{deGroot2022}%
  \BibitemOpen
  \bibfield  {author} {\bibinfo {author} {\bibfnamefont {C.}~\bibnamefont {de~Groot}}, \bibinfo {author} {\bibfnamefont {A.}~\bibnamefont {Turzillo}},\ and\ \bibinfo {author} {\bibfnamefont {N.}~\bibnamefont {Schuch}},\ }\bibfield  {title} {\bibinfo {title} {Symmetry {P}rotected {T}opological {O}rder in {O}pen {Q}uantum {S}ystems},\ }\href {https://doi.org/10.22331/q-2022-11-10-856} {\bibfield  {journal} {\bibinfo  {journal} {{Quantum}}\ }\textbf {\bibinfo {volume} {6}},\ \bibinfo {pages} {856} (\bibinfo {year} {2022})}\BibitemShut {NoStop}%
\bibitem [{\citenamefont {O’Sullivan}\ \emph {et~al.}(2025)\citenamefont {O’Sullivan}, \citenamefont {Reuer}, \citenamefont {Grigorev}, \citenamefont {Dai}, \citenamefont {Hernández-Antón}, \citenamefont {Muñoz-Arias}, \citenamefont {Hellings}, \citenamefont {Flasby}, \citenamefont {Colao~Zanuz}, \citenamefont {Besse}, \citenamefont {Blais}, \citenamefont {Malz}, \citenamefont {Eichler},\ and\ \citenamefont {Wallraff}}]{OSullivan2025}%
  \BibitemOpen
  \bibfield  {author} {\bibinfo {author} {\bibfnamefont {J.}~\bibnamefont {O’Sullivan}}, \bibinfo {author} {\bibfnamefont {K.}~\bibnamefont {Reuer}}, \bibinfo {author} {\bibfnamefont {A.}~\bibnamefont {Grigorev}}, \bibinfo {author} {\bibfnamefont {X.}~\bibnamefont {Dai}}, \bibinfo {author} {\bibfnamefont {A.}~\bibnamefont {Hernández-Antón}}, \bibinfo {author} {\bibfnamefont {M.~H.}\ \bibnamefont {Muñoz-Arias}}, \bibinfo {author} {\bibfnamefont {C.}~\bibnamefont {Hellings}}, \bibinfo {author} {\bibfnamefont {A.}~\bibnamefont {Flasby}}, \bibinfo {author} {\bibfnamefont {D.}~\bibnamefont {Colao~Zanuz}}, \bibinfo {author} {\bibfnamefont {J.-C.}\ \bibnamefont {Besse}}, \bibinfo {author} {\bibfnamefont {A.}~\bibnamefont {Blais}}, \bibinfo {author} {\bibfnamefont {D.}~\bibnamefont {Malz}}, \bibinfo {author} {\bibfnamefont {C.}~\bibnamefont {Eichler}},\ and\ \bibinfo {author} {\bibfnamefont {A.}~\bibnamefont {Wallraff}},\ }\bibfield  {title} {\bibinfo {title} {Deterministic generation of two-dimensional
  multi-photon cluster states},\ }\href {https://doi.org/10.1038/s41467-025-60472-3} {\bibfield  {journal} {\bibinfo  {journal} {Nature Communications}\ }\textbf {\bibinfo {volume} {16}},\ \bibinfo {pages} {5505} (\bibinfo {year} {2025})}\BibitemShut {NoStop}%
\bibitem [{\citenamefont {Bluvstein}\ \emph {et~al.}(2022)\citenamefont {Bluvstein}, \citenamefont {Levine}, \citenamefont {Semeghini}, \citenamefont {Wang}, \citenamefont {Ebadi}, \citenamefont {Kalinowski}, \citenamefont {Keesling}, \citenamefont {Maskara}, \citenamefont {Pichler}, \citenamefont {Greiner}, \citenamefont {Vuleti{\'c}},\ and\ \citenamefont {Lukin}}]{Harvard}%
  \BibitemOpen
  \bibfield  {author} {\bibinfo {author} {\bibfnamefont {D.}~\bibnamefont {Bluvstein}}, \bibinfo {author} {\bibfnamefont {H.}~\bibnamefont {Levine}}, \bibinfo {author} {\bibfnamefont {G.}~\bibnamefont {Semeghini}}, \bibinfo {author} {\bibfnamefont {T.~T.}\ \bibnamefont {Wang}}, \bibinfo {author} {\bibfnamefont {S.}~\bibnamefont {Ebadi}}, \bibinfo {author} {\bibfnamefont {M.}~\bibnamefont {Kalinowski}}, \bibinfo {author} {\bibfnamefont {A.}~\bibnamefont {Keesling}}, \bibinfo {author} {\bibfnamefont {N.}~\bibnamefont {Maskara}}, \bibinfo {author} {\bibfnamefont {H.}~\bibnamefont {Pichler}}, \bibinfo {author} {\bibfnamefont {M.}~\bibnamefont {Greiner}}, \bibinfo {author} {\bibfnamefont {V.}~\bibnamefont {Vuleti{\'c}}},\ and\ \bibinfo {author} {\bibfnamefont {M.~D.}\ \bibnamefont {Lukin}},\ }\bibfield  {title} {\bibinfo {title} {A quantum processor based on coherent transport of entangled atom arrays},\ }\href {https://doi.org/10.1038/s41586-022-04592-6} {\bibfield  {journal} {\bibinfo  {journal} {Nature}\
  }\textbf {\bibinfo {volume} {604}},\ \bibinfo {pages} {451} (\bibinfo {year} {2022})}\BibitemShut {NoStop}%
\bibitem [{\citenamefont {Qin}\ and\ \citenamefont {Scarola}(2025)}]{VTech}%
  \BibitemOpen
  \bibfield  {author} {\bibinfo {author} {\bibfnamefont {Z.}~\bibnamefont {Qin}}\ and\ \bibinfo {author} {\bibfnamefont {V.~W.}\ \bibnamefont {Scarola}},\ }\bibfield  {title} {\bibinfo {title} {Scaling of computational order parameters in rydberg-atom graph states},\ }\href {https://doi.org/10.1103/PhysRevA.111.042617} {\bibfield  {journal} {\bibinfo  {journal} {Phys. Rev. A}\ }\textbf {\bibinfo {volume} {111}},\ \bibinfo {pages} {042617} (\bibinfo {year} {2025})}\BibitemShut {NoStop}%
\bibitem [{\citenamefont {Cao}\ \emph {et~al.}(2023)\citenamefont {Cao} \emph {et~al.}}]{Cao2023}%
  \BibitemOpen
  \bibfield  {author} {\bibinfo {author} {\bibfnamefont {S.}~\bibnamefont {Cao}} \emph {et~al.},\ }\bibfield  {title} {\bibinfo {title} {Generation of genuine entanglement up to 51 superconducting qubits},\ }\href {https://doi.org/10.1038/s41586-023-06195-1} {\bibfield  {journal} {\bibinfo  {journal} {Nature}\ }\textbf {\bibinfo {volume} {619}},\ \bibinfo {pages} {738} (\bibinfo {year} {2023})}\BibitemShut {NoStop}%
\bibitem [{\citenamefont {Jiang}\ \emph {et~al.}(2025)\citenamefont {Jiang} \emph {et~al.}}]{jiang2025generation95qubitgenuineentanglement}%
  \BibitemOpen
  \bibfield  {author} {\bibinfo {author} {\bibfnamefont {T.}~\bibnamefont {Jiang}} \emph {et~al.},\ }\href@noop {} {\bibinfo {title} {Generation of 95-qubit genuine entanglement and verification of symmetry-protected topological phases}} (\bibinfo {year} {2025}),\ \Eprint {https://arxiv.org/abs/2505.01978} {arXiv:2505.01978 [quant-ph]} \BibitemShut {NoStop}%
\bibitem [{\citenamefont {Flammia}\ and\ \citenamefont {Liu}(2011)}]{FlammiaLiu}%
  \BibitemOpen
  \bibfield  {author} {\bibinfo {author} {\bibfnamefont {S.~T.}\ \bibnamefont {Flammia}}\ and\ \bibinfo {author} {\bibfnamefont {Y.-K.}\ \bibnamefont {Liu}},\ }\bibfield  {title} {\bibinfo {title} {Direct fidelity estimation from few pauli measurements},\ }\href {https://doi.org/10.1103/PhysRevLett.106.230501} {\bibfield  {journal} {\bibinfo  {journal} {Phys. Rev. Lett.}\ }\textbf {\bibinfo {volume} {106}},\ \bibinfo {pages} {230501} (\bibinfo {year} {2011})}\BibitemShut {NoStop}%
\bibitem [{\citenamefont {Kumar}\ \emph {et~al.}(2025)\citenamefont {Kumar} \emph {et~al.}}]{Google}%
  \BibitemOpen
  \bibfield  {author} {\bibinfo {author} {\bibfnamefont {S.}~\bibnamefont {Kumar}} \emph {et~al.},\ }\href@noop {} {\bibinfo {title} {Quantum-classical separation in bounded-resource tasks arising from measurement contextuality}} (\bibinfo {year} {2025}),\ \Eprint {https://arxiv.org/abs/2512.02284} {arXiv:2512.02284 [quant-ph]} \BibitemShut {NoStop}%
\bibitem [{\citenamefont {Affleck}\ \emph {et~al.}(1987)\citenamefont {Affleck}, \citenamefont {Kennedy}, \citenamefont {Lieb},\ and\ \citenamefont {Tasaki}}]{AKLT}%
  \BibitemOpen
  \bibfield  {author} {\bibinfo {author} {\bibfnamefont {I.}~\bibnamefont {Affleck}}, \bibinfo {author} {\bibfnamefont {T.}~\bibnamefont {Kennedy}}, \bibinfo {author} {\bibfnamefont {E.~H.}\ \bibnamefont {Lieb}},\ and\ \bibinfo {author} {\bibfnamefont {H.}~\bibnamefont {Tasaki}},\ }\bibfield  {title} {\bibinfo {title} {Rigorous results on valence-bond ground states in antiferromagnets},\ }\href {https://doi.org/10.1103/PhysRevLett.59.799} {\bibfield  {journal} {\bibinfo  {journal} {Phys. Rev. Lett.}\ }\textbf {\bibinfo {volume} {59}},\ \bibinfo {pages} {799} (\bibinfo {year} {1987})}\BibitemShut {NoStop}%
\bibitem [{\citenamefont {Levin}\ \emph {et~al.}(2006)\citenamefont {Levin}, \citenamefont {Peres},\ and\ \citenamefont {Wilmer}}]{LevinPeresWilmer2006}%
  \BibitemOpen
  \bibfield  {author} {\bibinfo {author} {\bibfnamefont {D.~A.}\ \bibnamefont {Levin}}, \bibinfo {author} {\bibfnamefont {Y.}~\bibnamefont {Peres}},\ and\ \bibinfo {author} {\bibfnamefont {E.~L.}\ \bibnamefont {Wilmer}},\ }\href {http://scholar.google.com/scholar.bib?q=info:3wf9IU94tyMJ:scholar.google.com/&output=citation&hl=en&as_sdt=2000&ct=citation&cd=0} {\emph {\bibinfo {title} {{Markov chains and mixing times}}}}\ (\bibinfo  {publisher} {American Mathematical Society},\ \bibinfo {year} {2006})\BibitemShut {NoStop}%
\bibitem [{\citenamefont {Yang}\ and\ \citenamefont {White}(2020)}]{Krylov}%
  \BibitemOpen
  \bibfield  {author} {\bibinfo {author} {\bibfnamefont {M.}~\bibnamefont {Yang}}\ and\ \bibinfo {author} {\bibfnamefont {S.~R.}\ \bibnamefont {White}},\ }\bibfield  {title} {\bibinfo {title} {Time-dependent variational principle with ancillary krylov subspace},\ }\href {https://doi.org/10.1103/PhysRevB.102.094315} {\bibfield  {journal} {\bibinfo  {journal} {Phys. Rev. B}\ }\textbf {\bibinfo {volume} {102}},\ \bibinfo {pages} {094315} (\bibinfo {year} {2020})}\BibitemShut {NoStop}%
\bibitem [{\citenamefont {Wietek}\ \emph {et~al.}(2021)\citenamefont {Wietek}, \citenamefont {He}, \citenamefont {White}, \citenamefont {Georges},\ and\ \citenamefont {Stoudenmire}}]{Wietek}%
  \BibitemOpen
  \bibfield  {author} {\bibinfo {author} {\bibfnamefont {A.}~\bibnamefont {Wietek}}, \bibinfo {author} {\bibfnamefont {Y.-Y.}\ \bibnamefont {He}}, \bibinfo {author} {\bibfnamefont {S.~R.}\ \bibnamefont {White}}, \bibinfo {author} {\bibfnamefont {A.}~\bibnamefont {Georges}},\ and\ \bibinfo {author} {\bibfnamefont {E.~M.}\ \bibnamefont {Stoudenmire}},\ }\bibfield  {title} {\bibinfo {title} {Stripes, antiferromagnetism, and the pseudogap in the doped hubbard model at finite temperature},\ }\href {https://doi.org/10.1103/PhysRevX.11.031007} {\bibfield  {journal} {\bibinfo  {journal} {Phys. Rev. X}\ }\textbf {\bibinfo {volume} {11}},\ \bibinfo {pages} {031007} (\bibinfo {year} {2021})}\BibitemShut {NoStop}%
\end{thebibliography}%

\appendix

\section{Minimally entangled typical thermal states (METTS) algorithm} \label{sec:METTS}

We review a numerical method to efficiently approximate expectation values of Gibbs thermal states. The Gibbs thermal state is defined as
\begin{align}
    \rho=\frac{1}{\mathcal{Z}}e^{-\beta H},\hspace{5mm}\mathcal{Z}=\tr(e^{-\beta H}),
\end{align}
where $\beta=(k_BT)^{-1}$. As direct calculation is intractable for large system sizes, we instead use matrix product states (MPS) \cite{Perez-Garcia} via ITensor \cite{itensor} to efficiently simulate the states. To obtain expectation values of the symmetry operators (Fig. \ref{fig:SOP}) in the Gibbs state, we use a Monte Carlo Markov chain approach known as the minimally entangled typical thermal states (METTS) algorithm \cite{Stoudenmire_2010,White}, which approximates expectation values by sampling over ``typical thermal" states distributed over an appropriate probability distribution. This is efficiently performed within the MPS framework.

The METTS algorithm relies on a rearrangement of the expectation value of a desired operator $A$ in the Gibbs state:
\begin{align}
\ev{A}_\beta=\frac{1}{\mathcal{Z}}\tr(Ae^{-\beta H})&=\frac{1}{\mathcal{Z}}\sum_j\bra{j}Ae^{-\beta H}\ket{j}\nonumber\\
&=\sum_j\frac{p_j}{\mathcal{Z}}\bra{\phi_j}A\ket{\phi_j},\label{eq:METTS}
\end{align}
where $p_j=\bra{j}e^{-\beta H}\ket{j}$ and $ \ket{\phi_j}=p_j^{-1/2}e^{-\beta H/2}\ket{j}$. In the second line we used the cyclic trace property to symmetrically distribute the exponential around $A$. The final equation implies that by sampling the states $\ket{\phi_j}$ from the distribution $p_j/\mathcal{Z}$, one may approximate the expectation value in the Gibbs state.

An efficient algorithm was constructed in Ref. \cite{Stoudenmire_2010} to sample the states $\ket{\phi_j}$ from this probability distribution:
\begin{enumerate}
\item Pick a random product state $\ket{j}$.
\item Calculate $\ket{\phi_j}$ through imaginary time evolution and record the expectation value $\bra{\phi_j}A\ket{\phi_j}$.
\item Collapse to a new product state $\ket{k}$ in a chosen basis with probability $\abs{\bra{k}\ket{\phi_j}}^2$ and return to step 2.
\end{enumerate}
In the rest of this section, we refer to a cycle of steps 2 and 3 as a ``METTS iteration''.
By using the transition probabilities in step 3
\begin{align}
    T(j\xrightarrow{}k)=\abs{\braket{k}{\phi_j}}^2,\label{eq:transitionprobs}
\end{align}
it is guaranteed that the Markov chain satisfies the detailed-balance equation
\begin{align}
p_jT(j\xrightarrow{}k)=p_kT(k\xrightarrow{}j),\label{eq:detailedbalance}
\end{align}
which is known to be a sufficient condition for the fixed point of the chain to be $P(j)=p_j/\mathcal{Z}$ \cite{LevinPeresWilmer2006}. Following the steps outlined above and accumulating an ensemble of METTS states $\{\ket{\phi_j}\}_{j=1}^{N_I}$, then the ensemble average of the expectation value of A
\begin{align}
    \frac{1}{N_I}\sum_{j=1}^{N_I}\bra{\phi_j}A\ket{\phi_j},\label{eq:METTSev}
\end{align}
will converge to $\ev{A}_\beta$ by the law of large numbers. This approach does not suffer from the sign problem which compromises other quantum Monte-Carlo methods; the weights $p_j$ are simply the norm of the (unnormalized) state $e^{-\beta H/2}\ket{j}$ and are therefore always nonnegative.

\subsection{Imaginary time evolution}\label{sec:TDVP}

Creating the states $\ket{\phi_j}$ by applying the imaginary time evolution $e^{-\beta H/2}\ket{j}$ is the most expensive part of the algorithm, however, it can be done efficiently using MPS. We implement this using ITensor's Time-Dependent Variational Principle (TDVP) algorithm which time evolves the input state within an MPS manifold of controlled bond dimension. A ``two-site'' variant of the function lets the bond dimension temporarily increase before truncating to a specified maximum, allowing the algorithm to better capture entanglement growth. 

We also use a Krylov expansion method \cite{Krylov} to more accurately apply the time-evolution operator during each TDVP step. The Krylov method constructs a subspace consisting of the span of powers of $H$ applied to the state $\ket{j}$ up to some power $\delta$, then exactly solves the projection of this Hamiltonian in this low-dimensional subspace. This improves the fidelity of the evolution within the MPS manifold and allows for larger time steps without sacrificing accuracy. In our numerics we perform a Krylov expansion with $\delta=1$ between each TDVP step. The TDVP algorithm performs imaginary time evolution in steps of $\Delta T$, so that $\beta/2$ is broken into $\beta/2\Delta T$ discrete steps. We find that a TDVP time step size of
\begin{align}
    \Delta T=0.5\label{eq:TDVPstep}
\end{align}
is sufficient to capture the relevant properties of the time-evolved state. A larger step size resulted in both inaccurate expectation values when compared against exact solutions available within certain parts of the phase diagram and increased variance of the $N_I$ individual METTS expectation values. The specifics of these exact solutions are discussed in Appendix \ref{app:exactsolns}, and comparisons against the METTS estimates and exact solutions are discussed more in Section \ref{sec:results}.

Additionally, the TDVP function takes in a maximum bond dimension $m$ and cutoff $\varepsilon$ which determine the accuracy at which the state is captured. Surprisingly, we find that a low bond dimension $m=16$ is sufficient to capture accurate thermal behavior. We choose a cutoff of $\varepsilon=10^{-10}$.

\subsection{Error analysis}\label{sec:erroranalysis}

\begin{figure}[t]
\centering

\begin{tikzpicture}[font=\small]

% -----------------------
% Layout parameters
% -----------------------
\def\W{8cm}          % width of each panel
\def\Ysep{0.6cm}     % vertical separation between panels

% ======================
% Panel (a)
% ======================
\node[inner sep=0pt, anchor=north] (A) at (0,0)
{
    \includegraphics[width=\W, keepaspectratio]{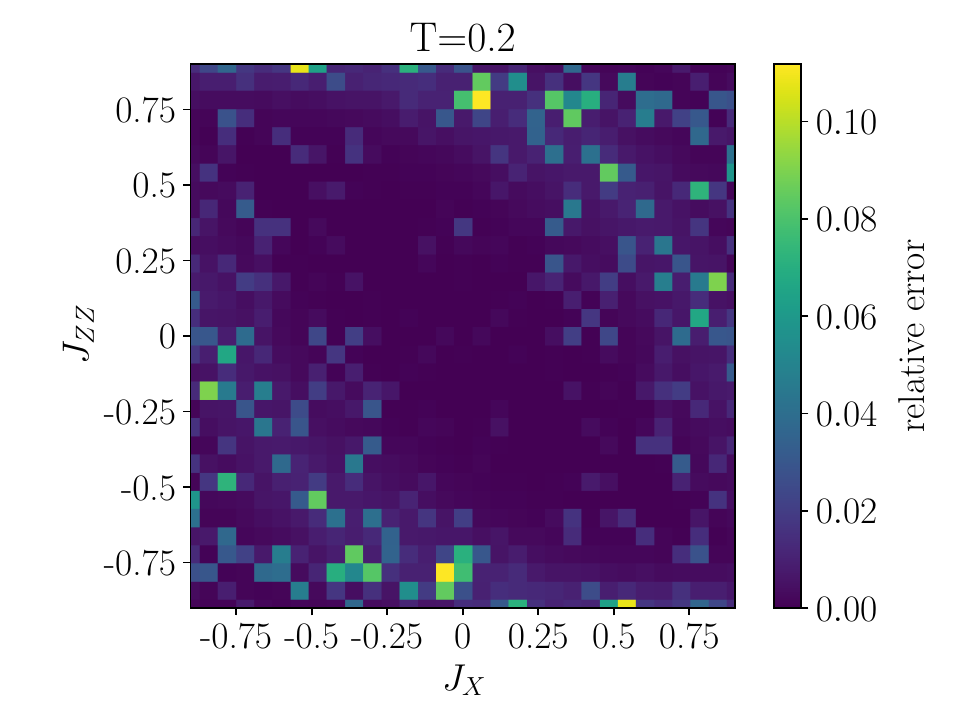}
};
\node[anchor=north] at (A.south) {(a)};

% ======================
% Panel (b)
% ======================
\node[inner sep=0pt, anchor=north] (B)
    at ($(A.south) - (0,\Ysep)$)
{
    \includegraphics[width=\W, keepaspectratio]{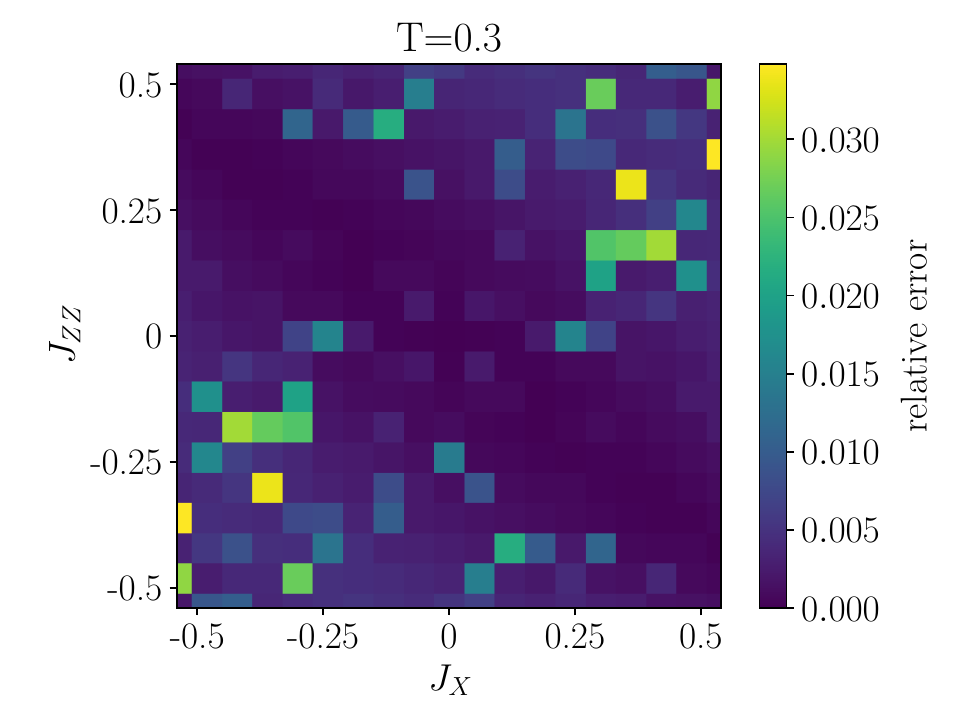}
};
\node[anchor=north] at (B.south) {(b)};

\end{tikzpicture}

\caption{ Relative error of the winning probability $P_{64}^{(Q,min)}$ estimate over METTS iterations is low across the phase diagram.}
    \label{fig:relativeerror}
\end{figure}

From the METTS estimate of the expectation value $\ev{A}_\beta$ given by Eq. \eqref{eq:METTSev} (which we will now denote as just $\ev{A}$) we estimate the standard error across the $N_I$ contributing expectation values $A_j=\bra{\phi_j}A\ket{\phi_j}$ as
\begin{align}
    s=\sqrt{\frac{\tau}{N_I}}\sigma,\label{eq:stderr}
\end{align}
where $\sigma^2$ is the sample variance
\begin{align}
    \sigma^2=\frac{1}{N_I-1}\sum_j(A_j-\ev{A})^2,\label{eq:samplevar}
\end{align}
and $\tau$ is the autocorrelation time
\begin{align}
    \tau=1+2\sum_{t=1}^\infty C(t),\quad C(t)=\frac{\ev{A_jA_{j+t}}-\ev{A}^2}{\ev{A^2}-\ev{A}^2},\label{eq:autocorrelation}
\end{align}
where $\ev{A_jA_{j+t}}$ is averaged over all $j$ in the Markov chain.
The autocorrelation time is included as the transition probabilities in Eq. \eqref{eq:transitionprobs} guarantee that consecutive samples are not independent. In practice one sets a finite cutoff $M$ for the sum by choosing the value of $t$ for which $C(t)$ first reaches 0.
The expectation values above are approximated by averaging over METTS samples. For example:
\begin{align}
\ev{A}&\approx\frac{1}{N_I}\sum_{k=1}^{N_I}A_k,\\
\ev{A_jA_{j+t}}&\approx\frac{1}{N_I-t}\sum_{k=1}^{N_I-t}A_kA_{k+t}.\label{eq:sampleevs}
\end{align}
One is usually interested in the relative error given by $s/\ev{A}$. We choose $N_I=110$ for all METTS calculations and start recording expectation values after a 10 iteration warm-up, which is sufficiently high to give a small relative error across the Hamiltonians of interest (Fig. \ref{fig:relativeerror}).

\subsection{Choice of collapse basis}\label{sec:collapse}
The choice of collapse basis $\{\ket{k}\}$ in step 3 of the algorithm is a critical step: a good choice can result in low autocorrelation between METTS iterations, leading to fast convergence. On the other hand, a poor choice of collapse basis will result in a complete failure of the algorithm.

Several basis choices are common in the METTS literature. The simplest approach is to collapse into a single fixed basis for each METTS iteration. However, this often results in large autocorrelation times across METTS iterations due to the large overlap between iteratively adjacent METTS states. Two common remedies to this problem are to either randomly draw each single-qubit collapse basis from a uniform distribution over all states on the Bloch sphere or to collapse into alternating orthogonal bases on successive METTS iterations. 

Interestingly, in our numerics we find contradictory behavior to the usual; collapsing into the $Z$ basis outperforms other choices of collapse basis throughout the SPT phase, including the alternating orthogonal basis and random collapse methods. However, a single basis collapse with other bases does not perform as well as the $Z$ basis. For example, collapsing in either the $X$ or $Y$ basis does not always predict the correct expectation values. For the $X$ basis, this is due to the fact that this is a symmetry eigenbasis of our Hamiltonian. Repeated steps of collapsing into this basis and evolving the state by $e^{-\beta H/2}$ will not allow the algorithm to explore different symmetry sectors. However, why the $Y$ collapse does not work is more mysterious. Similar behavior was seen in Ref. \cite{Wietek}, however, their model has a non-abelian and continuous $SU(2)$ symmetry, whereas ours has an abelian and discrete $\mathds{Z}_2\times\mathds{Z}_2$ symmetry.

\begin{figure}[t]
        \centering
        \includegraphics[width=\linewidth]{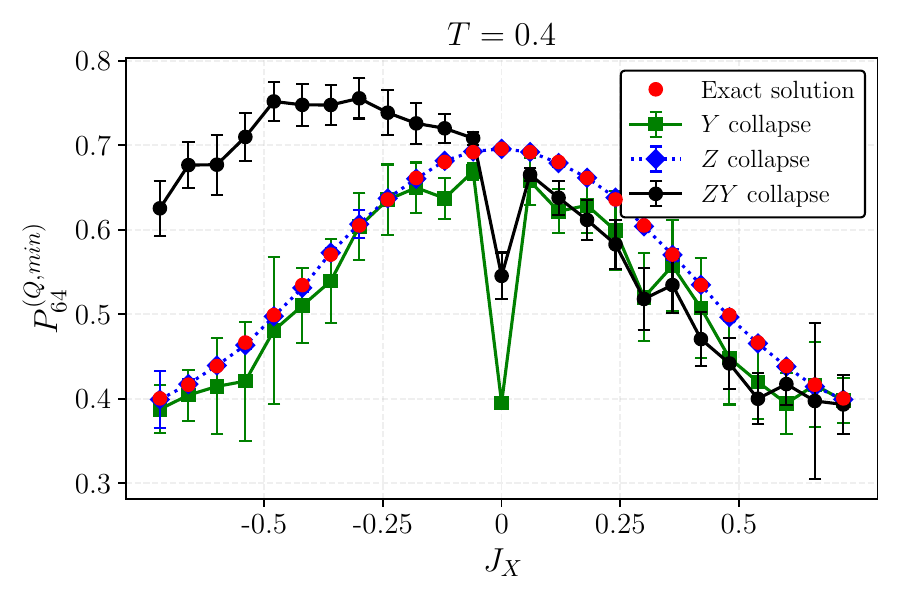}

    \caption{ Effect of collapse basis on the METTS estimate of the minimum winning probability $P_{64}^{(Q,min)}$. The $Z$-only collapse basis performs best in terms of both accuracy and relative error across METTS iterations when benchmarked against the exact solution of the $J_X$ axis of the Hamiltonian Eq. \eqref{eq:Hamiltonian}.}

    \label{fig:Z_ZX_rand_vs_exact}
\end{figure}

In Fig. \ref{fig:Z_ZX_rand_vs_exact} we analyze the accuracy and precision of collapse algorithms by benchmarking their results against the exact solutions of the $J_X$ and $J_{ZZ}$ axes at a fixed temperature $T=0.4$. In each case the METTS algorithm is run the same, with the only difference being the choice of collapse basis. The $Z$-only collapse algorithm vastly outperforms the other collapse choices.

\section{Exact solutions of $J_X$ and $J_{ZZ}$ axes}\label{app:exactsolns}

The Hamiltonian Eq. \eqref{eq:Hamiltonian} is exactly solvable along the $J_X$ and $J_{ZZ}$ axes by a Jordan-Wigner (JW) transformation into a fermionic quadratic form. Here we outline the general method to calculate exact expectation values of symmetry operators and twisted string order parameters in the corresponding Gibbs states, from which the triangle game winning probabilities can be deduced by Eq. \eqref{eq:pwin}. Our procedure follows the approach of Ref. \cite{LIEB1961407}. We will outline the approach for the case along the $J_X$ axis, i.e. the transverse-field cluster model. Expectation values along the $J_{ZZ}$ can also be obtained, however the analysis is more complicated as it relies on a global unitary transformation through $\tr(e^{-\beta H}\mathcal{O})=\tr(U_{CZ}e^{-\beta H}U_{CZ}^\dag U_{CZ}\mathcal{O}U_{CZ}^\dag)$ where $U_{CZ}=\prod_jCZ_{j,j+1}$. The transformed Hamiltonian is then a decoupled transverse-field Ising model on even and odd sublattices, which can be exactly solved by the same method that follows. However, we have found that these expectation values are exactly the same as those along the $J_X$ axis.

For the JW transformation it is necessary to work with the Hamiltonian in the form
\begin{align}
    H=-\sum_{j=1}^n\bigg(X_{j-1}Z_jX_{j+1}+J_XZ_j\bigg),\label{eq:JxAxis}
\end{align}
which is just a rotated version of the $J_X$ axis of our Hamiltonian. We define the (inverse) JW transformation which maps Paulis to fermionic operators as
\begin{align}
    Z_j &= \mathds{1}-2c_j^\dag c_j,\nonumber\\
    X_j &= \exp(i\pi \sum_{k=1}^{j-1}c_k^\dag c_k)(c_j^\dag + c_j). \label{eq:JordanWigner}
\end{align}
The Hamiltonian in Eq. \eqref{eq:JxAxis} then takes the form
\begin{align}
H=\sum_{j=1}^n\bigg(\big(-c_{j-1}^\dag c_{j+1}^\dag - c_{j-1}^\dag c_{j+1} + h.c.\big) + 2J_Xc_j^\dag c_j\bigg),\label{eq:JWtransformedH}
\end{align}
up to a constant energy shift $-nJ_X$. Note that we have ignored boundary effects which come about from the nonlocal JW mapping. Although this is inexact for small system sizes, it quickly converges to the exact solution as the system size increases, as the relative weight of these boundary terms decreases.

Eq. \eqref{eq:JWtransformedH} is now in a quadratic form, i.e. it can be written as
\begin{align}
    H=\sum_{i,j}\bigg(c_i^\dag A_{ij}c_j+\frac{1}{2}(c_i^\dag B_{ij}c_j^\dag - c_i B_{ij}^* c_j)\bigg),\label{eq:quadraticform}
\end{align}
for some $A,B\in\mathds{C}^{n\times n}$. Since all the coefficients in Eq. \eqref{eq:JWtransformedH} are real, $A,B\in R^{n\times n}$,  in which case the Hermiticity of $H$ requires $A^T=A$. Additionally, anticommutation relations require $B^T=-B$ with zeros on the diagonal. By rewriting $c_i^\dag c_j$ terms as $(c_i^\dag c_j + \delta_{ij}-c_j c_i^\dag)/2$ we can write
\begin{align}
    H=\frac{1}{2}\Gamma^\dag M\Gamma+\frac{\mathds{1}}{2}\tr(A),
\end{align}
where
\begin{subequations}
\begin{align}
    \Gamma^\dag&=
    \begin{pmatrix}
        c_1^\dag & c_2^\dag & \hdots & c_n^\dag & \vline & c_1 & c_2 & \hdots & c_n
    \end{pmatrix},
\\
    M&=
    \left(
    \begin{array}{c|c}
    A & B\\
    \hline
    -B & -A
    \end{array}
    \right).
    \label{eq:Mmatrix}
\end{align}
\end{subequations}
Fermionic Hamiltonians in quadratic form are exactly solvable, in that their energy spectrum, ground state wavefunctions, and expectation values have an analytic form. This is because there always exists a similarity transformation that turns the Hamiltonian into an equivalent free fermion model with the form
\begin{align}
H=\sum_{k=1}^n\Lambda_k\eta_k^\dag\eta_k,\label{eq:freefermion}
\end{align}
where $\eta_k^\dag$ and $\eta_k$ are operators which satisfy the fermionic anticommutation relations. In a free fermion model, all hopping and pairing terms have zero expectation value, and all expectation values of number operators are given by the Fermi-Dirac distribution $f(\Lambda_j)=(e^{\beta\Lambda_j}+1)^{-1}$, i.e.
\begin{subequations}\label{eq:freefermionevs}
\begin{align}
\ev{\eta_j^\dag\eta_k}&=\delta_{jk}f(\Lambda_j),\\
\ev{\eta_j^\dag\eta_k^\dag}&=0,\\
\ev{\eta_j\eta_k}&=0.
\end{align}
\end{subequations}
The similarity transformation writes fermionic creation and annihilation operators from linear combinations of free fermion operators $\eta,\eta^\dag$. Since $M$ in Eq. \eqref{eq:Mmatrix} is a real symmetric matrix, it can be diagonalized by an orthogonal matrix $O$, constraining the transformation to take the form \begin{subequations}
\begin{align}
c_j=\sum_{k=1}^n\bigg(U_{jk}\eta_k+V_{jk}\eta_k^\dag\bigg),\\
c_j^\dag=\sum_{k=1}^n\bigg(V_{jk}\eta_k+U_{jk}\eta_k^\dag\bigg).
\end{align}
\end{subequations}
In other words, $\Gamma=O\Delta$. Then $\Gamma^\dag M \Gamma = \Delta^\dag O^T M O \Delta=\Delta^\dag D \Delta$, where
\begin{subequations}
\begin{align}
    \Delta^\dag&=
    \begin{pmatrix}
        \eta_1^\dag & \eta_2^\dag & \hdots & \eta_n^\dag & \vline & \eta_1 & \eta_2 & \hdots & \eta_n
    \end{pmatrix},
\\
    O&=
    \left(
    \begin{array}{c|c}
    U & V\\
    \hline
    V & U
    \end{array}
    \right),
\end{align}
\end{subequations}
and the form of $M$ allows us to deduce that $D=\text{diag}(\Lambda_1,\Lambda_2,\hdots,\Lambda_n,-\Lambda_1,-\Lambda_2,\hdots,-\Lambda_n)$.
By defining the symmetric and antisymmetric combinations
\begin{subequations}
\begin{align}
    \phi &= U+V,\label{eq:phi}\\
    \psi &= U-V,\label{eq:psi}
\end{align}
\end{subequations}
one obtains a set of coupled eigenvalue equations
\begin{subequations}\label{eq:phipsieqns}
\begin{align}
(A+B)\ket{\phi_k}=\Lambda_k\ket{\psi_k},\label{eq:phipsieqnsa}\\
    (A-B)\ket{\psi_k}=\Lambda_k\ket{\phi_k},\label{eq:phipsieqnsb}
\end{align}
\end{subequations}
where $\ket{\phi_k}$ ($\ket{\psi_k}$) is the $k$'th column of $\phi$ ($\psi$). We use the translational symmetry of our model to obtain
\begin{align}
\ket{\phi_k}=\sum_{j=1}^n
\begin{cases}\sqrt{\frac{2}{n}}\cos(\frac{2\pi jk}{n})\ket{j}, & k\in\{1,2,...,\frac{n}{2}-1\}\\
\frac{1}{\sqrt{n}}(-1)^j\ket{j}, & k=\frac{n}{2}\\
\sqrt{\frac{2}{n}}\sin(\frac{2\pi jk}{n})\ket{j}, & k\in\{\frac{n}{2}+1,...,n-1\}\\
\frac{1}{\sqrt{n}}\ket{j}, & k=n
\end{cases}\label{eq:phik}
\end{align}
and
\begin{align}
    \Lambda_k=2\sqrt{1+J_X^2-2J_X\cos(\frac{4\pi k}{n})}.\label{eq:lambdak}
\end{align}
Equation \eqref{eq:phipsieqnsa} can then be used to find the corresponding $\ket{\psi_k}$. Since the set $\{\ket{\phi_k}, \ket{\psi_k}\}_{k=1}^n$ completely describe the transformation matrices $U$ and $V$, the similarity transformation is exactly solved.

To calculate expectation values, one considers the operators
\begin{subequations}
\begin{align}
    P_j&=c_j^\dag+c_j,\\
    Q_j&=c_j^\dag-c_j,
\end{align}
\end{subequations}
which satisfy the anticommutation relations $\{P_j,Q_k\}=0$ $\forall j,k$ and $\{P_j,P_k\}=\{Q_j,Q_k\}=0$  $\forall j\ne k$. They are expressed in terms of the free fermion operators $\eta, \eta^\dag$ as
\begin{subequations}
\begin{align}
    P_j=\sum_k\phi_{jk}(\eta_k^\dag+\eta_k),\\
    Q_j=\sum_k\psi_{jk}(\eta_k^\dag-\eta_k).
\end{align}
\end{subequations}
Then using Eq. \eqref{eq:freefermionevs},
\begin{subequations}
\begin{align}
\ev{P_iP_j}&=\sum_{k}\phi_{ik}\phi_{jk}=\delta_{ij},\\
\ev{Q_iQ_j}&=-\sum_{k}\psi_{ik}\psi_{jk}=-\delta_{ij},\\
\ev{Q_iP_j}&=-\sum_k\tanh(\frac{\beta\Lambda_k}{2})\psi_{ik}\phi_{jk}\equiv G_{ij}.\label{eq:G}
\end{align}
\end{subequations}
Using these results, all expectation values of twisted SOPs and symmetry operators can be evaluated. For example, the symmetry operator $U(z)=\prod_jX_j$ corresponding to the choice of $g=z=(1,1)$ in Eq. \eqref{eq:fullsymmetry} becomes $\prod_jZ_j$ in the rotated model, which then becomes $\prod_j(\mathds{1}-2c_j^\dag c_j)=\prod_jP_jQ_j=\prod_jQ_jP_j$ after a JW transformation (the last step results in no negative because $n$ is even). Using Wick's theorem, the expectation value of $U(z)$ can then be written as the determinant of G defined in Eq. \eqref{eq:G}

The twisted SOP expectation values can similarly be found by determinants of particular submatrices of G. For example, the cluster state twisted SOP chosen with $g=z=(1,1)$ and $h=x=(0,1)$ in \eqref{eq:twistedSOP} can be written as $T_{[p,q]}^{(z,x)}=-U(x)S_{[p,q]}(z)$. Defining $j=2p$ and $k=2p-1$ this can be written
\begin{align}
   -\bigg(\prod_{\substack{u=1 \\ u\text{ odd}}}^{j-2}X_u\bigg)Z_{j-1}Z_j\bigg(\prod_{\substack{v=j+1 \\ v\text{ even}}}^{k-1}X_v\bigg)Y_kY_{k+1}\bigg(\prod_{\substack{w=k+2 \\ w\text{ odd}}}^nX_w\bigg),
\end{align}
or in the rotated model
\begin{align}
    -\bigg(\prod_{u=1}^{j-2}Z_u\bigg)X_{j-1}X_j\bigg(\prod_{v=j+1}^{k-1}Z_v\bigg)Y_kY_{k+1}\bigg(\prod_{w=k+2}^nZ_w\bigg).
\end{align}
After the JW transformation this becomes
\begin{equation}
\begin{aligned}
-\bigg(\prod_{u=1}^{j-2} P_u Q_u\bigg)
  Q_{j-1} P_j
  \bigg(\prod_{v=j+1}^{k-1} P_vQ_v\bigg) \\
\quad\cdot P_k Q_{k+1}
  \bigg(\prod_{w=k+2}^{n} P_w Q_w\bigg).
\end{aligned}
\end{equation}
We consider the case when $j$ is even and $k,j$ are evenly spaced from the boundaries so that $k=n-j+1$. In this case rearranging the operators to $QP$ ordering removes the negative sign. From Wick's theorem, the simplest contribution to $\ev{T_{[p,q]}^{(g,h)}}$ is then
\begin{align}
    \bigg(\prod_{u=1}^{j-2}G_{uu}\bigg)G_{j-1,j}\bigg(\prod_{v=j+1}^{k-1}G_{vv}\bigg)G_{k,k+1}\bigg(\prod_{w=k+2}^nG_{ww}\bigg).
    \label{eq:diagonal}
\end{align}
The full expectation value can be deduced by taking the determinant of the matrix which has Eq. \eqref{eq:diagonal} as the product of the diagonal elements, with all other matrix elements logically deduced from the diagonal. For example, for $n=16,j=6,k=11$  (in other words $p=3$, $q=6$) this is
\begin{widetext}
\begin{align}
\ev{T_{[3,6]}^{(z,x)}}=\det
\left(
\begin{array}{ccc|ccc|cc}
    G_{22} & G_{24} & G_{24} & G_{27} & G_{29} & G_{2,11} & G_{2,14} & G_{2,16} \\
    G_{42} & G_{44} & G_{46} & G_{47} & G_{49} & G_{4,11} & G_{4,14} & G_{4,16}\\
    \hline
    G_{52} & G_{54} & G_{56} & G_{57} & G_{59} & G_{5,11} & G_{5,14} & G_{5,16}\\
    G_{72} & G_{74} & G_{76} & G_{77} & G_{79} & G_{7,11} & G_{7,14} & G_{7,16}\\
    G_{92} & G_{94} & G_{96} & G_{97} & G_{99} & G_{9,11} & G_{9,14} & G_{9,16}\\
    \hline
    G_{12,2} & G_{12,4} & G_{12,6} & G_{12,7} & G_{12,9} & G_{12,11} & G_{12,14} & G_{12,16}\\
    G_{14,2} & G_{14,4} & G_{14,6} & G_{14,7} & G_{14,9} & G_{14,11} & G_{14,14} & G_{14,16}\\
    G_{16,2} & G_{16,4} & G_{16,6} & G_{16,7} & G_{16,9} & G_{16,11} & G_{16,14} & G_{16,16}
\end{array}
\right).
\label{eq:Gtwisted}
\end{align}
\end{widetext}
The lines in Eq. \eqref{eq:Gtwisted} show submatrices that have a consistent structure. The corner blocks are elements of the $G$-submatrix for $U(x)$, whereas the center block is an element of the corresponding $G$-submatrix for $U(y)$. The blocks connecting these two regions show the ``twist'' between different symmetry operators.
In general this will be a $n/2 \times n/2$-dimensional matrix with block boundaries between columns $j/2,j/2+1$ and $(k+1)/2,(k+3)/2$ and rows $j/2-1,j/2$ and $(k-1)/2,(k+1)/2$.
\end{document}